\newif\iflong\longtrue
\newif\iflongF\longFtrue 
\newif\ifcutLICS\cutLICStrue
\newif\ifcutLMCS\cutLMCSfalse

\newif\ifcutCBNLMCS\cutCBNLMCSfalse  

\newif\ifnotes\notesfalse

\newif\ifcommentreview\commentreviewfalse

\documentclass{lmcs}
\pdfoutput=1
\usepackage[utf8]{inputenc}

\usepackage{lastpage}
\lmcsdoi{22}{4}{2}
\lmcsheading{}{\pageref{LastPage}}{}{}%
{Apr.~28,~2025}{Oct.~05,~2026}{}

\usepackage{hyperref}
\usepackage{cite}
\usepackage{xspace}
\usepackage{bm}
\usepackage{amssymb,amsmath,stmaryrd,mathtools}
\DeclareMathAlphabet{\mathbbmsl}{U}{bbm}{m}{sl}
\usepackage{amsthm}
\usepackage{xcolor}
\usepackage{mathpartir}
\input{davidemacros.sty}
\input{guilhemmacros.sty}

\theoremstyle{plain}
\newtheorem{theorem}[thm]{Theorem}
\newtheorem{proposition}[thm]{Proposition}
\newtheorem{lemma}[thm]{Lemma}
\newtheorem{corollary}[thm]{Corollary}

\theoremstyle{definition}
\newtheorem{definition}[thm]{Definition}

\newtheorem{remark}[thm]{Remark}

\ifnotes
\newcommand{\DS}[1]{[ DS: \emph{#1}]}
\else
\newcommand{\DS}[1]{}
\fi

\ifcommentreview
\newcommand{\reviewComment}[1]{\textbf{Review}: {\color{red} #1}}
\newcommand{\reviewCommentDS}[1]{\textbf{Review[DS]}: {\color{green} #1}}
\newcommand{\reviewCommentGJ}[1]{\textbf{Review[GJ]}: {\color{blue} #1}}
\else
\newcommand{\reviewComment}[1]{}
\newcommand{\reviewCommentDS}[1]{}
\newcommand{\reviewCommentGJ}[1]{}
\fi

\newcommand{\cutLMCSsecondRound}[1]{}

\def\doublewedge{\mathrel{\wedge\!\!\!\wedge}}

\usepackage{booktabs}   
\usepackage{subcaption} 

\begin{document}

\title[Games, mobile processes, and functions]{Games, mobile processes, and functions -- alternating, concurrent, and
well-bracketed semantics}

\author[G.~Jaber]{Guilhème Jaber\lmcsorcid{0009-0006-7149-5073}}[a]
\author[D.~Sangiorgi]{Davide Sangiorgi\lmcsorcid{0000-0001-5823-3235}}[b]


\address{Nantes Université, École Centrale Nantes, CNRS, INRIA, LS2N, UMR 6004, France}


\address{Università di Bologna, Italy  and INRIA,   France}


\begin{abstract}
 We establish a tight connection between two models of the $\lambda$-calculus, 
 namely Milner's encoding into the $\pi$-calculus (precisely, the Internal $\pi$-calculus), 
 and operational game semantics (OGS).
 We first investigate the operational correspondence between 
 the behaviours of the encoding provided by
 $\pi$ and OGS. 
 We do so for various LTSs:  the standard LTS for  
 $\pi$ and a new 
 `concurrent' LTS for OGS;  an `output-prioritised' LTS for  $\pi$ 
 and the standard alternating LTS for OGS.  We then show that the equivalences induced on
 $\lambda$-terms by all these LTSs (for $\pi$ and OGS) coincide. 
 We also prove that when equivalence is based on complete traces,  
 the `concurrent' and `alternating' variants of   OGS 
 also coincide with the `well-bracketed' variant. 

 These connections allow us to transfer results and techniques between $\pi$ and OGS. In
 particular: 
 we import up-to techniques from $\pi$ onto OGS; we derive
 congruence and compositionality results for OGS from those of $\pi$; 
 we transport the notion of complete traces from OGS onto $\pi$, obtaining a new  
 behavioural equivalence that yields a full abstraction result for the encoding of
 $\lambda$-terms with respect to contexts written in a $\lambda$-calculus extended with store. 
 The study is illustrated for both call-by-value and call-by-name.
\end{abstract}

\maketitle

\section{Introduction}
\label{s:intro}

The topic of the paper is the comparison between 
\emph{Operational Game semantics} (OGS) and the \emph{$\pi$-calculus}, as generic
models or frameworks for the semantics of higher-order languages.  

Game semantics~\cite{10.1006/inco.2000.2917,10.1006/inco.2000.2930} provides
intensional models of higher-order languages,
where the denotation of a program brings up its possible  interactions with the
surrounding context.
Distinct points of game semantics are the 
rich categorical structure and the emphasis on compositionality.
Game semantics provides a modular characterization of higher-order languages with
 computational effects like control operators~\cite{10.5555/788019.788859}, 
 mutable store~\cite{abramsky1997linearity,10.5555/788020.788891} 
 or concurrency~\cite{ghica2004angelic,10.1007/11944836_38}.
 This gives rise to the ``Semantic Cube''~\cite{10.5555/645726.667210},
 a characterization of the absence of such computational effects
 in terms of appropriate restrictions on the interactions, with conditions
 like \emph{alternation, well-bracketing, visibility} or \emph{innocence}.
For instance, well-bracketing corresponds to the absence of control operators
like $\callcc$.

Game semantics has spurred Operational Game Semantics
(OGS)~\cite{jagadeesan2009open,jeffrey2005java,laird2007,10.1016/j.entcs.2012.08.013,10.1145/2603088.2603150},
as a way to describe the interactions of a program with its 
environment by embedding  programs into appropriate configurations and then defining rules
that turn such configurations into an LTS.
Besides minor differences in the representation of causality between actions,
the main distinction with ``standard'' game semantics
is in the way in which the denotation of programs is obtained:
via an LTS, rather than, 
compositionally, by induction on the structure of the programs (or their types).
It is nonetheless possible to establish a formal correspondence between these two 
representations~\cite{10.1145/2603088.2603150}.

OGS is particularly effective on higher-order programs.
To avoid being too intensional, functional values exchanged between the program and its environment 
are represented as atoms, seen as free variables.  
Therefore OGS configurations include open terms. 
The basic actions in the LTS produced by OGS represent
the calls and returns of functions between
a program
and its environment.
The OGS semantics
has been shown  fully-abstract, that is, to characterize observational equivalence,
for a wide class of
programming languages, including effectful subsets of 
ML~\cite{laird2007,10.1145/2933575.2934509},
fragments of Java~\cite{jeffrey2005java}, aspect-oriented 
programs~\cite{jagadeesan2009open}. 
The conditions  in the above-mentioned Semantic Cube (alternation,
well-bracketing, etc.)
 equally apply to OGS.

In this paper, we consider forms of OGS
for the  pure untyped call-by-value  $\lambda$-calculus,
which enforce some of such conditions.
Specifically, we consider:
\begin{itemize}
 \item an \emph{Alternating} OGS, where only one term can be run at a time, and 
the control of the interactions alternates between the  term and the environment;
 \item a \emph{Concurrent} OGS, where multiple terms can be run in parallel;
\item a \emph{Well-Bracketed} OGS, where the calling convention between a term and its
environment follows a stack discipline.
\dsOLD{mention also complete traces?}
\end{itemize}

\vskip .2cm 
\noindent
The $\pi$-calculus is the paradigmatical name-passing calculus, that is, a calculus where
names (a synonym for `channels') may be passed around.  In the literature about the
$\pi$-calculus, and more generally in  Programming Language theory,
Milner's work on functions as processes~\cite{Mil92s}, which
shows how the evaluation strategies of {\em call-by-name $\lambda$-calculus} and {\em
  call-by-value $\lambda$-calculus}~\cite{Abr88,Plo75} can be
faithfully mimicked, is generally considered a landmark.  
The work promotes the $\pi$-calculus to be a model for
 higher-order programs, and  
provides the means to study $\lambda$-terms in contexts other than
the purely sequential ones and with the instruments available to reason about 
processes.  
In the paper, $\pi$-calculus is actually meant to be 
the  \emph{Internal $\pi$-calculus} (\piI), 
a subset of the original $\pi$-calculus 
in which only fresh names may be exchanged among processes \cite{San95i}.
The use of \piI 
avoids  a few shortcomings
of Milner's encodings, notably for call-by-value; e.g., 
the  failure of the $\beta_v$ rule (i.e., the encodings of $ (\lambda x.M ) V$ and $M
\sub V x$ may be behaviourally distinguishable in  $ \pi$).
%

Further investigations into Milner's  encodings \cite{San93d,San95lazy,DurierHS18} have revealed 
what is the equivalence induced on 
 $\lambda$-terms by the encodings, whereby
two $\lambda$-terms are equal if their  encodings are behaviourally
equivalent (i.e., bisimilar) $\piI$ terms. 
In call-by-value, this equivalence is \emph{eager normal-form bisimilarity}
\cite{lassentrees}, 
a tree structure
proposed  by Lassen 
(and indeed sometimes referred to as `Lassen's trees')
as the call-by-value  counterpart 
   of
B{\"o}hm Trees (or L{\'e}vy-Longo Trees). 

In a nutshell, when  used to give semantics to a language, 
major strengths of the $\pi$-calculus are its algebraic structure and the related algebraic
properties and proof techniques;
major strengths of     OGS   are its proximity to the source
language~--- the configurations of OGS are built directly from the terms of the source
language,  as opposed to an encoding as in the $\pi$-calculus~--- 
and its flexibility~--- the semantics can be tuned to account for specific features of the
source language like control operators or references. 

The general goal of this paper is to show that there is a tight and precise correspondence
between OGS and $\pi$-calculus as models of programming languages, and that such a
correspondence may be profitably used to take advantage of the strengths of the two models. 
We carry out the above program in the specific case of (untyped) call-by-value
$\lambda$-calculus,  $\LasV$, 
which is richer and (as partly suggested above) with some more subtle
aspects than call-by-name.
We also sketch at the end of the paper a similar correspondence for the
untyped call-by-name $\lambda$-calculus.
Analogies and similarities between game semantics and  $\pi$-calculus have  been
pointed  out in various papers in the literature (e.g., \cite{10.1145/224164.224189,berger2001sequentiality}; see
Section~\ref{s:cf}), and used to, e.g., explain game semantics using $\pi$-like processes,
and enhance type systems for $\pi$-terms.  In this paper, in contrast, we carry out a direct comparison
between the two models, on their interpretation of functions.

We take the (arguably) canonical representations of $\LasV$  into \piI 
and OGS.
The latter representation is Milner's encoding,
rewritten in \piI.  We consider two variant  behaviours for the \piI terms, respectively
produced by  the ordinary LTS
of \piI, and by  an `output-prioritised' LTS, \opLTS, 
in which   input actions
may be observed only in the absence of outputs and internal actions. 
Intuitively, the \opLTS is intended to respect sequentiality  constraints in the
\piI terms:  
an  output action stands for an ongoing computation (for instance, returning  the result of
 a previous request) whereas 
an  input
 action starts a 
   new computation
 (for instance, a request of a certain service); therefore,  
in a sequential system, an output action should have priority over input actions.  
 For OGS, the $\LasV$ representation is the straightforward adaptation of the OGS
representations of typed $\lambda$-calculi in the literature, e.g., \cite{laird2007}
designed to build fully-abstract models for languages with mutable store.
As outlined above, a few variants exist:
we consider the  alternating, concurrent,
and well-bracketed variants.

We also consider the refinements obtained by focusing on  \emph{complete traces}, 
intuitively, traces that describe ``terminating'' interactions.
Complete traces have been introduced in~\cite{abramsky1997linearity} to obtain an
effective presentation of a fully abstract model for a language with higher-order store.

We then develop a  thorough comparison between the behaviours of the OGS and \piI
representations. 
For this, we define a mapping from OGS configurations to \piI processes. We also exploit
the fact that, syntactically, the actions in the  OGS and   \piI LTSs are the same.
We derive a tight correspondence between the two models, which allows us to transfer techniques
and  to 
switch  freely between the two models in the analysis of 
the OGS and \piI representations of $\LasV$, 
so to establish new results or obtain new proofs.
On these aspects, our main results  are the following: 

\begin{enumerate}
\item
We show that
the representation of $\LasV$ in the Alternating 
  OGS  is behaviourally the same as the representation in \piI assuming the \opLTS.
Thus the semantics on $\lambda$-terms induced by the OGS and \piI representations
coincide. 
 The same results are obtained between the Concurrent  OGS and \piI under  its ordinary
 LTS. 

\item 
We   transfer 
 `bisimulation up-to techniques' for  $\piI$, notably a form of `up-to context',
 onto  (Concurrent) OGS.
The result is a powerful technique, called `up-to composition', that allows us to split an OGS
configuration into more elementary configurations during the bisimulation game. 

\item
We show that the semantics induced on $\LasV$ by the Alternating 
and by the Concurrent OGS
are the same, both when the equality in OGS is  based  on traces and
when it is based on bisimulation. 
In other words, all the OGS views of $\LasV$
(Alternating or Concurrent, traced-based or bisimulation-based) coincide. 
Moreover,  we show 
that 
 such
induced semantics is the equality of 
Lassen's trees. We derive the result in two ways:
one  in which we  directly import it from \piI; the other
in which we lift eager normal-form bisimulations into  OGS bisimulations via the up-to-composition technique. 


\item We derive congruence and compositionality properties for the OGS semantics,
as well as a notion of tensor product over configurations that computes
interleavings of traces.
\dsOLD{maybe hint at the link between tensor and the compositionality things in the previous sentence?}

\item 
We show that when equality in OGS is based on complete traces,
the semantics induced on $\LasV$ is the same for the 
Alternating, Concurrent, and Well-Bracketed OGS.
\cutLMCSsecondRound{Therefore, drawing from the results about the Well-Bracketed OGS,
such a semantics also  coincides with that produced by contextual equivalence in $\LasV$ extended
with higher-order references.  }

\cutLMCSsecondRound{
\item We transport the idea of complete traces onto \piI; 
 the formalisation makes use of a standard linear-receptive type system for \piI.
 We thus obtain a new behavioural equivalence for \piI, 
 called \emph{\completetrace equivalence}.  
 The OGS full abstraction result under complete traces 
 can then be transported onto \piI,
 becoming full abstraction for Milner's encoding 
 with respect to contextual equivalence of $\LasV$ extended with higher-order references.

\dsOLD{probably remove sentence above}}
\end{enumerate}

\noindent
The results about OGS in (1-4) are obtained by exploiting
  the mapping  into \piI and  its algebraic  properties and proof techniques, as well as  
the up-to-composition technique for OGS imported from \piI.
Analogous results are established for the call-by-name strategy.
To summarize our results, we consider in this paper the following LTSs and equivalences:
\begin{center}
\begin{tabular}{@{}lccccc@{}}
  \toprule
  LTS & \piI & {\small output-prioritised} \piI & \COGS & \AOGS & \WBOGS \\
  \midrule
  bisimilarity & $\wbPI$ & $\wbOPI$ & $\wbC$ & $\wbA$ & $\wbWB$ \\
  trace equivalence & $\TEPI$ & $\TEN$ & $\TEp$ & $\TEA$ & $\TEwb$ \\
  complete-trace equivalence & & & $\TEpc$ & $\TEc$ & $\TEwbc$ \\
  \bottomrule
\end{tabular}
\end{center}
Terms of the call-by-value $\lambda$-calculus $M,N$ can then be embedded into the LTS $\piI$ and $\opLTS$ with the encoding
$\encoW M$ and into the $\COGS$, $\AOGS$ and $\WBOGS$ as initial configurations $\conf M$.
Using these embeddings, the various equivalence of the previous figure can be lifted into equivalence over terms.
We then prove that:
\begin{enumerate}
\item $\wbPI, \wbOPI, \wbC, \wbA, \wbWB, \TEPI, \TEN, \TEp, \TEA$ all coincides;
\item $\TEpc, \TEc$ and $\TEwbc$ all coincides.
\end{enumerate}
We also prove that the first group of equivalences also coincides with eager normal-form bisimilarity.

\emph{Structure of the paper.} 
Sections~\ref{s:nota} to~\ref{s:encoLpi} contain background material:
general  notations,
 \piI, $\LasV$,
the representations of $\LasV$  in the   Alternating OGS (\AOGS)
and  in \piI.
The following sections contain the new material.
In Section~\ref{s:encoGames} 
we study  the relationship  between the two $\LasV$ representations, 
 in   \piI using the output-prioritised LTS. 
In  Section~\ref{s:cogs}
we establish a similar relationship between 
 a  new  Concurrent OGS (\COGS) and  \piI using its ordinary LTS. We also transport 
up-to techniques onto OGS, and prove that all the semantics of $\LasV$ examined 
(OGS, \piI, traces, bisimulations)
 coincide. 
We import compositionality results for OGS from \piI in Section~\ref{s:compo}.
In Section~\ref{s:wb-ogs} we introduce the Well-Bracketed OGS (\WBOGS) and complete
traces. 
In Section~\ref{sec:tensor},
we introduce a notion of product of OGS configurations, that is used to
provide a wide range of decomposition results.
In  Section~\ref{s:wbc_vs_c} 
we relate all complete-trace equivalences in OGS (\AOGS, \COGS, \WBOGS). 
\cutLMCSsecondRound{In Section~\ref{s:fapiI}  we import   complete-trace equivalence and the full
abstraction result (for pure $\lambda$-terms)  into \piI.}
In Section~\ref{s:cbnNEW} we discuss the case of call-by-name.
Finally, in Section~\ref{s:cf} some further related work and future research directions. 
Appendixes~\ref{a:behav}-\ref{a:cbn} report additional  material for proofs and results in
Sections~\ref{s:encoGames}, \ref{s:cogs}, \ref{s:wb-ogs}, and~\ref{s:cbnNEW}. 

An extended abstract of this paper had appeared at CSL'22
\cite{JaberS22}, without proofs and without the contents in
Sections~\ref{s:wb-ogs}-\ref{s:cbnNEW}.

\section{Notations}
\label{s:nota}

In the paper, we use various LTSs and behavioural relations for them, 
both for OGS and for the $\pi$-calculus.  
In this section, we introduce or summarise common notations. 

We use a tilde, like in $\til a$, for  (possibly empty) tuples of objects (usually names). 
Let $ K \arrG \mu K'$ be a generic LTS (for OGS or \piI; the grammar for actions in the
LTSs  for OGS
and \piI will be the same). 
Actions, ranged over by $\mu$,  can be of the form $\inp a \tilb$, $\bout a \tilb$,
$\tau$,  and $(\tila)$, where $\tau$, called \emph{silent} or (\emph{invisible}) action, 
represents an internal step in $K$, that is, an action that does not require interaction
with the outside,  and $(\tila)$ is  a special action performed by abstractions in \piI
and initial configurations in OGS.   If $\mu \neq \tau$ then $\mu$ is a \emph{visible} action; we use
$\ell$ to range over them.

We sometimes abbreviate $\arrG\tau$ as $ \longrightarrowG $.
We  write $ \LongrightarrowG$ for the reflexive and transitive closure of 
$\arrG\tau$;
i.e., $K
\LongrightarrowG K'
 $ holds if $K $ can evolve into $K'$ by performing some
silent steps~-- possibly none.
We also write 
$K \ArrG \mu K' $ if   $K \LongrightarrowG \arrG\mu \LongrightarrowG K' $ (the
composition of the three 
relations); i.e.,  
$K \ArrG \mu K'$ holds if there are $K_1 $ and $K_2$ with 
$K \LongrightarrowG K_1$, $K_1 \arrG\mu K_2$ 
and $K_2 \LongrightarrowG K'$.
Then $ \ArcapG \mu  $  is  $\ArrG \mu $ if $\mu\neq \tau$, and  
$\LongrightarrowG $ if $\mu=\tau$. 

 \emph{Traces}, ranged over by $\trace$,  are finite (and possibly empty) sequences of
 visible actions.
If $\trace =\ell_1, \ldots, \ell_n$ ($n\geq 0$), then  $K \ArrG \trace K'$  holds if there
are $K_0, \ldots, K_n$ with $K_0 = K$, $K_n =K'$, and $K_i \ArrG{\ell_{i+1}} K_{i+1}$  for
$0\leq i < n$; and    $K \ArrG \trace $ if there is $K' $ with $K \ArrG \trace  K'$. 

In an action $\inp a\tilb$ or $\bout a\tilb$ we call $a$ the \emph{subject} and
$\tilb$ the object; similarly in a trace $s$ a name $a$ occurs in \emph{subject position}
(resp.\ in \emph{object position})
if there is an action in $s$ whose subject is $a$ (resp.\ whose object includes $a$). 

Two states $K_1, K_2 $ of the LTS   are \emph{trace equivalent}, written 
 $K_1 \TEG K_2$, if    ($K_1 \ArrG \trace $  iff 
$K_2 \ArrG \trace  $), for all $\trace$. 
 
Similarly,  \emph{bisimilarity}, written $\wbG $,
 is the largest symmetric relation on the state of the LTS such that 
whenever $K_1 \wbG K_2$ then 
\begin{itemize}
\item $K_1 \ArrG \mu  K_1'$ implies there is $K_2'$ with 
$K_2 \ArcapG \mu  K_2'$ and $K'_1 \wbG K'_2$.
\end{itemize} 
For instance, in the \piI LTS   $\arrPI \mu $ of Section~\ref{ss:piI}, 
$P \TEPI Q$ means that the \piI  processes $P$ and $Q$ are trace equivalent, and 
$P \wbPI Q$ means that they are bisimilar.

\begin{remark}[bound names]
\label{r:bn} 
In an action $\inp a \tilb$ or $\bout a \tilb$ or $(\tilb)$, 
name $a$ is \emph{free} whereas $\tilb$ are \emph{bound}; 
the free and bound names of a trace are defined accordingly.

Throughout the paper, 
in any statement (concerning OGS or \piI), the bound names of an action or of a trace that
appear in the statement are supposed to be all \emph{fresh}; i.e.,  
all distinct from each other and from the free names of the objects in the statement. 
For instance, in a  trace $s= s_1, \mu, s_2$, a  bound name of $\mu$ should not appear in the
preceding prefix  $s_1$
(neither bound nor free), should not have other occurrences  in $\mu$, 
and may only appear free in the following suffix  $s_2$; moreover, if  
$s$ is the trace 
in the definition above  of trace equivalence,  any  bound names  of
$s$ should not appear free in the initial states $K_1,K_2$.  
Similarly, in the following definition of bisimilarity, the bound names of the action $\mu$ do
not appear free in the states $K_1,K_2$. 
In the same way, the convention applies to statements about operational correspondence between an
OGS state $K_1$ and a \piI state $K_2$ (the bound names of actions or traces mentioned in
the statement do not appear free in  $K_1,K_2$).
\qed\end{remark}
\section{Background}
\label{s:back}

\subsection{The Internal $\pi$-calculus}
\label{ss:piI}

The Internal $\pi$-calculus, $\piI$, is, intuitively, 
a subset of the $\pi$-calculus 
in which all outputs are bound. This is syntactically enforced by
having outputs written as $\bout a\tilb$ (which in the $\pi$-calculus would be an
abbreviation for $\res\tilb \out a \tilb$). 
We use the polyadic version of the calculus, in which a tuple of names may be exchanged in
a communication step. 
All tuples of names in \piI 
are made of pairwise distinct components. 

\emph{Abstractions} are used to write name-parametrised processes,  for
instance, when writing recursive process definitions.
 The instantiation of the parameters of an
abstraction $B$ is done via the {\em application} construct $\app B \tila$.  
Processes and abstractions form the set of \emph{agents}, ranged over by $\AG$.
Lowercase letters
$a,b, \ldots, x,y, \ldots$ range over the infinite set of names.  The grammar of \piI
 is thus:
\[\begin{array}{ccll}
  P & \defeq & \nil    \midd    \inp a \tilb . P    \midd    \bout a \tilb . P 
    \midd    \res a P \midd   P_1 |  P_2   \midd  ! \inp a \tilb . P   
  \midd \app  B \tila & \mbox{(processes)} \\[\myptD]
  B & \defeq & \bind \tila P \midd {\rmmm K} & \mbox{(abstractions)}
\end{array}\]
We omit the standard definition of free names,
bound names, and names of an agent, respectively indicated with $\fn -$, $\bn -$, and
$\n -$. The grammar for processes includes the operators  
$\nil$ (representing  the inactive process), 
input and output prefixes, restriction, parallel composition, replication, and
application.   

In the grammar of abstractions, 
${\rmmm K}$ is a constant,  used to write recursive definitions. Each
constant ${\rmmm K}$ has a defining equation of the form 
${\rmmm K} \defeq  \bind{\tilx} P$, where $\bind{\tilx} P  $ is name-closed (that is, 
without free names); $\tilx$ 
 are the formal parameters of
the constant (replaced by the actual parameters whenever the constant
is used).
Replication could be avoided in the syntax since it can be encoded
with recursion. However its semantics is simple, and it is 
a useful construct for encodings, 
thus we chose to include it in the grammar.

\ifcutLMCS
Sometimes, we use $\defi$ as an abbreviation mechanism, to
assign a name to an expression to which we want to refer later.
\gj{Is this use of $\defi$ used anywhere?}
\fi

An {application redex}, that is, a term syntactically of the form 
 $\app{((\til x) P)}{\til a}$, can be normalised as $P \subst {\til x}{\til a}$. An agent is
\emph{normalised} if all such application redexes have been
contracted.
Although the full
application $\app  {\rmmm K}{\til a} $ is often convenient, when
it comes to reasoning on behaviours 
it is useful to assume that all expressions are normalised, in the above sense.
Thus in the remainder of the paper 
\emph{we identify an agent with its normalised expression}.
\reviewComment{does this include recursive definitions? In
which case, are normalized expressions infinite?}
\reviewCommentDS{Answer: no, only expressions of that syntactic form are
  modified. Recursive definition are only unfolded via rule {\trans{con}} pf Figure 1. This
kind of normalisation is a (finite) syntactic simplification; it is not mandatory, however
it sometimes simplifies  structural reasoning  on processes. We have rephrased a bit the sentence.}

Since the calculus is polyadic, 
we assume a \emph{sorting system}~\cite{Mil91}
   to avoid disagreements  in the arities of
the tuples of names 
carried by a given name 
and in applications of abstractions.
Being not essential, it will not be presented here.
The reader should
take for granted that all agents described obey a sorting. 


\subsubsection{Operational semantics and behavioural relations}
\label{sss:os_piI}

In Figure~\ref{t:pii_ST}
we report the standard LTS of \piI.
The symbol $\equival$ is used to indicate $\alpha$-equivalence
between \piI agents.
Rule  {\trans{abs}} is used to formalise, by means of an action, the ground instantiation
of the parameters of an abstraction. 
In both tables, we omit the symmetric of {\tt parL}, called {\tt parR}.
In rules {\tt com}, $\outC \mu$ is the complementary of $\mu$, 
thus if $\mu = \inp  a \tilx$
then  $\outC\mu = \bout  a \tilx$, and conversely if  $\mu = \bout  a \tilx$
then  $\outC\mu = \inp  a \tilx$.

\begin{figure*}[ht]
\begin{center}
\begin{tabular}{@{}r@{\,}l@{\hspace{1.2em}}r@{\,}l@{\hspace{1.2em}}r@{\,}l@{}}
{\trans{alpha}}: &  $\displaystyle{  P\equival P' \andalso   { P'} \arrPI\mu    { P''}
  \over 
  P \arrPI\mu {P''}  
 } $ &  
{\trans{pre}}: &  
$\displaystyle{
\over
{\mu.      P} \arrPI\mu {P}}$   
& {\trans{parL}}: &  $\displaystyle{   P \arrPI\mu { P'} \over  { P | Q}   \arrPI\mu
{P'| Q} } $   
\\[\mysp] 
{\trans{res}}: & $\displaystyle{{  P} \arrPI{\mu}{P'}
\andalso  x \not \in \n\mu
 \over
{ \res x     P}   \arrPI{ \mu}{\res x P'  }} $ 
&
\multicolumn{4}{c}{ 
    {\trans{com}}:  \; \;    $\displaystyle{ {P} \arrPI{ \mu }{P'} \andalso   
 Q
\arrPI{\outC\mu}{Q'}
\andalso \mu\neq \tau,   \; \tilx = \bn\mu 
  \over
{     P |  Q }\arrPI{ \tau} 
{
 \res{  \tilx} (P'
|  Q') } }$    
}
 \\[\mysp]
    {\trans{rep}}: &
 $\displaystyle{ {\mu.  P }\arrPI{\mu} {P'}
 \over
{ !\mu. P }   \arrPI{ \mu}{ P' | !\mu. P } } $  
&
\multicolumn{4}{c}{ 
    {\trans{con}}: \;\;
$  \displaystyle{   P \arrPI \mu {P'}
\andalso  \mbox{ } {\rmmm K}  \defi
  \abs{\til y} Q  
\andalso 
\mbox{ } \abs{\til y} Q \equival \abs{\til x} P
    \over 
{
 \app {\rmmm K}{\til x}}  \arrPI\mu 
{P'} }$
}
\\[\mysp] 
\multicolumn{6}{c}{ 
    {\trans{abs}}:  \; \;    
$  \displaystyle{ 
\mbox{( $B =  \abs{\til y} Q $   or 
($B = K $ and $K \defi   \abs{\til y} Q $) )}  
  \andalso  \abs{\til y} Q \equival \abs{\til x} P
  \over 
{B \arr{({\til x})}     P }
 }$
}
\end{tabular} 
\end{center} 
\caption{The standard LTS for \piI}
\label{t:pii_ST}
 \Mybar 
\end{figure*}

In the LTS for \piI,
transitions are of the form $\AG \arrPI \mu \AG '$, where the bound names of $\mu$ are {fresh},
i.e., they  do not appear free
in $\AG$.  The choice of bound names can be made more explicit by following the operational game
semantics approach in which transitions are defined on configurations that include a name
set component. 
(Some presentations of $\pi$-calculi indeed make use of configurations.)  
The two approaches are interchangeable, and the choice is largely a matter of taste. 
The standard LTS  is however more convenient in reasoning, e.g.,  the up-to
proof techniques of Section~\ref{ss:upto}.

All behavioural relations are extended to abstractions by requiring ground
instantiation of the parameters; this is expressed by means of a transition; e.g., the action 
$ \bind{\tilx} P \arrPI {\bind{\tilx}{}}  P$ (see rule  {\trans{abs}} in Figure~\ref{t:pii_ST}). 

\finish{see if we want to mention also structural congruence $\equiv$} 

\subsection{The expansion relation}

The expansion relation, $\expa$~\cite{SaWabook}, 
is an asymmetric variant of 
$\wbPI$  which allows us to  count the number of $\tau$-actions
performed by the processes. 
Thus, $P \expa Q$ holds if $P \wbPI
Q$ but also $Q $ has at least as many $\tau$-moves as $P$. 
We write $P \arcap\mu P'$  
when either $P \arr\mu P'$ or ($\mu =\tau$ and
$P=P'$).

\begin{definition}[expansion]
  A relation $\R \subseteq   \pr \times \pr $ 
  is an {\em expansion } if $P \RR Q$
  implies:
  \begin{enumerate}
  \item  Whenever 
  $P \arr\mu P'$,    
  there exists  $Q' $  s.t.\ $Q \ArrPI\mu Q'$ 
  and $P'  \RR  Q'$;
  \item whenever $Q \arr\mu Q'$, 
  there exists  $P' $  s.t.\ $P \arcap\mu P'$ 
  and $P'  \RR  Q'$.
  \end{enumerate}
  \noindent
  We say that $Q$  {\em expands } $P$, written
  $P \expa  Q$, 
  if $P  \R  Q$,  for some expansion $ \R $.   
\end{definition}

Relation $\expa$ is a precongruence, and  is
strictly included in $\wbPI$.  
\subsection{The ``up-to'' techniques}
\label{ss:upto}
The ``up-to'' techniques allow us to reduce
the size of a relation $\R$ to exhibit for proving bisimilarities.

\subsubsection{Bisimulation up-to context and expansion}

Our main up-to technique will be
{\em up-to context and expansion} \cite{San93c},
which admits the use of contexts and of behavioural equivalences such as 
expansion to achieve the
closure of a relation in the bisimulation game.

For our purposes, a restricted form of 'up-to context' will be enough,
in which only contexts of the form $\res \tilc (R   |  \contexthole)$ are
employed.

\begin{definition}[bisimulation up-to context and up-to $\contr$]
  \label{d:uptp_sce}
  A symmetric relation $  \R   \subseteq \pr \times \pr $ is a 
  {\em bisimulation up-to context and up-to $\contr$ }
  if $P \RR Q$ and 
  $P \arr\mu P''$
  then there are processes  $P'$ and $Q'$ 
  and a static context $\qct$
  s.t. $ P'' \contr \ct {P'}$, 
  $Q \Arcap\mu  \contr \ct{Q'}$ and 
  $P'  \RR  Q'$.
\end{definition}

The soundness of the above technique relies on the following lemma.
\begin{lemma}
\label{l:uptocon-expa}
Suppose $ \R $ is a bisimulation up-to context and up-to $\contr$,
 $(P, Q) \in  \R $ holds and  $\qct$ is a 
 static context. 
If  $\ct{P}\arr{\mu_1} \cdots\arr{\mu_n}  P_1$, $n\geq 0$,
 then
there are a static context $\qctp$ and processes $P'$ and $Q'$
s.t.\  $P_1 \contr \ctp{P'}$, $\ct Q \Arcap{\mu_1} \cdots
\Arcap{\mu_n} \contr \ctp{Q'}$ and $P'  \RR  Q'$. 
\end{lemma} 

\begin{theorem}
\label{t:up-to-con-expa}
If $ \R  $  is a 
bisimulation up-to context and up-to $\contr$, then $ \R  \subseteq \wbPI$.
\end{theorem} 
 
\begin{proof}
By showing that  
the relation $\S$ containing 
all pairs $(P,Q)$ such that 
\[ \begin{array}{l}
 \mbox{ for some static context $\qct$ and
processes $P', Q'$} \\
 \mbox{ it holds that } P \contr \ct{P'}, Q \contr \ct{Q'} \mbox{ and }
P'  \RR  Q' 
 \end{array}
 \]
is a bisimulation. 
See \cite{San93c} for details (adapting the proof to \Intp is
straightforward).
\end{proof}

\subsubsection{Bisimulation up-to bisimilarity}
\label{ss:bisimulation_upto_wb}

We will also employ {\em bisimulation up-to $\wbPI$}~\cite{Mil89}, 
a widely-used technique,
whereby bisimilarity itself is employed to achieve 
the closure of the candidate relation during the bisimulation game.

\begin{definition}
  [bisimulation up-to $\wbPI$]
  A symmetric relation $\R$ on $\pr  \times \pr$
  is a {\em bisimulation up-to $\wbPI$ } if $P  \RR  Q$ and $P
  \ArrPI\mu P'$ imply that  there exists $Q'$ s.t.
  $Q \Arcap\mu Q'$ and 
  $P' \wbPI  \RR \wbPI Q'$. 
\end{definition}

\begin{theorem}
\label{t:up-to}
  Suppose $\R$ is a {\em bisimulation up-to $\wbPI$}.
  Then $\R \subseteq \wbPI$.
\end{theorem}

\subsubsection{Bisimulation up-to context and up-to $(\contr,
\wbPI)$}

We will also employ
a variant of bisimulation up-to context and expansion, called 
{\em bisimulation up-to context and up-to $(\contr,\wbPI)$}.
In Definition~\ref{d:uptp_sce}, the occurrence of $\contr$  in 
$Q \Arcap\mu  \contr  \ct{Q'}$  could be weakened to $\wbPI$, and the
soundness of the technique, namely Theorem \ref{t:up-to-con-expa},
would still hold. However, in general, the same weakening
cannot be made on the requirement $P'' \contr \ct {P'}$.
We show that however, the weakening is possible when 
the action $\mu$ is visible and the context cannot interact with the
processes in its holes.

This proof technique, 
(which, as far as we know, does not appear in the literature) 
will be essential in some of our key results 
about the correspondence with game semantics.

\begin{definition}
  \label{d:interact}
  A static context $\qct$  \emph{cannot interact} with a process $P$
  if no name appears free both in the context and in
  the process and with opposite polarities (that is, either in input
  position in the context and output position in the process, or the
  converse).  
\end{definition} 

If  $\qct$ {cannot interact} with $P$,
whenever $\ct P \longrightarrow Q$ then  the
interaction is either internal to $\qct$ or internal to $P$. Moreover,
the property is invariant for reduction; i.e., 
if $Q = \ctp{P'}$ where $\qctp$ is the derivative of $\qct$ and $P'$
the derivative of $P$, then also 
$\qctp$  cannot interact with  $P'$.
Notice the invariance is false in the full $\pi$-calculus.

We recall that we use $\ell$ to range over visible actions. 
\begin{definition}[bisimulation up-to context and up-to $(\contr,
 \wbPI)$]
\label{d:uptocontext_contr_wb}
A symmetric relation $  \R   \subseteq \pr \times \pr $ is a 
{\em bisimulation up-to context and up-to $(\contr,\wbPI)$}
if $P \RR Q$ imply:
  \begin{enumerate}
  \item if $P \arr\ell P''$
  then there are processes $P'$ and $Q'$ 
  and a static context $\qct$  
  s.t.\  $P'' \wbPI \ct {P'}$, $Q \Arr\ell  \wbPI  \ct{Q'}$ 
  and $P' \RR Q'$, 
  and, moreover, $\qct$ cannot interact with $P$ or $Q$; 
  \item if $P \longrightarrow P''$
  then there are processes $P'$ and $Q'$
  and a static context $\qct$
  s.t.\  $ P'' \contr \ct {P'}$, 
  $Q \Longrightarrow \contr  \ct{Q'}$ 
  and $P' \RR Q'$.
  \end{enumerate}
\end{definition}

\begin{lemma}
  \label{l:upto_new}
  Suppose $ \R  $  is a 
  bisimulation up-to context and up-to $(\contr,\wbPI)$ 
  and $\qct$ is a static context 
  that does not interact with processes $P$ and $Q$.  
  \begin{enumerate}
  \item  
  If  $\ct{P}\Longrightarrow   P_1$, 
  then there are processes $P'$ and $Q'$ 
  and a static context $\qctp$
  s.t.\  
  $P_1 \contr \ctp{P'}$, 
  $\ct Q \Longrightarrow \contr \ctp{Q'}$ 
  and $P'  \R  Q'$. 
  
  \item 
  If $\ct{P}\arr\ell   P_1$, 
  then there are processes $P'$ and $Q'$ and 
  a static context $\qctp$ and 
  s.t.\  
  $P_1 \wbPI \ctp{P'}$, 
  $\ct Q \ArrPI\ell \wbPI \ctp{Q'}$ 
  and $P'  \R  Q'$. 
  \end{enumerate} 
\end{lemma}

\noindent
The proof of  Lemma~\ref{l:upto_new} is along the lines of similar results in the
literature combining up-to expansion and up-to bisimilarity, see e.g.~\cite{SaWabook}. 

\reviewComment{there is no proof. Is that a classical lemma? Is it
because the proof is straightforward?}
\reviewCommentDS{Answer: while strictly speaking the result does not appear in the
  literature, the proof is straightforward adapation of  similar techniques. We have added
a sentence with a reference  (To be honest: the result is not entirely trivial, and
``normally'' one should have indeed added a proof. In other words, the referee  is right
is complaining. However it is certainly not the most important proof that is missing from
the main text! This paper should have really been 2 or even 3 ordinary papers)}

\begin{theorem}
\label{t:up-to-con-expa-nointer} 
  If $\R $ is a bisimulation 
  up-to context and up-to $(\contr,\wbPI)$, 
  then ${ \R } \subseteq {\wbPI}$.
\end{theorem}

\begin{proof}
  By showing that 
  the relation $\R$ containing 
  all pairs $(P,Q)$ such that 
  \[
   \begin{array}{l}
   \mbox{ for some static context $\qct$ and
  processes $P', Q'$} \\
   \mbox{ where $\qct$ does not interact with $P',Q'$ }  \\
  \mbox{ it holds that } P \wbPI \ct{P'}, Q \wbPI \ct{Q'} \mbox{ and }
  P'  \RR  Q' 
   \end{array}
   \]
  is a bisimulation. 
  In a challenge transition $P \arr\mu P_1$ one distinguishes the case
  when $\mu$ is a silent or visible action. 
  In both cases, one exploits Lemma~\ref{l:upto_new}. 
\end{proof}
\subsection{The Call-By-Value $\lambda$-calculus}
\label{ss:cbv}

The grammar
of the untyped call-by-value $\lambda$-calculus, $\LasV$, 
has values $V$, terms $M$, evaluation contexts $E$, and general contexts $C$:
\[
\begin{array}{llll}
\Values &  V & \defeq & x \midd \lambda x.M \\
\Terms & M,N & \defeq & V \midd M N \\
\EContexts & E & \defeq & \contexthole \midd V E \midd E M \\
\Contexts & C & \defeq & \contexthole \midd \lambda x.C \midd M C \midd C M
\end{array}\]
 where $\contexthole$ stands for the hole of a context. 
The call-by-value reduction $\redv$ has two rules:
  \begin{mathpar}
   \inferrule*{ }{(\lambda x.M) V \redv M\subst{x}{V}} \qquad
   \inferrule*{M \redv N}{E[M] \redv E[N]}
 \end{mathpar}
 
 In the following, we write 
 $M \Converge M'$ to indicate that $M \Redv M'$ with $M'$ an eager normal form,
 that is, either a value or a stuck call $E[x V]$.

 We recall  Lassen's 
\emph{eager normal-form (enf) bisimilarity} \cite{lassentrees}.
Enf-bisimulations relate terms, values, and evaluation contexts.

\begin{definition}
\label{d:enf}
 An \emph{enf-bisimulation} is a triple of relations on terms $\RR_{\enfM}$,
 values $\RR_{\enfV}$, and evaluation contexts $\RR_{\enfK}$ that satisfies:
 
 \begin{itemize}
  \item $M_1 \RR_{\enfM} M_2$ if either:
  \begin{itemize}
   \item both $M_1,M_2$ diverge;
   \item $M_1 \Converge E_1[x V_1]$ and $M_2 \Converge E_2[x V_2]$
   for some $x$, values $V_1,V_2$, and evaluation contexts $E_1,E_2$ with
   $V_1 \RR_{\enfV} V_2$ and $E_1 \RR_{\enfK} E_2$;
   \item $M_1 \Converge V_1$ and $M_2 \Converge V_2$ for some values $V_1,V_2$
   with $V_1 \RR_{\enfV} V_2$.
  \end{itemize}
  \item $V_1 \RR_{\enfV} V_2$ if $V_1 x \RR_{\enfM} V_2 x$ for some fresh $x$;
  \item $E_1 \RR_{\enfK} E_2$ if $E_1[x] \RR_{\enfM} E_2[x]$ for some fresh $x$.
 \end{itemize}
 \noindent
 The largest \emph{enf-bisimulation} is called \emph{enf-bisimilarity}.
\end{definition}  
\section{Operational Game Semantics}
\label{s:AOGS}
 
We introduce the representation of $\LasV$ in OGS. 
The LTS produced by the embedding of a term intends to capture the possible interactions 
between this term and its environment.
Values exchanged between the term and the environment are represented by free variables, 
ranged over by $x,y,z$. 
Continuations (i.e., evaluation contexts) are also
represented by names, called \emph{continuation names} and ranged over by  $p,q,r$.

\emph{Actions} $\act$ have been introduced in Section~\ref{s:nota}. 
In OGS, we have five kinds of (visible) actions:
\begin{itemize}
  \item \emph{Player Answers} (PA), $\ansP{p}{x}$,
  and \emph{Opponent Answers} (OA),  $\ansO{p}{x}$,
  that exchange a variable $x$ through a continuation name $p$;
  \item \emph{Player Questions} (PQ), $\questP{x}{y,p}$,
  and \emph{Opponent Questions} (OQ), $\questO{x}{y,p}$, 
  that exchange a variable $y$ and a continuation name $p$ through a variable $x$;
  \item \emph{Initial Opponent Questions} (IOQ), $\questOinit{p}$, that
  introduce the initial continuation name $p$.
\end{itemize}
 
\begin{remark}
 \label{rmk:justified-sequence}
The denotation of terms is usually represented in game semantics using 
the notion of \emph{pointer structure} rather than traces.
A pointer structure is defined as a sequence of \emph{moves}, 
together with a pointer from each move (but the initial one) to a previous move 
that  `justifies' it.
Taking a trace $\tr$, one can reconstruct this pointer structure in the following way:
an action $\act$ is \emph{justified} by an action $\act'$
if the free name of $\act$ is bound by $\act'$ in $\tr$ (here we are taking advantage of
the `freshness' convention on the bound names of traces, Remark~\ref{r:bn}).
\end{remark}

\emph{Environments}, ranged over by $\gamma$, 
maintain the association from names to values and evaluation contexts, and  
 are partial maps with finite domain. A single mapping is either of the form 
$\pmap{x}{V}$
(the variable $x$ is mapped onto the  value $V$), or 
$\pmap{p}{(E,q)}$ 
(the continuation name $p$ is mapped onto the 
pair of the  evaluation context $E$ and the continuation $q$).
We write $\emptymap$ for the empty environment,
and $\gamma \cdot \gamma'$ for the union of two such partial maps,
provided that their domains are disjoint.

There are two main kinds of \emph{configurations} $\FF$:
\emph{active configurations} $\conf{M,p,\gamma,\phi}$ and 
\emph{passive configurations}
$\conf{\gamma,\phi}$, where $M$ is a term, $p$ a continuation name, 
$\gamma$ an environment and $\phi$ a set of names called its \emph{name-support},
corresponding to the set of all names that has been exchanged so far.

 \begin{figure*}
 \[
  \begin{array}{l|llll}
   (P\tau) & \conf{M,p,\gamma,\phi} & \arrA{\ \tau \ } & 
     \conf{N,p,\gamma,\phi} & \text{ when } M \redv N\\
   (PA) & \conf{V,p,\gamma,\phi} & \arrA{\ansP{p}{x}} & 
     \conf{\gamma \cdot \pmap{x}{V},\phi \uplus \{x\}} \\
   (PQ) & \conf{E[xV],p,\gamma,\phi} & \arrA{\questPV{x}{y}{q}} & 
     \multicolumn{2}{l}{\conf{\gamma \cdot \pmap{y}{V}\cdot \pmap{q}{(E,p)},\phi 
       \uplus \{y,q\}}} \\
   (OA) & \conf{\gamma\cdot \pmap{q}{(E,p)},\phi} & \arrA{\ansO{q}{x}} & 
     \conf{E[x],p,\gamma,\phi \uplus \{x\}}\\
   (OQ) & \conf{\gamma,\phi} & \arrA{\questOV{x}{y}{p}} & 
     \conf{V y,p,\gamma,\phi \uplus \{y,p\}} & \text{ when } \gamma(x) = V\\
   (IOQ) & \initconf{\phi}{M} & \arrA{\questOinit{p}} &
     \conf{M,p,\emptymap,\phi \uplus \{p\}}
  \end{array}\]
 \caption{The  LTS for the  Alternating OGS (\AOGS)}
 \label{fig:seq-ogs-lts}
\end{figure*}

The LTS is introduced in Figure~\ref{fig:seq-ogs-lts}.
It is called \emph{Alternating}, since, forgetting the P$\tau$ transition,
it is bipartite between active configurations that perform Player actions and $\tau$-transitions,
and passive configurations that perform Opponent actions.  
Accordingly, we call Alternating the resulting OGS, abbreviated \AOGS. 
In the OA rule, $E$ is ``garbage-collected''  from $\gamma$,
a behavior corresponding to \emph{linear continuations}.

The term of an active configuration determines the next transition performed.
First, the term needs to be reduced, using the rule (P$\tau$).
When the term is a value $V$, a Player Answer (PA) is performed,
providing a fresh variable $x$ to Opponent, while $V$ is stored in $\gamma$ at position $x$. 
Freshness is enforced using the disjoint union $\uplus$.
When the term is a callback $E[xV]$, with $p$ the current continuation name,
a Player Question (PQ) at $x$ is performed, providing two fresh names $y,q$
to Opponent, while storing $V$ at $y$ and $(E,p)$ at $q$ in $\gamma$.

On passive configurations, Opponent has the choice to perform different actions.
It can  perform an Opponent Answer (OA) by interrogating an evaluation context 
$E$ stored in $\gamma$. For this, Opponent provides a fresh variable $x$ that is plugged into the hole of $E$,
while the continuation name $q$ associated with $E$ in $\gamma$ is restored.
Opponent may also  perform  an Opponent Question (OQ),
by interrogating a value $V$ stored in $\gamma$. 
For this, Opponent provides a fresh variable $y$ as an argument to $V$.

To build the denotation of a term $M$, we introduce an \emph{initial
configuration} associated with it, 
written $\initconf{\phi}{M}$, with $\phi$ the set of free variables we start with.
When this set is taken to be the free variables of $M$, we simply write it as $\conf{M}$.
In the initial configuration, the choice of the continuation
name $p$ is made by performing an Initial Opponent question
(IOQ). So initial configurations are passive.

From an environment $\gamma$, we extract its \emph{continuation structure}
which is defined
as the relation over continuation names 
\[\cs(\gamma) \defeq \{(p,q) \mid p \in \dom{\gamma} \text{ and } \gamma(p) = (\_,q)\}\]
As we will see, the continuation structure indicates
on which continuation name $q$ Player should answer 
following an Opponent answer over $p$.
Continuation names of $\phi$ that are either
the toplevel continuation name -- for active configurations --
or that are in the continuation structure of $\gamma$ are called \emph{available}.
So continuation names that are not available are the ones that
have already been used to perform an answer.
This distinction is needed since we consider linear continuations here.
Variables are always considered to be available since there is
no constraint on the number of times they can be interrogated.

Names in $\dom{\gamma}$ are called \emph{P-names}, 
and those in $\phi \backslash \dom{\gamma}$ that are available are called \emph{O-names}.
So we obtain a \emph{polarity function} $\Label_{\FF}$ 
associated to $\FF$, 
defined as the partial maps
from $\phi$ to $\{O,P\}$  mapping names to their polarity.
 In the following, we  only consider \emph{valid configurations},
 for which:
 \begin{itemize}
  \item $\dom{\gamma} \subseteq \phi$
  \item $\FV(M),p$ are O-names;
  \item for all $a \in \dom{\gamma}$, the names appearing in $\gamma(a)$ are O-names.
 \end{itemize} 
\section{The encoding of call-by-value  
$\lambda$-calculus  into the $\pi$-calculus }
\label{s:encoLpi}

\begin{figure}[t]
  \begin{mathpar}
    \begin{array}{rcl}
\encoI {V}  &\defeq& 
\bind p 
 \bout p y .    \encoIVa {V} y
\\[\myptD]
      \encoI {MN}  &\defeq&  
\bind p 
\res q \big(
\begin{array}[t]{l}
\encoIa M q
                         | 
\inp q y. \res r \big(
\encoIa N r | \\[\myptD]
\! \! \! \!\inp r w .\bout y {w',p'}.(\fwd {w'} w|\fwd {p'} p)\big)\big)
\end{array}\\[\mysp]
\encoIV {\lambda  x.  M}  &\defeq& 
\bind y 
 !\inp y {x,q}.\encoIa M q
\\[\myptD]
\encoIV x   & \defeq& 
\bind y 
 \fwd y x
   \end{array}
  \end{mathpar}
 \caption{The encoding of call-by-value  $\lambda$-calculus into \piI} 
 \label{f:enc_internal}
 \end{figure}

We recall here Milner's encoding of call-by-value $\lambda$-calculus, transplanted into
\piI. 

The core of any encoding of the  $\lambda$-calculus into a process calculus is the 
translation of function application. This  
becomes a particular form of parallel combination of two processes, the function, and its argument;  
 $\beta$-reduction is then modelled as a process interaction.

The encoding of a $\lambda$-term is parametric over a name, that is, it is an abstraction.
This parameter may
be thought of as 
its \emph{continuation}.  
A term that becomes a value signals so at its continuation name
and, in doing so, it grants access 
to the body of the value. 
Such body is replicated, so that the value may be
copied several times.
When the value is a function, its body can
receive two names: (the access to) its value-argument, and the
following continuation.
In the translation of application, first, the function is run, then
the argument; finally, the function is informed of its argument and 
continuation.
(Milner noticed \cite{Mil92s} that his
call-by-value encoding can be easily tuned so to mimic
forms of 
evaluation in which, in an application $MN$, the function $M$ is run
first, or the argument $N$ is run first, or function and argument are
run in parallel.  Here we follow  the first option,  the most common  one in
 call-by-value.)

As in OGS, so in \piI 
the encoding uses 
two kinds of names, 
the \emph{continuation names} $p,q,r,\dots$, and 
the \emph{variable names} $x,y,v,w\dots$.
Figure~\ref{f:enc_internal} presents the encoding.
The encoding function $\cbvSymb$, on values, is decomposed by means of the auxiliary
function  $\cbvSymb^\variant$; this will be useful later, e.g., in the encoding of
environments. 
\reviewComment{please type the different
components of the encoding; in particular say that V* is for values}
\reviewCommentDS{We have added an explanation}
Process   $\fwd ab$  represents a \emph{link} (sometimes called forwarder; 
for readability we have adopted the infix
notation $\fwd ab$ for the constant $\fwd{}{}$). It 
transforms all outputs at $a$ into outputs at 
$b$ (therefore $a,b$ are
names of the same sort); thus the body of $\fwd ab$ is replicated,
 unless $a$ and $b$ are {continuation names}:
 $$\begin{array}{rcl}
 \fwd{}{} &\defeq &
 \left\{ \begin{array}{ll}
 \abs{p,q} \inp p {x}.\bout {q} {y}
           . {\fwd{ y}{ x}}\\[\myptSmall] 
 \multicolumn{1}{l}{  
        ~\quad  \mbox{if $p,q$ are continuation names}
                    }              \\[\myptD]
 \abs{x,y} ! \inp x {z,p}.\bout y {w,q}
           .({\fwd{ q}{ p}}|\fwd w z)  
\\[\myptSmall]
 \multicolumn{1}{l}{  
           ~\quad \mbox{if $x,y$ are variable names}
 }\end{array} \right. 
 \end{array}
  $$

The equivalence induced on call-by-value  $\lambda$-terms by their encoding into \piI 
coincides with 
Lassen's 
\emph{eager normal-form (enf) bisimilarity} \cite{lassentrees}.

\begin{thmC}[\cite{DurierHS18}]
\label{t:adrien}
$\enca M \wbPI \enca N$ iff $M$ and $N$ are enf-bisimilar. 
\end{thmC} 

\DSOCTb

\noindent
In proofs about the behaviour of the \piI representation of $\lambda$-terms
we sometimes follow Durier et al.\ 
\cite{DurierHS18} \reviewComment{why the possessive?}
 \reviewCommentDS{It was an English typo, corrected}
and use an optimisation of Milner's encoding,
 reported in
Appendix~\ref{a:oc}.
The optimised encoding, 
indicated in the remainder as 
 $\qencV$,
   is 
obtained  from the  initial one 
$\cbvSymb$ 
by performing a few (deterministic) reductions,
at the price of 
 a more complex definition, in which a case analysis is made for the encoding of
 applications.
We only report here the  result about its behavioural correctness:

\begin{lemma}[correctness of the optimisation]
  \label{l:opt_sound} 
For all $M $, it holds that
 $\encoI M\exn\encoV M$.
Similarly, for all \AOGS configurations $\FF$, 
it holds that
 $
\encoConI \FF \exn\encoConV \FF$.
\end{lemma} 
\reviewComment{A bit confusing because you haven’t introduced this encoding}
\reviewCommentDS{well, i see no way of helping here, as we do not want to move the
  optimised encoding from the appendix to here. Hence i would not do anything here}

\noindent
The lemma is established by  algebraic
reasoning \cite{DurierHS18,Durier20}; its extension to encodings of configurations is straightforward.

\DSOCTe

\section{Relationship between \piI and \AOGS}
\label{s:encoGames}

To compare the \AOGS and \piI representations of the (call-by-value)
$\lambda$-calculus, we set a mapping from \AOGS configurations and environments 
to  \piI processes.
The mapping is reported in Figure~\ref{f:OGS_pi_cbv}. 
It is an extension of Milner's encoding of the $\lambda$-calculus and is
therefore indicated with the same symbol  $\cbvSymb$.
The optimisation $\qencV$ of Milner's encoding mentioned at the end of
the previous section is correspondingly extended to an optimisation of the mapping for
configurations. Again, we shall sometimes make use of such optimisations in proofs. 

The mapping uses a representation of environments $\gamma$ as associative lists, introduced in
Section~\ref{s:nota}.
So it relies on a choice of an order on the names that form the domain of $\gamma$.
This choice is irrelevant in the definition of the mapping into \piI, as behavioural equivalence in
\piI is invariant under commutativity and associativity of parallel composition.

\begin{remark}
 The  encoding of a configuration $\FF$ with name-support $\phi$ 
does not depend on $\phi$.
This name-support $\phi$ is used in OGS both to enforce freshness of names,
and to deduce the polarity of names,
as represented by the function $\Label$.
And indeed, the process $\encoConI{\FF}$ has its set of free names included
in $\phi$, and uses $P$-names in outputs and $O$-names in inputs.
The freshness property could be made more explicit by employing configurations also in \piI.
The polarity property could also be stated in $\pi$-calculus using 
i/o-sorting~\cite{PiSa96b}.
Indeed, a correspondence between arenas of game semantics (used to enforce polarities of moves) 
and sorting has  been 
explored~\cite{10.1145/224164.224189,10.1016/S0304-3975-99-00039-0}.
\end{remark}

\begin{figure}
Encoding of environments: \hfill   $ $ 
\[ 
\begin{array}{rcl}
\encoEnvI{\pmap{y}{V}\cdot\gamma'  } & \defeq & \encoIVa {V} y | \encoEnvI{\gamma'  }
\\[\myptD] 
\encoEnvI{\pmap{q}{(E,p)}\cdot\gamma'  } & \defeq & 
\inp qx .\encoIa {E[x]} p  | \encoEnvI{\gamma'  } \\[\myptD]
\encoEnvI{ \emptymap} & \defeq & \nil
\end{array} 
\]
Encoding of configurations: \hfill  $ $ 
\[ 
\begin{array}{rcl}
\encoConI{ \conf{M,p,\gamma, \phi}}
& \defeq &   
\encoIa {M} p | \encoEnvI \gamma 
\\[\myptD]
 \encoConI{ \conf{\gamma, \phi}}& \defeq & 
 \encoEnvI \gamma  
\\[\myptD]
 \encoConI{ \initconf{\phi}{M}} & \defeq &  
 \encoI   M
\end{array}
 \] 
\caption{From OGS environments and configurations to \piI }
\label{f:OGS_pi_cbv}
\end{figure}

\subsection{Operational correspondence}
\label{s:OC}

The following theorems establish the operational correspondence between the \AOGS and \piI representations. 

\begin{theorem} 
\label{t:opcorrCON_w}
$ $ \begin{enumerate}

\item If $\FF \LongrightarrowA \FF'$, then 
$\encoConI \FF \LongrightarrowPI  \; \contr \encoConI {\FF'}$;

\item If $\FF \ArrA\ell \FF'$, then
$ \encoConI \FF \ArrPI\ell  
\; \wbPI
\encoConI {\FF'}$.
\end{enumerate} 
\end{theorem}

\begin{theorem} 
\label{t:opcorrCON2_w}
$ $ \begin{enumerate}

\item If 
$ \encoConI \FF \LongrightarrowPI \PP$ then there is $\FF'$ such that 
$\FF \LongrightarrowA \FF'$  and  
$\PP  \wbPI \encoConI {\FF'}$;

\item If
$ \encoConI \FF \ArrPI\ell \PP$
and  $\ell$ is an output, 
 then there is $\FF'$ such that
 $\FF \ArrA\ell \FF'$
and $\PP \wbPI  \encoConI {\FF'}$;

\item If  $\FF$ is passive and 
$ \encoConI \FF \ArrPI{\ell} \PP$, 
then 
 there is $\FF'$ such that 
$\FF \ArrA{\ell} \FF'$
and $\PP  \wbPI \encoConI {\FF'}$.
\end{enumerate} 
\end{theorem} 

\noindent
Both theorems above,  as well as similar operational correspondence theorems in this and next
section,  are first proved in terms of the optimised encoding  $\qencV$ 
(introduced in \cite{DurierHS18} and discussed in the previous section and in 
Appendix~\ref{a:oc}), 
and exploiting its
tight correspondence with the main encoding $\cbvSymb$. 
Moreover, 
  the operational correspondence 
is first established 
on strong
transitions,  along the lines of analogous correspondence results, such as  those detailed  in
Section~\ref{subsec:seq-ts-pi}  and Appendix~\ref{a:oc}. 
The correspondence on strong transitions, as usual in these cases, 
is carried out by  induction on the depth of the proof of a given transition; then the
 results are extended to weak transitions using induction on the length of the  given weak
 transition.  
 The advantage of   $\qencV$ is that it does not have initial
`administrative' reductions; hence any initial reduction corresponds to an internal move
of a $\lambda$-term.   This, for instance, allows us to have 
 the expansion relation  $\contr$ (in place of the coarser $\wbPI$, even though the final
 goal is  the result for  $\wbPI$), 
in the  statement about silent actions for the 
operational correspondence from \piI to \AOGS 
(see,  e.g., Lemma~\ref{l:opcorrCON2} in the Appendix, and similar lemmas). 
The use of expansion 
is
essential, both to derive the  operational correspondence theorems  on weak transitions,  and to use
the theorems in up-to techniques for  $\piI$ (more generally, in applications of the
theorems in which one reasons about the number of steps performed).

\reviewComment{
Theorems 6.2 and 6.3: this seems to be at the core of the
translation. Could we have some details?}
\reviewCommentDS{We have added a paragraph following the theorems. The paper contains
  several operational correspondence results, whose proofs are not particularly
  interesting, and follow rather standard induction arguments.  These specific results are
similar to operational correspondence results that are found in papers that study the
relationship between lambda terms and their encoding pi-calculus  terms }

In Theorem~\ref{t:opcorrCON2_w}, a  clause is 
 missing
 for
input actions from $ \encoConI \FF$ when 
 $\FF$ active. Indeed such actions are possible in \piI, stemming from the (encoding of
 the) environment of $\FF$, whereas they are not possible in \AOGS. This is rectified in
 Section~\ref{subsec:seq-ts-pi}, introducing a constrained LTS for \piI, and in Section~\ref{s:cogs}, 
 considering a concurrent OGS.   
 
\begin{corollary}
\label{c:tracesGSpi}
If $\FF \ArrA s$ then also $ \enca \FF \ArrPI s $.
\end{corollary} 

\begin{proof}
By induction on the length of trace $s$.
\end{proof} 

\subsection{An \outputpr Transition System}
\label{subsec:seq-ts-pi}

We define an LTS for \piI in which input actions are
visible only if no output can be consumed, 
either as a  visible action or through an internal action 
(i.e., syntactically the process has no unguarded output). 
The new LTS, called \emph{\outputpr} and indicated as \opLTS,  
is defined on top of the ordinary one 
by means of the two rules below. A process $P$ is \emph{input reactive} if 
whenever $\PP \arrPI \mu  P'$, for some $\mu,P'$ 
 then  $\mu$ is an input action.
 \[
\begin{array}{l}
  \inferrule{\PP \arrPI{ \mu} \PP' \andalso  
  \mbox{$\PP$ input reactive}
 }{\PP \arrN {\mu} \PP'}
\\
  \inferrule{\PP \arrPI{ \mu} \PP'
 }{\PP \arrN {\mu} \PP'}
\mbox{$\mu$ is an output or $\tau$  action}
\end{array}
 \]
The \opLTS captures an aspect of   \emph{sequentiality} 
in  $\pi$-calculi:  
a  free input prefix  is to be  thought of as a service offered to the external environment; 
in a sequential system, such a service is available only if 
there are no ongoing computations due to previous interrogations of the server. 
An ongoing computation is represented by  
a $\tau$-action, indicating a step  of computation  internal to the server, or 
an output, indicating either an answer to a client  or a request to an
external server.  
The constraint imposed by the 
new LTS can also be formalised compositionally,
as done in Figure~\ref{t:pii_op}.

\begin{figure*}[ht]
\begin{center}
\begin{tabular}{rlcrl}
{\trans{alpha}}: &  $\displaystyle{  P\equival P' \andalso   { P'} \arrN\mu    { P''}
  \over 
  P \arrN\mu {P''}  
 } $ &  
{\trans{ pre}}:&  
$\displaystyle{ 
\over
{\mu.      P} \arrN\mu {P}}$   
 \\[\mysp]
\multicolumn{5}{c}{ 
  {\trans{ parL}}:\; \; $\displaystyle{   P \arrN\mu { P'}
 \andalso  
  \mbox{($Q$ input reactive) or ($\mu$ is an output or  $\tau$)}
 \over  { P | Q}   \arrN\mu
{P'| Q} } $  }
\\[\mysp] 
\multicolumn{5}{c}{ 
    {\trans{com}}:  \; \;    $\displaystyle{ {P} \arrN{ \mu }{P'} \andalso   
 Q
\arrN{\outC\mu}{Q'}
\andalso \mu\neq \tau,   \; \tilx = \bn\mu 
  \over
{     P |  Q }\arrN{ \tau} 
{
 \res{  \tilx} (P'
|  Q') } }$    
}
 \\[\mysp]
{\trans{ res}}:& $\displaystyle{{  P} \arrN{\mu}{P'}
\andalso  x \not \in \n\mu
 \over
{ \res x     P}   \arrN{ \mu}{\res x P'  }} $ 
&
    {\trans{rep}}: &
 $\displaystyle{ {\mu.  P }\arrN{\mu} {P'}
 \over
{ !\mu. P }   \arrN{ \mu}{ P' | !\mu. P } } $  
 \\[\mysp]
\multicolumn{5}{c}{ 
    {\trans{con}}: \;\;
$  \displaystyle{   P \arrN \mu {P'}
\andalso  \mbox{ } {\rmmm K}  \defi
  \abs{\til y} Q  
\andalso 
\mbox{ } \abs{\til y} Q \equival \abs{\til x} P
    \over 
{
 \app {\rmmm K}{\til x}}  \arrN\mu 
{P'} }$
}\end{tabular} 
\end{center} 
\caption{A compositional presentation of the \outputpr LTS  for \piI}
\label{t:pii_op}
 \Mybar 
\end{figure*}

\reviewComment{rule "pre": what is $\phi$?}
\reviewCommentDS{removed (probably a wrong cut-and-paste)}
Under the  \opLTS,  the analogue of Theorem~\ref{t:opcorrCON_w} continues to hold: 
in  \AOGS configurations,
input transitions only occur in passive configurations, 
and the encodings of passive configurations are input-reactive processes. 
However, now we have   the full  converse  of Theorem~\ref{t:opcorrCON2_w} and, as a
consequence, we can also establish  the converse direction of 
Corollary~\ref{c:tracesGSpi}.

\DSOCTb

\begin{lemma} 
\label{l:opcorrCON3_w}
$ $ \begin{enumerate}


\item If 
$ \encoConI \FF \ArrN {} \PP$ then there is $\FF'$ such that 
$\FF \LongrightarrowA \FF'$  and  
$\PP  \wbPI \encoConI {\FF'}$ ;

\item If
$ \encoConI \FF \ArrN\ell \PP$
 then there is $\FF'$ such that
 $\FF \ArrA\ell \FF'$
and $\PP \wbPI \encoConI {\FF'}$.
\end{enumerate} 
\end{lemma} 

\noindent
Again, in the proof,  we first establish a correspondence result on strong transitions,
and reason on the optimised encoding $\qencV$. 
See Appendix~\ref{a:oc} for details.

Lemma~\ref{l:opcorrCON3_w}  
talks about the \opLTS of \piI;
  however the theorems make use of the
ordinary expansion relation  $\expa $, which is defined on the ordinary LTS. 
We now show that 
such uses of
expansion can actually be replaced by expansion on the \opLTS, defined as ordinary
expansion, but on the \opLTS.

A behavioural relation on the ordinary LTS `almost always' implies the corresponding
relation on the \opLTS.  The only exceptions are produced by processes that are related
although they exhibit different divergence behaviours (by divergence we mean the
possibility of performing an infinite number of $\tau$ actions) and at the same time may
perform input actions. For instance, using $\tau^\omega$ to indicate a purely divergent
process  (i.e., a process whose only transition is $\tau^\omega \arr\tau \tau^\omega$), 
 the processes $a.\nil | \tau^\omega $ and $a.\nil$ are  bisimilar (and trace equivalent)
 in the ordinary LTS, but they are not so in the \opLTS, since the divergence in the former
 process prevents observation of the input at $a$.

Indeed, a `bisimilarity respecting divergence' would imply bisimilarity on the \opLTS. 
Such  `bisimilarity respecting divergence', written $\wbPIDIV$, is defined as $\wbPI$ with
the additional requirement that $P \wbPI Q$ implies 
\begin{itemize}
\item
$P \Uparrow $ iff  $Q \Uparrow $ \hfill $(*)$
\end{itemize} 
where the predicate $ \Uparrow $ holds for an  agent $R$  if it
has a divergence (an infinite path 
consisting of only $\tau$ actions). 

\begin{lemma}
\label{l:bisDIV_op}
Relation $\wbPIDIV$ implies $\wbOPI$.
\end{lemma} 
\noindent
The same result applies to  other relations; e.g., the `expansion  respecting divergence'
(defined by adding the clause $(*)$ above to those of the expansion
relation $\expa$)
 on the ordinary LTS implies expansion  on the \opLTS.

A comment about expansion ($\expa$) is worthwhile. 
As pointed out earlier,   the use of expansion
is important in proofs  about operational correspondence results. 
Such uses are  due to 
the appearance of expansion in 
 lemmas such as  
 Lemma~\ref{l:opt_sound} (and Lemmas~\ref{l:opt_stuck}-\ref{l:opt_con}, in
 Appendix~\ref{a:oc}),
 and
come from application of some
simple algebraic laws that respect  divergence (i.e., the laws do not add or remove any
divergence). Indeed the laws either are valid for strong bisimilarity, or they are laws
such as 
\[ 
{\res{a} ( \inp a {\tilx}.P | \bout a \tilx . Q)} \;  \contr \;
{\res{a} (P |  Q)  }
 \]
As a consequence, all occurrences of $\expa$ 
can be replaced by the expansion relation defined on the \opLTS ($\expaOP$). 
In this manner, for instance, we can   refine (the proof of)  Lemma~\ref{l:opcorrCON3_w}
and derive: 

 \begin{theorem} 
 \label{t:opcorrCON3_w}
$ $ \begin{enumerate}


\item If 
$ \encoConI \FF \ArrN {} \PP$ then there is $\FF'$ such that 
$\FF \LongrightarrowA \FF'$  and  
$\PP  \wbOPI 
 \encoConI {\FF'}$ ;

\item If
$ \encoConI \FF \ArrN\ell \PP$
 then there is $\FF'$ such that
 $\FF \ArrA\ell \FF'$
and $\PP \wbOPI \encoConI {\FF'}$.
\end{enumerate} 
\end{theorem} 

\noindent
A final remark  about 
behavioural equivalences on the ordinary and \opLTS
concerns the encoding of $\lambda$-terms. 
On such terms, 
the ordinary
bisimilarity $\wbPI$ implies the `bisimilarity respecting divergence' $\wbPIDIV$. This
is because, intuitively, a term  $\encoIa {M} p$ is divergent iff $M$ is so in the call-by-value
$\lambda$-calculus. 
(Again, this property boils down to the fact that the uses of
$\expa$ in Lemmas such as~\ref{l:opt_sound}-\ref{l:opt_con} 
 respect divergence.)
   See Appendix~\ref{a:behav} for   details.

From Theorem~\ref{t:opcorrCON3_w} we derive: 

\DSOCTe

\begin{corollary}
\label{c:traces}
For any configuration $\FF$ and trace $s$, we have 
  $\FF \ArrA s$   iff 
 $ \encoEnvI \FF  \ArrN s$.
\end{corollary} 

\begin{remark}
\label{r:bn2} 
We recall that, following Remark~\ref{r:bn} on the usage of bound names, 
in Corollary~\ref{c:traces} 
the bound names in $s$ are fresh;   thus they do not appear in $\FF$.
(Similarly in Theorems~\ref{t:opcorrCON2_w} and 
\ref{t:opcorrCON3_w}, for 
 the bound names
in $\ell$.)
\end{remark}

 \begin{remark}
 \label{r:expa_op_ord} 
 Corollary~\ref{c:traces} relies on Theorems~\ref{t:opcorrCON_w} and \ref{t:opcorrCON3_w}.
 The corollary talks about the \opLTS of \piI; however the theorems make use of the
 ordinary expansion relation  $\expa $, which is defined on the ordinary LTS.  Such uses of
 expansion can however be replaced by expansion on the \opLTS (defined as ordinary
 expansion, but on the \opLTS). For more details on this, see
 Appendix~\ref{a:behav}. 
 \end{remark} 

\noindent
As a consequence of Corollary~\ref{c:traces}, trace equivalence is the same, 
on \AOGS configurations and on the encoding of \piI terms. 
Moreover, from Theorem~\ref{t:opcorrCON3_w} the same result holds under a bisimulation semantics.
Further,
since the LTS produced by \AOGS is deterministic, 
its trace semantics coincides with its bisimulation semantics.
We can thus conclude as in Corollary~\ref{c:fa_traces_bis}.  
We recall that $\TEN$ and $\wbOPI$ are, respectively, trace equivalence and bisimilarity 
 between \piI processes in the \opLTS; similarly for  
$\TEA$ and $\wbA$ between \AOGS terms.  

\begin{corollary}
\label{c:fa_traces_bis}
For any $\FF,\FF'$ we have:
 $\FF\TEA \FF'$ iff 
$ \encoConI \FF \TEN \encoConI {\FF'}$
iff
 $\FF\wbA \FF'$ iff 
$ \encoConI \FF \wbOPI \encoConI {\FF'}$.
 \end{corollary} 

\noindent
Corollary~\ref{c:fa_traces_bis}  holds in particular when $\FF$ is the initial configuration
for a $\lambda$-term.  That is, 
the equality induced on call-by-value $\lambda$-terms by  
their representation in \AOGS and in \piI 
(under the \opLTS) is the same, both employing traces and employing bisimulation 
to handle   the observables for the two
models. 
\begin{corollary}
\label{c:OGS_opLTSpiI}
For any $\lambda$-terms $M,N$,  we have: 
\\
$\conf{M} \TEA \conf {N}$ iff
$\conf{M} \wbA \conf {N}$ iff
$\encoEnvI {{M} } \TEN\encoEnvI {{N}}$ iff 
$\encoEnvI {{M}} \wbOPI \encoEnvI {{N}}$. 
\end{corollary} 

\noindent
From Theorem~\ref{t:opcorrCON2_w} and Corollary~\ref{c:traces}, it also follows that  $\FF$ and
$\encoConI \FF$ are weakly bisimilar,  on the union of the respective  LTSs.

\section{Concurrent Operational Game Semantics}
\label{s:cogs}

In order to establish an exact correspondence between OGS and the encoding into $\pi$-calculus,
as established in Theorem~\ref{t:opcorrCON3_w}, we had to introduce the \opLTS for the 
$\pi$-calculus.

In this section, we explore another way to derive an exact correspondence between OGS and \piI,
by relaxing the Alternating LTS for OGS to allow multiple terms in 
configurations to run concurrently. 
We refer to the resulting OGS as the \emph{Concurrent OGS}, briefly \COGS (we
recall that  \AOGS refers to the Alternating  OGS of Section~\ref{s:AOGS}).  

\subsection{The Concurrent OGS LTS}
\label{ss:cs_cogs-lts}

We introduce \emph{running terms}, ranged over by $A,B$,
as partial maps with finite domains 
from continuation names to $\lambda$-terms.
Each mapping $\pmap{p}{M}$ of a running term can be seen as
a thread evaluating $M$, with $p$ an identifier
used to communicate the result of the evaluation.
A \emph{concurrent configuration} is a triple $ \conf{A,\gamma,\phi}$ 
of a running term $A$, an environment 
$\gamma$, and a set of names $\phi$.
Moreover, the domains of $A$ and $\gamma$
must be disjoint.
Both the running term and the environment may be empty.
The continuation structure of such a concurrent configuration
is then defined as the pair $(\dom{A},\cs(\gamma))$
formed by the domain of $A$ and the continuation structure of $\gamma$.
We extend the definition of the polarity function,
considering names in the domain of both $A$ and $\gamma$ as Player names. 
%

%


Passive and active configurations
can be seen as special cases of \COGS configurations with empty and singleton
running term, respectively.
For this reason we still use $\FF, \GG$ to range
over \COGS configurations. Moreover, we freely take  \AOGS
configurations to be  \COGS configurations, 
and conversely for \COGS configurations
with empty and singleton running term, 
omitting the obvious syntactic coercions. 

\begin{figure*}
 \[\begin{array}{l|lll@{}l}
   (P\tau) & \conf{A \cdot \pmap{p}{M},\gamma,\phi} & \arrC{\ \tau \ } & 
     \conf{A \cdot \pmap{p}{N},\gamma,\phi} & \text{ when } M \redv N\\
   (PA) & \conf{A \cdot \pmap{p}{V},\gamma,\phi} & \arrC{\ansP{p}{x}} & 
     \conf{A,\gamma \cdot \pmap{x}{V},\phi \uplus \{x\}} \\
   (PQ) & \conf{A \cdot \pmap{p}{E[xV]},\gamma,\phi} & \arrC{\questPV{x}{y}{q}} & 
     \multicolumn{2}{l}{\conf{A,\gamma \cdot \pmap{y}{V} \cdot \pmap{q}{(E,p)},
     \phi \uplus \{y,q\}}} \\
   (OA) & \conf{A,\gamma\cdot\pmap{p}{(E,q)},\phi} & \arrC{\ansO{p}{x}} & 
     \conf{A \cdot \pmap{q}{E[x]},\gamma,\phi \uplus \{x\}} \\
   (OQ) & \conf{A,\gamma,\phi} & \arrC{\questOV{x}{y}{p}} & 
     \conf{A \cdot \pmap{p}{V y},\gamma,\phi \uplus \{y,p\}} 
     & \text{ when } \gamma(x) = V\\
   (IOQ) & \initconf{\phi}{M} & \arrC{\questOinit{p}} &
     \conf{\pmap{p}{M},\emptymap,\phi \uplus \{p\}}
  \end{array}\]
\caption{The LTS for the Concurrent OGS}
\label{fig:conc-ogs-lts}
\end{figure*}
 
We present the rules of \COGS in Figure~\ref{fig:conc-ogs-lts}.
Since there is no more distinction between passive and active configurations,
a given configuration can perform both 
Player and Opponent actions.
Notice that only Opponent can add a new term to the running term $A$.

A  \emph{singleton configuration} is a configuration $\FF$
that has only one P-name (that is, in \COGS,  $\FF$ is either of the form 
$\conf{\pmap{p}{M},\emptymap,\phi}$, or  
$\conf{\emptymap,\pmap{x}{V},\phi}$, or $\conf{\emptymap,\pmap{p}{(E,q)},\phi}$).

In this and in the following section $\FF,\GG$ range over \COGS configurations, as reminded
by the index `$\rmm{C}$'
  in the symbols for LTS and behavioural equivalence with  which $\FF,\GG$ appear (e.g.,
  $\genTE C$).

\subsection{Comparison between \COGS and \piI}
\label{ss:compa_cogs_piI}

The encoding of \COGS into \piI is a simple adaptation of that for \AOGS.
We only have to consider the new syntactic element of \COGS, 
namely running
terms; the encoding remains otherwise the same.
The encoding of a running term is: 
\[ 
\begin{array}{rcl}
\encoEnvI{\pmap{p}{M}\cdot A  } & \defi & 
 \encoIa {M}{p} |\encoEnvI{ A  } \\
\encoEnvI{\emptymap} & \defi & \nil
\end{array} 
\]
The encoding of configurations is then defined as:
\[
\encoEnvI{\conf{A, \gamma, \phi}}
\defi \encoEnvI A | \encoEnvI \gamma
\]

The results about the operational correspondence between \COGS and \piI 
are as those between \AOGS and \piI under the \opLTS. 
\DSOCTb
We  report the main result, namely the correspondence for weak transitions, as well as 
two lemmas about the  correspondence for strong transitions, since they are also needed
when lifting up-to techniques from 
 \piI to \COGS, in Section~\ref{ss:upto_games}.
The reasoning in their proofs is similar to that 
in the operational  correspondence for \AOGS
(Lemmas~\ref{l:opcorrCON_str} and
\ref{l:opcorrCON2}).

\begin{lemma}[from \COGS to \piI, strong transitions]
\label{l:opcorrCOGS}
$ $ \begin{enumerate}

\item If $\FF \longrightarrowp \FF'$, then 
$ \encoConV \FF \longrightarrowPI  \contr \encoConV {\FF'}$;

\item If $\FF \arrC\ell \FF'$, then
$ \encoConV \FF \arrPI\ell  
\wbPI
\encoConV {\FF'}$.
\end{enumerate} 
\end{lemma} 

\begin{lemma}[from \piI to  \COGS, strong transitions]
\label{l:opcorrCOGS2}
$ $ \begin{enumerate}

\item If 
$ \encoConV \FF \longrightarrowPI \PP$ then there is $\FF'$ such that 
$\FF \longrightarrowp \FF'$  and  
$\PP   \contr \encoConV {\FF'}$ ;

\item If
$ \encoConV \FF \arrPI\ell \PP$
 then there is $\FF'$ such that
 $\FF \arrC\ell \FF'$
and $\PP \wbPI \encoConV {\FF'}$.

\end{enumerate} 
\end{lemma}

\DSOCTe

\begin{theorem}
\label{t:COGS_pi}
$ $ \begin{enumerate}

\item
If   $\FF \Longrightarrowp \FF'$ then 
$\encoEnvI \FF \LongrightarrowPI  
\contr
\encoEnvI {\FF'}$;

\item
if   $\FF \ArrC{\ell} \FF'$ then 
$\encoEnvI \FF \ArrPI\ell  
\wbPI
\encoEnvI {\FF'}$;

\item 
if 
$\encoEnvI \FF \LongrightarrowPI   \PP$ then there is $\FF'$ such that 
   $\FF  \Longrightarrowp  \FF'$ and
$\PP  \wbPI
\encoEnvI {\FF'}$;

\item 
if 
$\encoEnvI \FF \ArrPI\ell   \PP$ then there is $\FF'$ such that 
   $\FF  \ArrC{\ell} \FF'$ and
$\PP \wbPI
\encoEnvI {\FF'}$.
 \end{enumerate} 
\end{theorem} 


\begin{corollary} 
\label{c:ogsp_pi}
For any  \COGS configuration $\FF$ 
 and trace $\tr$,  
we have  $\FF \ArrC{\tr}$ iff $\encoEnvI \FF \ArrPI \tr $.
\end{corollary} 

\noindent
From Corollary~\ref{c:ogsp_pi} 
and Theorem~\ref{t:COGS_pi}, we derive:

\begin{lemma} 
\label{l:ogsp_pi_traces_bis}
For any $\FF,\FF'$ we have:
\begin{enumerate}
\item
 $\FF_1 \TEp \FF_2$ iff 
 $\encoEnvI {\FF_1} \TEPI \encoEnvI {\FF_2}$;  
\item
 $\FF_1 \wbC \FF_2$ iff 
 $\encoEnvI {\FF_1} \wbPI  \encoEnvI {\FF_2}$.  
\end{enumerate} 
\end{lemma} 

\noindent
To derive the full analogue of Corollary~\ref{c:OGS_opLTSpiI},  
we now show that, on the \piI representation of $\lambda$-terms,
trace equivalence is the same as bisimilarity.
This result needs a little care:
it is known that on deterministic LTSs
bisimilarity coincides with trace equivalence.
\iflong
  (In a deterministic
 LTS, intuitively,  for any process and action there exists at most
 one proof of derivation of that action from the process).
\fi
However,
 the behaviour of the \piI representation of a  \COGS
 configuration
 need not be
deterministic, because there could be
multiple silent transitions as well as multiple output transitions 
(for instance, in \COGS
rule OQ may be
applicable to different terms).

\iflong
\finish{some more explanations could go here. The problem is really
the outputs. with only silent actions a weaker form of determinacy
would make things ok} 
\fi

\begin{lemma}
\label{l:tr_bisi_sing}
For any $M,N$ we have: $
{\encoEnvI M} \TEPI {\encoEnvI N}
$ iff 
${\encoEnvI M} \wbPI {\encoEnvI N}$.
\end{lemma} 

\begin{proof}
Bisimilarity implies trace equivalence. 
 The delicate direction is the converse one.
It is easier to work with the optimised encoding  $\qencV$.
The relation 
\[ 
(\encoConV \FF,  \encoConV \GG) \st \mbox{ $\FF,\GG$ are singleton with $\FF \TEp
\GG$}
\]
is a bisimulation up-to context and up-to $(\contr,
 \wbPI)$. 
The proof uses
Lemmas~\ref{l:immediate_tr} and \ref{l:par_tr}
and is reported in
Appendix~\ref{a:aux_tr_bisi_sing}.
\end{proof}

\noindent
 We can finally  combine Lemmas~\ref{l:tr_bisi_sing} and~\ref{l:ogsp_pi_traces_bis}  
 to  derive that the \COGS and \piI semantics of $\lambda$-calculus
coincide, both for traces and for bisimilarity. 

\begin{corollary}
\label{c:fa_all}
For all $M,N$ we have: 
$\conf{M} \TEp \conf {N}$ iff 
$\conf{M} \wbC \conf {N}$ iff 
$\encoEnvI {{M} } \TEPI\encoEnvI {{N}}$
iff 
$\encoEnvI {{M}} \wbPI \encoEnvI {{N}}$. 
\end{corollary} 


\subsection{Tensor Product}
\label{ss:tp}
We now introduce a way of combining configurations, which corresponds to the notion of
tensor product of arenas and strategies in (denotational) game semantics.

\begin{definition}
Two concurrent configurations $\FF,\GG$ are said to be \emph{compatible} if 
their polarity functions  $\Label_{\FF},\Label_{\GG}$ are compatible~--- that is,
for all $a \in \dom{\Label_{\FF}} \cap \dom{\Label_{\GG}}$,
we have $\Label_{\FF}(a) = \Label_{\GG}(a)$.
\end{definition}

\noindent
Following Definition~\ref{d:interact},
two $\pi$-calculus processes $P_1,P_2$ \emph{cannot interact} 
if there is no name that appears free in both processes
and with opposite polarities 
(that is, in input position in one, 
and output position in the other). 
As
 in \piI all names exchanged are fresh, 
the  `cannot interact' property is preserved by transitions.

\begin{lemma}
  \label{l:ci_trans}
If $P_1,P_2$ cannot interact and 
$P_1 | P_2 \ArrPI \trace P $, then $P =  P_1'| P_2'$, for some 
$P_1', P_2'$ that cannot interact.
\end{lemma}
 
\noindent
The \piI translation of
two compatible \COGS configurations gives two processes that 
cannot interact with each other.

\begin{lemma}
  \label{l:compat_cannot-interact}
  Taking two \COGS configurations $\FF_1,\FF_2$,
  if their polarity functions are compatible then $\encoConI {\FF_1}$ and
  $\encoConI {\FF_2}$ cannot interact.
\end{lemma}

\begin{definition}
  \label{d:tensor_cogs}
For  compatible configurations
$ \FF = \conf{ A , \gamma, \phi }$
and 
$ \GG = \conf{ B , \delta, \phi' }$,  the \emph{tensor product} 
$\FF \appendSYMB \GG$ is  defined as 
\[ \appendCon \FF \GG  \defeq
 \conf{ A \cdot B , \gamma \cdot \delta, \phi \cup \phi' }
 \] 
 \end{definition}

\noindent
 The polarity function of $\appendCon \FF \GG$ 
 is then equal to $\Label_{\FF} \cup \Label_{\GG}$, and
$\encoConI {
\appendCon
{\FF_1} {\FF_2}}
\equiv 
\encoConI {
\FF_1}
|
\encoConI {
 {\FF_2}}$,  where $\equiv$ is the standard structural congruence of  $\pi$-calculi.

In the following, we write $\interl{\trace_1}{\trace_2}$ for the set of traces
obtained from an interleaving of the elements in the sequences $\trace_1$ and $\trace_2$.

We relate the traces generated by the tensor product 
of two configurations to the set of interleavings of traces generated 
by the component configurations themselves. 
The result is proved by going through the $\pi$-calculus,
via the following lemma.

\begin{lemma}
  \label{l:dec_pi}
  Suppose $P_1,P_2$ cannot interact.
  \begin{enumerate}
  \item
  If $P_1 | P_2 \ArrPI \trace   P_1'| P_2'$,
  then, for $i=1,2$,  there is  $ \trace_i $ 
  such that $P_i \ArrPI { \trace_i} P_i'$, and $ \trace  
  \in   \interl { \trace_1   }{ \trace_2 }$;
  \item 
  Conversely,
  if for  $i=1,2$,  we have $P_i \ArrPI { \trace_i} P_i'$, and $ \trace  
  \in   \interl { \trace_1   }{ \trace_2 }$, then 
  $P_1 | P_2 \ArrPI \trace   P_1'| P_2'$.
  \end{enumerate} 
\end{lemma}

\begin{proof}
  The interesting case is (1). 
  Consider the sequence of one-step  actions (including silent actions) that 
  are performed to obtain the trace
  $P_1 | P_2 \ArrPI \trace   P_1'| P_2'$.  Then tag each of these actions with ${\tt L}$ or 
  ${\tt R}$ depending on whether in the derivation proof of that action, the last rule
  applied was {\tt parL} or  {\tt parR}.
  Then $ \trace_1$ is obtained by collecting the subsequence of $\trace$ with  
  the {\tt ParL} tag, and similarly for   $ \trace_2$ with
  {\tt ParR}.
\end{proof}

\noindent
We finally prove that the set of traces generated by $\FF_1 \otimes \FF_2$
is the  union of the sets of interleaving $\interl{\trace_1}{\trace_2}$,
for $\FF_1 \ArrC{\trace_1}$ and $\FF_2 \ArrC{\trace_2}$.
\begin{lemma}
\label{l:COGS-interleave-tensor}
\iflong
Suppose $\FF_1,\FF_2$ are compatible concurrent configurations.
\begin{enumerate}
\item\label{l:COGS-interleave-tensor_1}
If $\appendCon
 {\FF_1} {\FF_2} \ArrC \trace  \FF  $, then, 
for $i=1,2$,  there are traces $ \trace_i $
such that $ {\FF_i} \ArrC { \trace_i} \FF'_i $, and $ \trace  
\in   \interl { \trace_1   }{ \trace_2 }$, and 
$\FF = \appendCon
 {\FF'_1} {\FF'_2}$;
\item\label{l:COGS-interleave-tensor_2}
Conversely,
if for  $i=1,2$, we have
 $ {\FF_i} \ArrC { \trace_i} \FF'_i$,
and
$ \trace  
\in   \interl { \trace_1   }{ \trace_2 }$, 
then 
$ 
\appendCon
{\FF_1} {\FF_2} \ArrC \trace   \appendCon
 {\FF'_1} {\FF'_2}$.
\end{enumerate}
\fi
\end{lemma}  

\begin{proof}
  It follows from Lemma~\ref{l:dec_pi}, 
  Lemma~\ref{l:compat_cannot-interact},
  Corollary~\ref{c:ogsp_pi}, and the fact that
  $\encoConI {
  \appendCon
  {\FF_1} {\FF_2}}
  \equiv 
  \encoConI {
  \FF_1}
  |
  \encoConI {
  {\FF_2}}$.
\end{proof}

\noindent
The tensor product of \AOGS configurations will be introduced
in Section~\ref{a:tensor-aogs}.

\subsection{Up-to techniques for games}
\label{ss:upto_games}

We introduce up-to techniques for \COGS, which allow, in bisimulation proofs, 
to split two \COGS configurations  
into separate components and then to reason separately on these.
These up-to techniques are directly imported from \piI.
Abstract settings for up-to techniques have been developed, see
\cite{SanPous,PousS19}; we cannot however derive the OGS techniques from
them because these settings are specific to first-order LTS (i.e., CCS-like, 
without binders within actions). 
They are then used to prove that
\COGS and \AOGS yield the same equivalences on
$\lambda$-terms.
A further application is in Section~\ref{s:enf}, 
discussing eager normal-form bisimilarity. 

In order to compare two configurations,
they must share the same structure, as defined by the following notion.
\begin{definition}
\label{d:compa}
Two configurations $\FF,\GG$ are \emph{support-equivalent}, written 
$\okC \FF\GG$, when they have the same polarity function (and so the same names) and the
same continuation structure.

\noindent
A relation $\R$ on configurations is \emph{well-formed} if it relates
support-equivalent configurations. 
\end{definition} 

\noindent
Below, all relations on configurations are meant to be well-formed. 
Given a well-formed relation $\R$ we write
\begin{itemize}
\item $\appendRel \R $ for the relation
$
\{ (\FF_1,\FF_2) \st
\mbox{$\exists$ $\GG$ s.t.\ }
 \FF_i = \appendCon{\FF'_i}\GG \mbox{ ($i=1,2$)}  \mbox{ and } 
 \FF'_1\RR \FF'_2\}
$.
\item
$\uptoComp \R $ for the reflexive and transitive closure of 
$\appendRel \R $.

Thus from $\FF_1 \RR  
\GG_1$ and  $\FF_2 \RR
\GG_2$ we obtain 
$(\appendCon{\FF_1}{\FF_2})  \uptoComp{\R}   
(\appendCon{\GG_1}{\FF_2}) \uptoComp{\R}  \appendCon{\GG_1}{\GG_2}.$

\item 
$\uptoCompRed \R $ for the closure of $\uptoComp \R $ under
reductions. That is, $\FF_1 \uptoCompRed \R \FF_2$ holds if there are 
$\FF_i'$, $i=1,2$ with $\FF_i \Longrightarrowp \FF'_i$ and 
$\FF_1' \uptoComp \R \FF_2'$.
(As $\Longrightarrowp$ is reflexive, we may have $\FF_i =
\FF'_i$.) 
\end{itemize}
 
\begin{definition}
\label{d:uptoComp}
A relation $\R$ on configurations is a \emph{bisimulation up-to
 reduction and   composition } if whenever $\FF_1 \RR \FF_2$: 
\begin{enumerate}
\item
 if $\FF_1 \arrC\act \FF_1'  $ then there is $\FF_2' $ such that 
$\FF_2 \Arcapp\act \FF_2'  $ and 
$\FF'_1 \uptoCompRed \R \FF'_2$ ;

\item the converse, on the transitions from  $\FF_2$.
\end{enumerate}
\end{definition}

\DSOCTb

\noindent
Sometimes we do not need  the `closure under reduction'; that is,  referring to  the clauses
of Definition~\ref{d:uptoComp} 
 we may take  
$\FF'_i = \FF_i$, ($i=1,2$).
In this case  
we say that $\R$ is a 
`{bisimulation up-to
   composition}'.

For the soundness of 
bisimulation up-to reduction and composition,
we first need the following Lemma~\ref{l:uptoOGSpi}, 
which is a  consequence of the compositionality of the mapping from
\COGS to \piI.  

\begin{lemma}
\label{l:uptoOGSpi}
For a relation $\R$ on \COGS configurations, 
if 
${(\FF_1, \FF_2)} \in  { \uptoComp \R}$, then 
we have $( \encoConV {\FF_1},  \encoConV {\FF_2}) \in  \ctx{( \encoConV{ \R })}$. 
\end{lemma} 

\noindent
We are now ready to show  the soundness of 
bisimulation up-to reduction and composition.

\begin{theorem}
\label{t:uptoComp}
If $\R$ is bisimulation up-to
reduction and composition then ${\R} \subseteq {\wbC}$.
\end{theorem}

\begin{proof}
We prove the result via \piI, exploiting Lemma~\ref{l:ogsp_pi_traces_bis}(2).
Here again, it is convenient to use
the optimised encoding $\qencV$.

If $\R$ is a relation on \COGS configurations, then 
$\encoConV \R$ is the relation on \piI processes obtained by mapping each pair
in $\R$ in the expected manner: 
\[ \encoConV \R \defi \{ (\encoConV {\FF_1}, \encoConV {\FF_2}) \st \FF_1 \RR \FF_2 
 \} \] 
And, if $\S$ is a relation on \piI processes, then $\ctx \S$ is the closure of $\S$ under
polyadic contexts, i.e., the set of pairs  of the form  
$(\ct {P_1,\ldots, P_n},\ct {Q_1,\ldots, Q_n})$ where $\qct$ is a polyadic context and for
each $i$, $P_i \SS Q_i$. 

As a consequence of the compositionality of the mapping from
\COGS to \piI,
for  a relation   $\R$ on \COGS configurations, 
if ${(\FF_1, \FF_2)} \in  { \uptoComp \R}$, then 
$( \encoConV {\FF_1},  \encoConV {\FF_2}) \in  \ctx{( \encoConV{ \R })}$.

Now, let $\R$ be the relation in the hypothesis of the theorem.
We  prove the theorem by  showing that 
$\encoConV \R$ is a bisimulation
 up-to context and up-to $(\contr,
 \wbPI)$ in 
\piI, and then appealing to Theorem~\ref{t:up-to-con-expa-nointer}.  

Suppose 
$\encoConV {\FF_1}  \arrPI \ell P$.  Then
there is $\FF_1'$ with  
$\FF_1  \arrPI \ell \FF_1' $ and $P  \wbPI \encoConV {\FF'_1}$ (this is derived from
Lemma~\ref{l:opcorrCOGS2}).  
Since $\R$ is bisimulation up-to
 reduction and  composition and $\FF_1 \RR \FF_2$, there are 
  $\FF_1''$, and $\FF_2'$, with  
$\FF'_1 \Longrightarrowp \FF''_1$, 
$\FF_2
  \arrPI \ell \Longrightarrowp \FF_2' $ and
$\FF''_1 \uptoComp \R \FF'_2$.

Using  Theorem~\ref{t:COGS_pi}(1) 
from 
$\FF'_1 \Longrightarrowp \FF''_1$ and the inclusion ${\contr} \subseteq
{\wbPI}$, 
we derive
$\encoConV {\FF'_1}  \wbPI  \encoConV {\FF''_1}$; 
and  by   
Theorem~\ref{t:COGS_pi}(1),
and  Lemmas \ref{l:opt_sound} and \ref{l:opcorrCOGS},  from 
$\FF_2
  \arrPI \ell \Longrightarrowp \FF_2' $ 
we derive
$\encoConV {\FF_2}  \arrPI \ell \wbPI  \encoConV {\FF'_2}$.

Finally,  
\iflong
by Lemma~\ref{l:uptoOGSpi},  
\fi
from $(\FF''_1, \FF'_2) \in  \uptoComp \R$, we derive 
 $( \encoConV {\FF'_1},  \encoConV {\FF'_2}) \in  \ctx{( \encoConV{ \R })}$. 
This closes the proof, up-to polyadic contexts and bisimilarity.

The case when 
$\encoConV {\FF_1}$ makes a silent move is simpler,  
using expansion in place of bisimilarity.
\end{proof} 

\noindent
We now present an additional up-to technique, 
which will be used when relating 
OGS bisimilarity with eager normal form bisimilarity in Section~\ref{s:enf}.
When $\R$ is a relation on singleton configurations, 
a straightforward simplification 
of `bisimulation up-to reduction and composition'
allows us to play the bisimulation game only on visible actions;
the bisimulation clause becomes: 
\begin{itemize}
\item
 if $\FF_1 \Longrightarrowp\arrC\ell \FF_1'  $ then there is $\FF_2' $ such that 
$\FF_2\Longrightarrowp\arrC\ell \FF_2'  $ and 
$\FF'_1 \uptoComp \R \FF'_2$.
\end{itemize} 
We call this a \emph{singleton bisimulation up-to composition}; 
its soundness is a corollary of that for bisimulation   
up-to reduction and  composition, 
derived from the up-to techniques for \piI and the
mapping from \COGS  onto \piI.

\begin{proposition}
\label{p:sing_bis}
If $\R$ is a singleton bisimulation up-to composition, then ${\R} \subseteq {\wbC}$.
\end{proposition} 

\begin{proof}
 Take 
$$\S \defi \R \cup \{ (\FF_1,\FF_2) \st 
\mbox{ there are } \FF'_1, \FF'_2  \mbox{ with } \FF'_1 \RR \FF'_2 
\mbox{ and   }\FF'_i \Longrightarrowp \FF_i  \mbox{ $i=1,2$}  \}
$$
Intuitively,  $\S$  is the closure of $\R$ under deterministic reductions. 
Then $\S$ is a bisimulation up-to composition.
Note that we exploit the fact that transitions of singleton configurations are deterministic.
\end{proof}

\DSOCTe

\begin{remark}
\label{r:CONFnoSINGL} 
Results such as Corollary~\ref{c:fa_all} and 
Theorem~\ref{t:uptoComp} might suggest that the equivalence between
two configurations implies the equivalence of all their singleton 
components. That is, if $\FF \TEp \GG$, with 
$\pmap{p}{M}$ part of $\FF$ and $\pmap{p}{N}$ part of $\GG$, 
then also $\conf {\pmap{p}{M}} \TEp \conf {\pmap{p}{N} }$. 
A counterexample is given by the configurations 
$$
\begin{array}{rcl|l}
\FF_1 & \defi & \conf{ \pmap{p_1}{M}\cdot \pmap{p_2}{\Omega} } 
& \text{ where } M \defi (\lambda z. \Omega)(x \lambda y . \Omega) \\
\FF_2 & \defi & \conf{ \pmap{p_1}{\Omega}\cdot \pmap{p_2}{M} }
& 
\end{array}
 $$ 
Intuitively the reason why $\FF_1 \TEp \FF_2$  is that
the term $M$ can produce an output (along the variable $x$), 
 but an observer  will
never obtain access to  the  name at which $M$ is located
($p_1$ or $p_2$).
That is, the term $M$ can interrogate $x$,
but it will never answer,  neither at $p_1$ nor at $p_2$.
However the result does hold when $\FF,\GG$ are reachable from
$\lambda$-terms, that is, there are $M,N$ and a trace $s $ such
that $\conf M \ArrC s  \FF$ and $\conf N \ArrC s \GG$. 
\end{remark} 

\begin{remark}
  \label{rm:up-to-comp-gs}
  Up-to composition can be understood as a decomposition principle for configurations 
  whose components do not interact with each other, even though they may share names 
  that are only used for communication with the environment. 
  This decomposition reduces the OGS bisimulation game to the analysis of singleton configurations, 
  effectively allowing one to reason independently on each branch of a (possibly infinite) 
  Böhm-like trees representation of normal forms, as it will be
  made apparent in Section~\ref{s:enf} with Lassen trees. 
  Through the well-known correspondence between such branches and P-views in game semantics~\cite{curien1998abstract}, 
  this yields a view-based reasoning principle: up-to composition enables one to analyse strategies 
  by considering each P-view independently, in a forward manner. 
  In this sense, it captures and operationalises a characteristic feature of innocent strategies, 
  namely that behaviours can be decomposed along views without mutual interference between them.
\end{remark}

\subsection{Relationship between Concurrent and Alternating OGS}
\label{subsec:relate-conc-alt-lts}

In Section~\ref{s:encoGames} we have proved that the
trace-based and bisimulation-based semantics produced by
\AOGS and by \piI under the \opLTS coincide. 
In Section~\ref{ss:compa_cogs_piI} we have obtained the
same result for \COGS and \piI under the ordinary LTS.
In this section, we develop these results to conclude that
all such equivalences for $\lambda$-terms actually coincide.  In other words, the equivalence
induced on $\lambda$-terms by their representations in OGS and \piI is
the same, regardless of whether we adopt the alternating or concurrent
flavour for OGS, the \opLTS or the ordinary LTS in \piI, a trace or a
bisimulation semantics.

For this, in one direction we show that the trace semantics induced by
\AOGS coincides with the subset of \emph{alternating} traces of \COGS.
 
\begin{lemma}
\label{l:traces_AOGS_COGS}
For any \AOGS configuration $\FF$, 
$\FF \ArrA{\tr}$ if and only if $\FF \ArrC{\tr}$ with $\tr$ an alternating
trace starting with an output when $\FF$ is active.
\end{lemma} 

\begin{proof}
Taking a trace $\tr$ such that $\FF \ArrA{\tr}$,
we directly get that $\tr$ is alternating.
Moreover, $\tr$ starts with an output when $\FF$ is active.
Since the rules of \AOGS LTS are a subset of the ones of \COGS,
we also get that $\FF \ArrC{\tr}$.

In the opposite direction, 
we take $\tr$ an alternating trace such that 
$\FF \ArrC{\tr}$ and $\tr$ starts 
with an output when $\FF$ is active.
We prove $\FF \ArrA{\tr}$ by induction over $\tr$.
If the trace is empty, this is direct.
Otherwise, we write $\tr$ as $\act \tr'$ such that
$\FF \arrC{\act} \GG \ArrC{\tr'}$.
Then we reason by case analysis over $\FF$:
\begin{itemize}
  \item either $\FF$ is active so that $\act$ must be an output by the statement
  of the lemma.
  So we get that $\FF \arrA{\act} \GG$ and $\GG$ is passive.
  \item or $\FF$ is passive, so that $\act$ has to be an input.
  We get that  $\FF \arrA{\act} \GG$ with $\GG$ active.
  Since $\tr$ is alternating, we then get that $\tr'$ is
  starting with an output. 
\end{itemize}
Since $\tr'$ is alternating,
in both cases we conclude using the induction hypothesis.
\end{proof}
 
\noindent
In the opposite direction, the
difficulty is that the \COGS LTS can perform more actions than the \AOGS:
it can perform input actions on active configurations.
So we state the correspondence only on singleton configurations.
But neither the \COGS nor the \AOGS output transitions preserve the fact of
being singleton.
The solution is to use bisimulation up-to composition over the \COGS LTS
in order to keep reasoning on singleton configurations.
We then have to only extract singleton configurations
from \AOGS configurations that can effectively be triggered directly by an action.

\begin{definition}
A singleton \COGS configuration $\FF$ is 
a \emph{triggerable} subconfiguration of an \AOGS configuration $\FF'$ when:
\begin{itemize}
  \item either $\FF'$ is active of the shape
  $\conf{M,p,\gamma,\phi}$,
  and $\FF = \conf{\pmap{p}{M},\emptymap,\phi \backslash{\dom{\gamma}}}$
  \item or $\FF'$ is passive of the shape
  $\conf{\gamma,\phi}$,
  and $\FF$ is:
  \begin{itemize}
    \item either of the shape
    $\conf{\pmap{x}{V},\phi \backslash{\dom{\gamma}} \cup\{x\}}$
    with $\pmap{x}{V} \subseteq \gamma$;
    \item or of the shape $\conf{\pmap{p}{(E,q)},\phi \backslash{\dom{\gamma}} \cup\{p\}}$
    with $\pmap{p}{(E,q)} \subseteq \gamma$.
  \end{itemize}
\end{itemize}
\end{definition}
\noindent
Notice that a triggerable subconfiguration $\FF$ of $\FF'$ is not formally a subconfiguration of $\FF'$,
as one needs to remove the Player names of $\FF'$ not in the domain of $\FF$, and in the case of $\FF'$ active,
we transform the singleton running term of $\FF$ into a term and a continuation name.

Configurations $\conf{\pmap{x}{V},\phi}$ and
$\conf{\pmap{p}{(E,q)},\phi}$ are not triggerable 
when $\FF$ is active,
since they cannot be directly triggered by the \AOGS LTS
(while they can be in the \COGS LTS).

\reviewComment{"triggerable subconfiguration" $\rightarrow$ terminology is confusing since
a triggerable subconfiguration is not a subconfiguration?}
\reviewCommentDS{I would add here this sentence:   
 $\FF$ being  
a triggerable subconfiguration of $\FF'$  means that $\FF$ is obtained from $\FF'$ as in
the definition; it 
does not imply that  $\FF$ is a syntactic subterm
of $\FF'$.}
\reviewCommentGJ{I've added an explanation above.}

\begin{lemma}
\label{l:OGS_COGS}
If $\FF, \GG$ are \AOGS singleton configurations and 
$\FF \wbA \GG$, then also $\FF \wbC \GG$. 
\end{lemma}

\begin{proof}
  Let $\R$ be the relation on \COGS singleton configurations with $\FF \RR \GG$ if
   $\okC \FF\GG$ and there exist \AOGS configurations $\FF', \GG'$
    with $\FF' \wbA \GG'$ such that 
    $\FF$ and $\GG$ are triggerable singleton configurations
    respectively of $\FF'$ and $\GG'$.

  First notice that if $\FF, \GG$ are \AOGS singleton configurations and 
  $\FF \wbA \GG$, we directly get that 
  $\FF \RR \GG$ since $\FF$ and $\GG$ are 
  triggerable subconfigurations of themselves,
  once the implicit coercion of \AOGS configurations is applied
  to \COGS configurations.
  
  We now show that $\R$ is a bisimulation up-to composition.
  We take $\FF \RR \GG$, such that there exist $\FF', \GG'$
  as stated by the definition of $\R$.
  Suppose 
  \begin{equation}
  \label{e:KoaH}
   \FF \arrC{\act} \FF_1 
  \end{equation} 
  in \COGS.  
  First, we suppose that $\FF,\GG$ are both passive, 
  so that the rule can be OA or OQ.  
  We assume it is OQ; the case for OA is similar. 
  We thus have, for some $p,V,x,y$ 
  \[\begin{array}{rcl}
   \FF &=&  \conf{ \gamma_1, \phi} \mbox{ with }  \gamma_1 = \pmap{x}{V} \\
   \act & = & \questO{x}{y,p} \\
   \FF_1 & = & \conf{\pmap{p}{V y}, \gamma_1,\phi \uplus \{y,p\}}
  \end{array}\]
  Moreover, we have $\GG = \conf{ \delta_1, \phi}$ 
  with $\delta_1 = \pmap{x}{W}$.
  Since $\FF,\GG$ are triggerable respectively from $\FF',\GG'$,
  we get that both $\FF',\GG'$ are passive.
  From $\FF' \wbA \GG'$, we get that $\okC{\FF'}{\GG'}$,
  so there exist $\gamma, \delta$ sharing the same domain such that:  
  \[\begin{array}{rcl}
  \FF' & = & \conf{ \gamma_1 \cdot \gamma,\phi \uplus \dom{\gamma}} \\ 
  \GG' & = & \conf{ \delta_1 \cdot \delta,\phi \uplus \dom{\delta}} 
  \end{array}\]
  We also have 
  $ \FF' \arrA{\act} \conf{V y,p, \gamma_1 \cdot \gamma, \phi\uplus \dom{\gamma} \uplus \{y,p\}} 
  \defi \FF''$. 
  Since $\FF' \wbA \GG'$, 
  \[ \GG'  \arrA{\act} \conf{W y,q, \delta_1 \cdot \delta, \phi \uplus \dom{\delta} \uplus \{y,p\}}  
  \defi \GG''\]
  and $\FF'' \wbA \GG''$.
  Let $\GG_1 =   \conf{\pmap{q}{W y}, \delta_1,\phi}$.
  We also have 
  \[ \GG \arrC  \act \GG_1. \] 
  This is sufficient to match \reff{e:KoaH}, as 
  $\FF_1 \uptoComp \R  \GG_1$. Indeed we have:
  \begin{itemize}  
  \item $\conf{\gamma_1,\phi\uplus \{y,p\}} \R  \conf{\delta_1,\phi\uplus \{y,p\}}$
  since $\FF \RR \GG$;
  \gj{There is a small glitch: we have extra names $y,p$ here
  that are not in $\FF,\GG$.}
  \item  $\conf{\pmap{p}{V y}, \phi\uplus \{y,p\}} \R  \conf{\pmap{q}{W y}, \phi\uplus \{y,p\}}$
  since they are triggerable subconfigurations respectively of $\FF''$ and $\GG''$.
  \end{itemize}

  \noindent
  Now suppose that $\FF,\GG$ are active,
  so there are 3 possibilities of transitions, corresponding to rules 
  P$\tau$, PA, PQ. The case of 
  P$\tau$ is straightforward since bisimilarity is preserved by internal
  actions.  We only look at PQ, as PA is simpler. 
  Thus we have 
  $\FF = \conf{\pmap{p}{E[x V]},\emptymap,\phi}$, for some $E,x,V$ and: 
  \begin{equation}
  \label{e:KpqH}
   \FF  \arrC{\questPV{x}{y}{q}} 
        \conf{\pmap{y}{V} \cdot \pmap{q}{(E,p)}, \phi \uplus \{y,q\}} 
        \defi \FF_1
  \end{equation} 
  Since $\GG$ shares the same continuation structure as $\FF$,
  we have $\GG = \conf{\pmap{p}{M},\emptymap,\phi}$ for some term $M$.
  Since $\FF,\GG$ are triggerable respectively from $\FF',\GG'$,
  there exist $\gamma,\delta$ sharing the same domain such that:  
  \[\begin{array}{rcl}
    \FF' & = & \conf{E[x V], p, \gamma, \phi\uplus\dom{\gamma}} \\ 
    \GG' & = & \conf{M, p, \delta, \phi\uplus\dom{\delta}} 
  \end{array}\]
  We also have 
  \[
  \FF' \arrA{\questPV{x}{y}{q}} 
   \conf{ \pmap{y}{V} \cdot \pmap{q}{(E,p)} \cdot \gamma, \phi  \uplus \{y,q\}}
   \defi \FF''
  \]
  As $\FF' \wbA \GG'$, there exist $\Ep,W$ such that 
  \[
  \GG'  \Longrightarrow 
   \conf{\Ep[x W],p,\delta,\phi}
  \arrA{\questPV{x}{y}{q}} 
  \conf{\pmap{y}{W} \cdot \pmap{q}{(\Ep,p)} \cdot \delta, \phi \uplus \{y,q\}} 
  \defi \GG''
  \]
  with
  $\FF'' \wbA \GG''$. 
  Defining $\GG_1$ as
  $\conf{\pmap{y}{W} \cdot \pmap{q}{(\Ep,p)}, \phi \uplus \{y,q\}}$,
  we also have $\GG \ArrC{\questPV{x}{y}{q}} \GG_1$. 
  We deduce that both 
  $\conf{ \pmap{y}{V},\phi \uplus \{y\}} 
   \RR
   \conf{ \pmap{y}{W},\phi \uplus\{y\}}$ 
  and 
  $\conf{ \pmap{q}{(E,p)},\phi \uplus \{q\}}
   \RR
   \conf{ \pmap{q}{(\Ep,p)},\phi \uplus \{q\}}$ hold
   since all these singleton configurations
   are triggerable from either $\FF''$ or $\GG''$,
   and that $\FF'' \wbA \GG''$.
  Hence $\FF_1 \uptoComp \R  \GG_1$.
  
  In summary, as an answer to the challenge 
  \reff{e:KpqH}, we have found  $\GG_1$ such that 
   $\GG \ArrC{\questPV{x}{y}{q}}
   \GG_1$ and  
  $\FF_1 \uptoComp \R  \GG_1$; this closes the proof.
\end{proof} 

\noindent
From Corollaries~\ref{c:OGS_opLTSpiI}, \ref{c:fa_all}, and 
Lemmas~\ref{l:traces_AOGS_COGS}, \ref{l:OGS_COGS} we can thus conclude. 
\begin{corollary}
\label{c:OGS_COGS}
For any $\lambda$-terms $M,N$, the following statements are the same: 
$\conf{M} \TEA \conf {N}$;
$\conf{M} \wbA \conf {N}$;
$\conf{M} \TEp \conf {N}$;
$\conf{M} \wbC \conf {N}$;
$\encoEnvI {{M} } \TEN\encoEnvI {{N}}$; 
$\encoEnvI {{M}} \wbOPI \encoEnvI {{N}}$; 
$\encoEnvI {{M} } \TEPI\encoEnvI {{N}}$; 
$\encoEnvI {{M}} \wbPI \encoEnvI {{N}}$. 
\end{corollary} 
\section{Eager Normal-Form Bisimulation}
\label{s:enf}

From Corollary~\ref{c:OGS_COGS} and Theorem~\ref{t:adrien},
we can immediately conclude that the semantics on $\lambda$-terms induced by OGS
(Alternating or Concurrent) coincides with 
enf-bisimilarity (i.e., Lassen's trees), as introduced in Definition~\ref{d:enf}. 

In this section, we show a direct proof of the result, for \COGS bisimilarity, 
as an example of application of the  
up-to composition technique for \COGS
 presented in Section~\ref{ss:upto_games}.


Terms, values, and evaluation contexts  of the $\lambda$-calculus can be directly lifted to singleton concurrent configurations,
meaning that we can transform a relation on terms, values, and contexts $\R$ into
a relation $\liftRel{\R}$ on singleton concurrent configurations in the following way:
\begin{itemize}
 \item If $(M_1,M_2) \in \R$ then $(\conf{\pmap{p}{M_1},\emptymap,\phi},
 \conf{\pmap{p}{M_2},\emptymap,\phi}) \in \liftRel{\R}$, 
 with $\phi = \fv{M_1,M_2} \uplus \{p\}$;
 \item If $(V_1,V_2) \in \R$ then $(\conf{\pmap{x}{V_1},\phi},\conf{\pmap{x}{V_2},\phi})
 \in \liftRel{\R}$, with $\phi= \fv{V_1,V_2} \uplus \{x\}$;
 \item If $(E_1,E_2) \in \R$ then $(\conf{\pmap{p}{(E_1,q)},\phi},
  \conf{\pmap{p}{(E_2,q)},\phi}) \in \liftRel{\R}$, with 
  $\phi= \fv{E_1,E_2} \uplus \{p,q\}$.
\end{itemize}

\begin{theorem}
 Taking $(\RR_{\enfM},\RR_{\enfV},\RR_{\enfK})$ an enf-bisimulation, then 
 $(\liftRel{\RR_{\enfM}} \cup \liftRel{\RR_{\enfV}} \cup \liftRel{\RR_{\enfK}})$
 is a singleton bisimulation up-to composition in \COGS (as defined in Section~\ref{ss:upto_games}).
\end{theorem}

\begin{proof}
 Taking $\FF_1,\FF_2$ two singleton configurations s.t. 
 $(\FF_1,\FF_2) \in \liftRel{\RR_{\enfM}}$,
 then we can write $\FF_i$ as $\conf{\pmap{p}{M_i},\phi}$
 for $i=1,2$, such that $(M_1,M_2) \in \R$.
 Suppose that $\FF_1 \Longrightarrowp\arrC\ell 
 \FF'_1 $,
 with $\ell$ a Player action. We only present the Player Question case,
which is where  `up-to composition' is useful.
So writing $\ell$ as $\questP{x}{y,q}$,
$\FF'_1 $ can be written as $\conf{[y \mapsto V_1]\cdot[q \mapsto (E_1,p)]}$,
so that $M_1 \redv^{*} E_1[x V_1]$.
 
As $(M_1,M_2) \in \R$, there are $E_2,V_2$ s.t.
 $M_2 \redv^{*} E_2[x V_2]$, $(V_1,V_2) \in \R$ and $(E_1,E_2) \in \R$.
Hence $\FF_2    
\Longrightarrowp\arrC\ell
 \FF'_2$ with $\FF'_2 = 
 \conf{[y \mapsto V_2]\cdot[q \mapsto (E_2,p)]}$.
 
Finally, $(\conf{\pmap{y}{V_1}},\conf{\pmap{y}{V_2}}) \in \liftRel{\R}$ 
 and $(\conf{\pmap{q}{(E_1,p)}},\conf{\pmap{q}{(E_2,p)}}) \in \liftRel{\R}$,
 so that 
 ${(\FF'_1,\FF'_2)} \in {\uptoComp{\liftRel{\R}}}$.
 
 The case of passive singleton configurations is proved in a similar way.
 \end{proof}

\section{Compositionality of OGS via \piI}
\label{s:compo}
We now present some compositionality results about \COGS and \AOGS, that can be proved
via the
correspondence between OGS and \piI.

Compositionality of OGS amounts to computing
the set of traces generated by
$\conf{M\subst{x}{V}}$ from the set of traces generated by
$\conf{M}$ and $\conf{\pmap{x}{V}}$.
This is the cornerstone of (denotational) game semantics, where the
combination of $\conf{M}$ and $\conf{\pmap{x}{V}}$
is represented via the so-called `parallel composition plus hiding'.
This notion of parallel composition of two processes $P,Q$ plus hiding over a
name $x$ is directly expressible in \piI as the process $\nu x(P | Q)$.
Precisely, 
suppose $\FF,\GG$  are two configurations that agree on their polarity functions,
except on a name $x$;  then
we write 
\vskip .1cm  \mbox{$ $}  \hfill $\nu x(\FF | \GG) $\hfill   $(*)$
\vskip .1cm 
\noindent  for 
the  \piI process $\nu x(\encoConI{\FF} | \encoConI{\GG})$, the 
 \emph{parallel composition plus hiding over $x$}.
To define this operation directly at the level of OGS, we would have
to generalize its LTS, allowing internal interactions over a name $x$
used both in input and in output.


As the translation $\cbvSymb$ from OGS configurations into
 $\piI$
 validates the  $\beta_v$ rule, we can prove  that the behaviour of 
 $\conf{M\subst{x}{V}}$
(e.g., its set of traces)
is the same as that of the parallel composition plus hiding over $x$
of $\conf{M}$ and $\conf{\pmap{x}{V}}$.
We use the notation $(*)$  above to express the following two results. 

\begin{theorem} 
$
  \encoConI{\conf{\pmap{p}{M\subst{x}{V}},\gamma,\phi\cup\phi'}} 
\ {\bowtie}\
 \nu x(
  \encoConI{\conf{\pmap{p}{M},\emptymap,\phi \uplus \{x\}}}
     | \encoConI{\conf{\gamma \cdot\pmap{x}{V},\phi' \uplus \{x\}}})
$
for $ \bowtie \;  \in \{ \wbOPI, \wbPI, \TEN,\TEPI\}$. 
 \end{theorem}  

\begin{proof}
The theorem is a  consequence of the validity of the $\beta_v$ rule; that
is the property $  \encoConI{ (\lambda x. M) V }  \wbPI   \encoConI{  M \sub V x }  $. 
Indeed, by simple algebraic manipulation one can show 
$$   \encoIa{(\lambda x. M) V} p 
 \wbPI  \nu x( \encoIa M p |   \encoIVa {V} x)$$
which is to say
$$ 
\begin{array}{lc}
  \encoConI{\conf{\pmap{p}{M\subst{x}{V}},\emptymap ,\phi\cup\phi'}} 
 \wbPI
 \nu x(
   \encoConI{\conf{\pmap{p}{M},\emptymap,\phi \uplus \{x\}}} 
     \ \ | \encoConI{\conf{[x \mapsto V],\phi' \uplus \{x\}}})
\end{array} 
$$
The result can then be lifted to a  non-empty environment $ \gamma $ by using the congruence
properties of $\wbPI$, and then to the other relations.
\end{proof}

\begin{corollary}
For any trace $\tr$:
 \begin{enumerate}
 \item  $\conf{M\subst{x}{V},p,\gamma,\phi\cup\phi'} \ArrA{\tr}$
 iff $\nu x(
   \conf{M,p,\emptymap,\phi \uplus \{x\}}
      | \conf{\gamma \cdot\pmap{x}{V},\phi' \uplus \{x\}}) \ArrN{\tr}$;
 \item $\conf{\pmap{p}{M\subst{x}{V}},\gamma,\phi\cup\phi'} \ArrC{\tr}$
 iff $\nu x(
   \conf{\pmap{p}{M},\emptymap,\phi \uplus \{x\}}
      | \conf{\gamma \cdot\pmap{x}{V},\phi' \uplus \{x\}}) \ArrPI{\tr}$.
\end{enumerate}
\end{corollary}

\noindent
Other important properties that we can import in OGS from \piI
are the congruence properties for the \AOGS and \COGS semantics. 
We report the result for  $\TEA$; the same result holds for 
$\wbA,\TEp, \wbC$.
\begin{theorem}
 If $\conf{M} \TEA \conf{N}$
 then for any $\LasV$  context $C$, $\conf{C[M]} \TEA \conf{C[N]}$.
\end{theorem}

\noindent
The result is obtained  from the congruence properties of \piI, 
Corollary~\ref{c:OGS_COGS}, and the compositionality
of the encoding $\cbvSymb$.

\iflong
The adequacy and congruence result yields soundness (i.e., inclusion) 
of the OGS relations (Alternating and
Concurrent) \wrt contextual equivalence. Soundness of OGS
is usually a non-trivial property, e.g., requiring 
ciu-equivalence~\cite{laird2007}.

\section{Well-Bracketed OGS}
\label{s:wb-ogs}

In the operational game semantics we have studied so far, Opponent is allowed to 
answer over continuation names in an arbitrary order, and is 
even allowed `to forget' to answer 
to some of these names.
This behavior of Opponent 
corresponds to that of a context that has 
access to control operators, e.g. $\mathtt{call/cc}$, with an affine usage of continuations.

\subsection{Well-Bracketed LTS}

\begin{figure*}
\[\begin{array}{lllll}
 (P\tau) & \conf{M,p,\stt,\gamma,\phi} & \arrWB{\ \tau \ } & 
  \conf{N,p,\stt,\gamma,\phi} & \text{ when } M \red N\\
 (PA) & \conf{V,p,\stt,\gamma,\phi} & \arrWB{\ansP{p}{x}} & 
  \conf{\stt,\gamma \cdot \pmap{x}{V},\phi \uplus \{x\}} \\
 (PQ) & \conf{E[xV],p,\stt,\gamma,\phi} & \arrWB{\questPV{x}{y}{q}} & 
  \multicolumn{2}{c}
    {\conf{q::\stt,\gamma \cdot \pmap{y}{V}\cdot \pmap{q}{(E,p)},\phi \uplus \{y,q\}}} \\
 (OA) & \conf{p::\stt,\gamma\cdot [p \mapsto (E,q)],\phi} & \arrWB{\ansO{p}{x}} & 
  \conf{E[x],q,\stt,\gamma,\phi \uplus \{x\}} \\
 (OQ) & \conf{\stt,\gamma,\phi} & \arrWB{\questOV{x}{y}{p}} & 
  \conf{V y,p,\stt,\gamma,\phi \uplus \{y,p\}} & 
  \text{ when } \gamma(x) = V\\
 (IOQ) & \initconf{\phi}{M} & \arrWB{\questOinit{p}} &
     \conf{M,p,[],\emptymap,\phi \uplus \{p\}}
\end{array}\]
\caption{Well-Bracketed OGS LTS}
\label{fig:stacked-lts}
\end{figure*}

To capture the absence of such control operators in the environment,
we enforce a `well-bracketed' behavior on the sequence of answers.
To do so, in Figure~\ref{fig:stacked-lts} we introduce the  \emph{Well-Bracketed OGS} (\WBOGS).  
It refines \AOGS by using \emph{stacked configurations} $\FFb$.
They use a \emph{stack} of continuation names $\stt$, that keeps track
of the order in which the Opponent answers must occur to be well-bracketed.
Such stacks
are represented as lists of Player continuation names.
We write $p:: \stt$ for the stack with $p$ pushed on top of $\stt$, 
and $\emptystack$ for the empty stack.
In \WBOGS, Player questions push their fresh continuation names on top of the stack, and
Opponent answers pop the stack to obtain the continuation name they are allowed to answer to.

Extending the notion of valid configurations from \AOGS,
we require the stack to be formed by distinct Player continuation names.
Since names in stacks are all distinct, removing the evaluation context 
from $\gamma$ after an Opponent answer is not necessary anymore, since the 
associated continuation name will never be allowed to be used again. 
But for uniformity \wrt the other LTS, we choose to keep it.
For a stacked configuration $\FFb$, we write $\es \FFb$ for the corresponding \AOGS
configuration, where the stack has been removed.

Notice that stacked configuration $\FFb$
only stores the stack of Player continuations as an independent component,
in order to restrict Opponent behavior.
By shuffling the continuation structure of $\gamma$ with
the stack $\sigma$, we could reconstruct
a stack for the whole interaction between Player and Opponent.

\subsection{Well-Bracketed Traces}

The goal is now to characterize the set of traces generated by a stacked configuration
$\FFb$ (on the well-bracketed LTS) as the subset of the traces generated by $\es \FFb$ 
(on the alternating LTS) that satisfies the following condition of \emph{well-bracketing}.

Considering a trace $\tr_1 \ \act_1  \ \tr_2 \ \act_2 \ \tr_3$,
with $\act_2$ an answer of the shape $\ansO{p}{x}$ or $\ansP{p}{x}$
and $\act_1$ the action that introduces the continuation name $p$,
then we say that $\act_1$ is \emph{answered} by $\act_2$.
Since only questions can introduce continuation names, 
$\act_1$ is indeed a question, and it must be of the opposite polarity
as $\act_2$.
Well-bracketing corresponds to the fact that questions must be answered in the last-in
first-out order they previously appear in a trace.
We formalize this idea using a pushdown system, defined as an LTS whose configurations are 
stacks of distinct continuation names with transition relation defined as:
\begin{mathpar}
  (\text{PQ}) \
    {\stt \pshdtrans{\questP{x}{y,p}} p::\stt }
  \and
  (\textsc{OQ}) \
    {\stt \pshdtrans{\questO{x}{y,p}} \stt }
  \and
  (\textsc{PA}) \
    {p::\stt \pshdtrans{\ansP{p}{x}} \stt }
  \and
  (\textsc{OA}) \
    {p::\stt \pshdtrans{\ansO{p}{x}} \stt }
\end{mathpar}

\begin{definition}
Taking a stack $\stt$,
a trace $\tr$ is said to be \emph{well-bracketed} \wrt $\stt$ when
there exists another stack $\stt'$ such that
$\stt \pshdTrans{\tr} \stt'$.
\end{definition}
For example $\ansO{c}{w} \questP{x}{y,c'} \ansO{c'}{z} \ansP{c''}{w}$
is well-bracketed \wrt $c::[]$.
Notice that we do not keep track of the Opponent continuation names in this definition.
Indeed, Proponent behaviour is automatically well-bracketed,
since running terms are taken only in the pure $\lambda$-calculus.

\begin{lemma}
\label{l:WBs_lts}
Taking $\FFb$ a configuration whose stack is $\stt$, then
$\FFb \ArrWB{\tr} \FFb'$ iff 
$\es \FFb \ArrA{\tr} \es \FFb'$ and $\tr$ is well-bracketed \wrt $\stt$.
\end{lemma} 

\begin{proof}
  Suppose that $\FFb \ArrWB{\tr} \FFb'$,
  then writing $\stt'$ for the stack of $\FFb'$,
  we prove by a direct induction on $\tr$ that $\stt \pshdTrans{\tr} \stt'$.

  In the reverse direction, suppose that 
  $\es \FFb \ArrA{\tr} \FF'$ and $\stt \pshdTrans{\tr} \stt'$,
  we consider the stacked configuration $\FFb'$ built 
  from $\FF'$ using $\stt'$.
  Then we prove by a direct induction on $\tr$ that
  $\FFb \ArrWB{\tr} \FFb'$.
\end{proof}

\subsection{Relationship between Alternating and Well-Bracketing OGS}
\label{subsec:relate-alt-wb-lts}

In Section~\ref{subsec:relate-conc-alt-lts},
we have related the denotation of $\lambda$-terms
given by the \COGS and \AOGS LTS.
We now establish a similar connection
between \AOGS and \WBOGS.
This is done by adapting the proof of Lemma~\ref{l:OGS_COGS}.
To do so, we first adapt the definition of
triggerable singleton subconfigurations.

\begin{definition}
A singleton \AOGS configuration $\FF$ is 
a \emph{triggerable} subconfiguration of a stacked configuration $\FFb$ when:
\begin{itemize}
  \item either $\FFb$ is active of the shape
  $\conf{M,p,\stt,\gamma,\phi}$,
  and $\FF = \conf{M,p,\emptymap,\phi \backslash{\dom{\gamma}}}$;
  \item or $\FFb$ is passive of the shape
  $\conf{\stt,\gamma,\phi}$,
  and $\FF$ is:
  \begin{itemize}
    \item either of the shape
    $\conf{\pmap{x}{V},\phi \backslash{\dom{\gamma}} \cup\{x\}}$
    with $\pmap{x}{V} \subseteq \gamma$;
    \item or of the shape $\conf{\pmap{p}{(E,q)},\phi \backslash{\dom{\gamma}} \cup\{p\}}$
    with $\pmap{p}{(E,q)} \subseteq \gamma$,
    when $p$ is the top element of $\stt$.
  \end{itemize}
\end{itemize}
\end{definition}
\noindent
So
$\conf{\pmap{p}{(E,q)},\phi}$ is not triggerable 
when $p$ is not the top element of the stack,
since it cannot be directly triggered by an OA action of the \WBOGS LTS
(while it can be by the \AOGS LTS).

We recall that we only consider well-formed relations, 
as introduced in Definition~\ref{d:compa}.
Therefore, whenever two configurations $\FF,\GG$ appear as
a pair in a relation (i.e., a behavioural relation) 
it is implicitly assumed that they are support-equivalent, written
$\okC{\FF}{\GG}$.
For relations on stacked configurations, $\okC{\FFb}{\GGb}$ 
also enforces that $\FFb,\GGb$ have the same stack.

\begin{lemma}
\label{l:AOGS_WBOGS}
If $\FFb, \GGb$ are singleton stacked configurations and 
$\FFb \wbWB \GGb$, then also $\es\FFb \wbA \es\GGb$. 
\end{lemma}

\begin{proof}
  Let $\R$ be the relation on \AOGS singleton configurations
  defined as $\FF \RR \GG$ when
   $\okC \FF\GG$ and there exist stacked configurations $\FFb, \GGb$
    with $\FFb \wbWB \GGb$ such that 
    $\FF$ and $\GG$ are triggerable singleton configurations
    respectively of $\FFb$ and $\GGb$.

  First notice that if $\FFb, \GGb$ are stacked singleton configurations and 
  $\FFb \wbWB \GGb$, we directly get that 
  $\es\FFb \RR \es\GGb$ since $\es\FFb$ and $\es\GGb$ are 
  triggerable subconfigurations of $\FFb$ and $\GGb$ respectively.
  
  We now show that $\R$ is a bisimulation up-to composition.
  We take $\FF \RR \GG$, such that there exists $\FFb, \GGb$
  as stated by the definition of $\R$.
  Suppose 
  \begin{equation}
  \label{e:arrA-left}
   \FF \arrA{\act} \FF_1 
  \end{equation} 
  in \AOGS.  
  First, we suppose that $\FF,\GG$ are both passive, 
  so that the transition is either OA or OQ.  
  We assume it is OA, since it is the hardest case, 
  Opponent answers being constrained in \WBOGS.
  We thus have, for some $p,V,x,y$ 
  \[\begin{array}{rcl}
   \FF &=&  \conf{ \gamma_1, \phi} \mbox{ with }  \gamma_1 = \pmap{p}{(E,q)} \\
   \act & = & \ansO{p}{x} \\
   \FF_1 & = & \conf{E[x],q, \emptymap,\phi \uplus \{x\}}
  \end{array}\]
  Moreover, we have $\GG = \conf{\delta_1, \phi}$ 
  with $\delta_1 = \pmap{p}{(\Ep,q)}$ for some evaluation context $\Ep$.
  Since $\FF,\GG$ are triggerable respectively from $\FFb,\GGb$,
  we get that both $\FFb,\GGb$ are passive
  and that the top element of the stack of $\FFb$ is $p$.
  From $\FFb \wbWB \GGb$, we get that $\okC{\FFb}{\GGb}$,
  so there exists a stack $\stt$ and two environments $\gamma, \delta$ sharing the same domain such that:  
  \[\begin{array}{rcl}
  \FFb & = & \conf{p::\stt,\gamma_1 \cdot \gamma,\phi \uplus \dom{\gamma}} \\ 
  \GGb & = & \conf{p::\stt,\delta_1 \cdot \delta,\phi \uplus \dom{\delta}} 
  \end{array}\]
  We have 
  $\FFb \arrWB{\act} \conf{E[x],q,\sigma, \gamma, \phi\uplus \dom{\gamma} \uplus \{x\}} 
  \defi \FFb'$. 
  Since $\FFb \wbWB \GGb$, 
  \[\GGb  \arrWB{\act} \conf{E'[x],q,\sigma, \delta, \phi \uplus \dom{\delta} \uplus \{x\}}  
  \defi \GGb'\]
  and $\FFb' \wbWB \GGb'$.
  Let $\GG_1 = \conf{\Ep[x],q,\delta_1,\phi}$,
  we also have 
  $\GG \arrA \act \GG_1$.
  This is sufficient to match \reff{e:arrA-left}, as 
  $\FF_1 \uptoComp \R  \GG_1$. Indeed we have:
  \begin{itemize}  
  \item $\conf{\gamma_1,\phi\uplus \{y,p\}} \R  \conf{\delta_1,\phi\uplus \{y,p\}}$
  since $\FF \RR \GG$;
  \item  $\conf{E[x],q, \phi\uplus \{x\}} \R  \conf{\Ep[x],q, \phi\uplus \{x\}}$
  since they are triggerable subconfigurations respectively of $\FFb'$ and $\GGb'$.
  \end{itemize}
\noindent
  Now suppose that $\FF,\GG$ are active,
  then the proof is similar to the one for Lemma~\ref{l:OGS_COGS}. \qedhere
\end{proof} 

\reviewComment{Lemma 10.4: this statement would be better appreciated with an
example. It seems to me that the point is that the whole "syntax tree"
(enf) is already explored following the wb discipline, because the
P-views are well-bracketed. Is this what is going on?}
\reviewCommentDS{I am not sure i'd like to add an example. 
I would simply add a sentence or two here, confirming the explanation given by the reviewer
}
\reviewCommentGJ{I agree with the beginning of the explanation of the reviewer, that the whole "syntax tree"
(enf) is already explored following the wb discipline. But I'm still trying to understand why the fact that P-views are well-bracketed
is important here. I think it comes from the fact that in the proof, we always work with singleton configurations (thanks to the up-to composition technique), 
which indeed corresponds to only considering the different P-views of a trace.}

\begin{remark}
In the proof of Lemma~\ref{l:AOGS_WBOGS}, we crucially use the fact that the
interaction on singleton configurations is always well-bracketed, as
the stack can have at most one element. Following Remark~\ref{rm:up-to-comp-gs}, this corresponds
to the fact that in game semantics, a P-view is always well-bracketed for Opponent. 
\end{remark}

\subsection{Complete Trace Equivalence}

To relate trace denotation to contextual equivalence, it is important to 
be able to represent terminating interactions between a term and a context.
To do so, 
we consider the notion of \emph{complete trace}, 
that has been introduced in game semantics~\cite{abramsky1997linearity} to  
obtain full-abstraction results for languages with mutable 
store. A trace is said to be complete when all its questions
are answered.
The notion of complete trace is normally considered only for well-bracketed traces,
we generalize it here to more general traces,
and with configurations having already Proponent continuation names.

\begin{definition}
\label{def:ctrace}
 A trace $\tr$ 
 is \emph{complete for an \AOGS (resp. \COGS) configuration $\FF$} if
 $\FF \ArrA{\tr}$ (resp. $\FF \ArrC{\tr}$) 
 and the subject names of the unjustified answers in $\tr$
 are exactly the Player continuation names of $\FF$.
 This definition is extended to a stacked configuration $\FFb$
 with $\stt$ its stack,
 by asking that $\FFb \ArrWB{\tr}$ and $\stt \pshdTrans{\tr} \emptystack$.
 Then we say that a configuration is \emph{terminating} 
 if it has at least one complete trace.
\end{definition}
\noindent
So a trace $\tr$ is complete for $\FF$ when it has no unanswered questions,
and furthermore $\tr$ has also answers to all the pending questions of $\FF$.

\begin{definition}
  \label{d:TEwbcN}
  We write $\FFb_1 \TEwbc \FFb_2$ if
  $\okC{\FFb_1}{\FFb_2}$
  and $\FFb_1,\FFb_2$ have the same set of complete traces in \WBOGS. 
  Similarly, we write $\FF_1 \TEc \FF_2$ (resp.\ $\FF_1 \TEpc \FF_2$)
  if $\okC{\FF_1}{\FF_2}$ 
  and $\FF_1,\FF_2$ have the same 
  sets of complete traces in \AOGS (resp.\ \COGS).
\end{definition} 

We now characterize configurations that are reached
by complete traces.

\begin{definition}
  \label{d:sp}
  An \AOGS or stacked configuration
  is \emph{strongly passive} if it is passive and its 
  continuation structure is empty.
  For \COGS, we also ask that its running term be empty.
\end{definition}
\noindent
This means, for stacked configurations,
that the stack is empty, and the configuration is of form
\[   
    \conf{\emptystack, 
      \pmap{x_1}{V_1}\cdot \ldots \cdot  \pmap{x_n}{V_n},\phi} \hskip .8cm n\geq 0 \]
so that indeed we can consider the interaction terminated, since Opponent
cannot answer anymore, even if it can still interrogate one of the $x_i$.

Traces that correspond to reduction to strongly passive configurations 
are complete traces, a result that holds for \WBOGS, \AOGS and \COGS.
  
\begin{lemma}
\label{l:WBOGS_complete-strongly-passive}
  Taking $\FFb,\FFb'$ two stacked configurations 
  and $\tr$ a trace such that $\FFb \ArrWB \tr \FFb'$.  
  Then $\tr$ is complete for $\FFb$
  iff $\FFb'$ is strongly passive.
\end{lemma}

\begin{proof}
  Writing $\stt,\stt'$ for the stack respectively of $\FFb,\FFb'$,
  we get from Lemma~\ref{l:WBs_lts} that
  $\stt \pshdTrans{\tr} \stt'$.
  By definition, $\tr$ is complete for $\FFb$ iff $\stt'$ is empty,
  and $\FFb'$ is strongly passive iff $\stt'$ is empty.
\end{proof}

\begin{lemma}
  \label{l:OGS_complete-strongly-passive}
    Taking $\FF,\FF'$ two \AOGS (resp. \COGS) configurations
    and $\tr$ a trace  s.t. $\FF \ArrA \tr \FF'$
    (resp. $\FF \ArrC \tr \FF'$).  
    Then $\tr$ is complete for $\FF$
    iff $\FF'$ is strongly passive. 
\end{lemma} 

\begin{proof}
  Suppose that $\tr$ is complete for $\FF$.
  All Player questions of $\tr$ having been answered during the interaction, 
  the set of Player continuation names of $\FF'$
  is included in the one of $\FF$.
  Moreover, $\tr$ contains answers for all Player continuation names of $\FF$
  so that no Player continuation name of $\FF'$ is a Player continuation name of $\FF$.
  Therefore, its set of continuation names is empty,
  so that $\FF'$ is strongly passive.

  Reciprocally, if $\FF'$ is passive, then its set of continuation names is empty.
  Thus, $\tr$ cannot have pending questions and must
  have answered all the questions of $\FF$.
  So $\tr$ is complete for $\FF$.
\end{proof}
\noindent
Complete well-bracketed trace equivalence corresponds to contextual equivalence for a 
$\lambda$-calculus with store, as proved in~\cite{abramsky1997linearity} for 
a compositional game model of Idealized Algol,
and in~\cite{laird2007} for an operational game model of RefML.
\cutLMCSsecondRound{
This last result can be adapted to the $\lambdarho$-calculus, 
an untyped call-by-value language with store. Its definition is given in 
Appendix~\ref{a:LambdaRho}, with the definition of 
its contextual equivalence $\ctxeqlr$.
Being a language where the exchange of locations is forbidden,
the proof of full abstraction of the \WBOGS semantics
for $\lambdarho$-calculus can be obtained by a direct adaptation of the result proven in~\cite{JM21}
for the `HOS' language.

\begin{theorem}
\label{t:rho}
Let $M_1,M_2 \in \LasV$
 with free variables in $\phi$.
We have $\initconf{\phi}{M} \TEwbc \initconf{\phi}{N}$
 iff  $M \ctxeqlr N$.
\end{theorem}}

\section{Tensor Product of OGS configurations}
\label{sec:tensor}

Following the introduction of the tensor product of configurations
for C-OGS in Definition~\ref{d:tensor_cogs}, we 
introduce a similar definition for the Alternating and Well-Bracketed LTS.
We also prove results relating the 
(complete) traces generated by $\FF \otimes \GG$
with the interleaving of traces generated by $\FF$ and $\GG$.

\subsection{Interleaving of complete traces for \COGS}

We first state the variant of Lemma~\ref{l:COGS-interleave-tensor}
for complete traces.
So given a complete trace for a configuration, the projections of the
trace on sub-configurations are also complete, as proven in the following lemma.

\begin{lemma}
  \label{l:COGS-decomp-ct}
  Suppose 
  \begin{itemize}
  \item $\FF = \appendCon{\FF_1}{\FF_2}$;
  \item $\FF \ArrC \tr \FF'$ and $ \tr $ is a complete trace
  for $\FF$.
  \end{itemize} 
  Then for both $i \in \{1,2\}$,
  there exists $\FF_i'$ and a complete trace $\tr_i$
  for $\FF_i$ such that
  \begin{enumerate}
  \item ${\FF_i} \ArrC{\tr_i} \FF_i'$;
  \item $\tr \in \interl{\tr_1}{\tr_2}$;
  \item $\FF' = \appendCon{\FF'_1} {\FF'_2}$.
  \end{enumerate}
\end{lemma}              
  
\begin{proof}
  From Lemma~\ref{l:COGS-interleave-tensor}(\ref{l:COGS-interleave-tensor_1}) 
  we get the existence of configurations $\FF_i'$ and traces $\tr_i$ 
  such that the three items above are verified.
  From Lemma~\ref{l:OGS_complete-strongly-passive}, we get that $\FF'$
  is strongly passive, so that both $\FF'_1,\FF'_2$ are strongly passive as well.
  Thus applying Lemma~\ref{l:OGS_complete-strongly-passive}
  in the other direction,
  we get that both $\tr_1,\tr_2$ are complete.
\end{proof} 

\begin{lemma}
  \label{l:COGS-recomp-ct}
  Suppose 
  \begin{enumerate}
  \item $\FF = \FF_1 \otimes \FF_2$;   
  \item $\FF_i \ArrC{\tr_i} \FF_i'$ where 
  $\tr_i$ is a complete trace for $\FF_i$, 
   $i=1,2$. 
  \end{enumerate} 
  Then there is a trace 
   $\tr \in \interl{\tr_1}{\tr_2}$ complete for $\FF$ such that
  $\FF \ArrC{\tr}{\FF'_1} \otimes {\FF'_2}$. 
\end{lemma}

\begin{proof}
  The existence of $\tr$ such that $\FF \ArrC{\tr}{\FF'_1} \otimes {\FF'_2}$ follows from 
  Lemma~\ref{l:COGS-interleave-tensor}(\ref{l:COGS-interleave-tensor_2}). 
  Since $\tr$ is the interleaving of two complete traces respectively from 
  $\FF_1$ and $\FF_2$, we get that it is complete for $\FF = \FF_1 \otimes \FF_2$.
\end{proof} 

\subsection{\AOGS}
\label{a:tensor-aogs}

\begin{definition}
Taking $\FF_1,\FF_2$ two compatible \AOGS configurations, with 
 $\gamma_i,\phi_i$ their respective environment and set of names,
$\appendCon{\FF_1} {\FF_2}$ is defined as:
\begin{itemize}
 \item $\conf{M,p,\gamma_1 \cdot \gamma_2,\phi_1 \cup \phi_2}$
 when one of $\FF_1,\FF_2$  is active, with $M,p$ its toplevel term and continuation name;
 \item $\conf{\gamma_1 \cdot \gamma_2,\phi_1 \cup \phi_2}$
 when $\FF_1,\FF_2$  are passive.
\end{itemize}
\end{definition}
\noindent
We can then state the adaptation of Lemma~\ref{l:COGS-interleave-tensor} to \AOGS.
We write $\interlA{\tr_1}{\tr_2}$ for the subset of $\interl{\tr_1}{\tr_2}$
formed by \emph{alternating} traces, with the
extra condition that if one of $\tr_1,\tr_2$
starts with an output, then the resulting interleaved trace must also start with an output.

\begin{lemma}
\label{c:dec_GS}
Suppose $\FF_1,\FF_2$ are compatible \AOGS configurations, and at most one of the two is active.
Then the set of traces generated by $\FF_1 \otimes \FF_2$
is equal to the set of interleavings $\interlA{\tr_1}{\tr_2}$,
with $\FF_1 \ArrA{\tr_1}$ and $\FF_2 \ArrA{\tr_2}$.
That is:
\begin{enumerate}
  \item\label{l:dec_GS_1} if $\appendCon{\FF_1}{\FF_2} \ArrA{\tr} \FF$, 
  then, for $i=1,2$,  there are traces $ \tr_i$
  such that $ {\FF_i} \ArrA{\tr_i} \FF'_i$, 
  and $\tr  \in \interlA{ \tr_1}{\tr_2}$, 
  and $\FF = \appendCon{\FF'_1}{\FF'_2}$;
  \item\label{l:dec_GS_2} conversely, if for $i=1,2$, we have
  $\FF_i \ArrA{ \tr_i} \FF'_i$,  
  and $\tr  \in \interlA {\tr_1}{\tr_2}$, 
  then $\appendCon{\FF_1}{\FF_2} \ArrA{\tr} 
  \appendCon{\FF'_1} {\FF'_2}$.
\end{enumerate}
\end{lemma}  
\noindent
The statement of (\ref{l:dec_GS_2}) above can be deduced from
the corresponding result for \COGS, namely
Lemma~\ref{l:COGS-interleave-tensor},
combined with the characterization
of traces generated by the \AOGS LTS 
as being the \emph{alternating} traces of \COGS LTS.
However, such a proof does not work for (\ref{l:dec_GS_1}),
since we cannot deduce directly that $\tr_1$ and $\tr_2$ are alternating.
To prove this direction,
we use the correspondence between \AOGS and $\ArrN{}$ the \outputpr LTS 
established in Corollary~\ref{c:traces}, 
combined with Lemma~\ref{l:dec_piI} below, whose proof uses  the following lemma.

\begin{lemma}
\label{l:dec_piI_aux}
Suppose $P| Q  \ArrN \mu P'|Q$ and 
$P  \Arr \mu P'$. Then also $P  \ArrN \mu P'$.
\end{lemma} 
\begin{proof}
If $P  \ArrN \mu P'$ were not possible, because of the constraint on input components in
the \opLTS with respect to the ordinary LTS, then also 
$P| Q  \ArrN \mu P'|Q$ would not be possible (if $P$ is not input reactive,  then 
$P|Q$ cannot be so either). 
\end{proof}

\begin{lemma}
\label{l:dec_piI}
Suppose $\FF_1,\FF_2$ are compatible configurations, and at most one of the two is active.
\begin{enumerate}
\item If $\encoConV {\FF_1} |\encoConV {\FF_2} \ArrN \tr  P_1 | P_2  $, 
  then, for $i=1,2$,  there are traces $ \tr_i $
  such that $\encoConV {\FF_i} \ArrN { \tr_i} P_i $, 
  and $ \tr \in \interlA { \tr_1   }{ \tr_2 }$;
\item
Conversely, if for  $i=1,2$, we have
  $\encoConV {\FF_i} \ArrN { \tr_i} P_i$,
  $\tr \in \interlA { \tr_1   }{ \tr_2 }$, 
  then also 
  $\encoConV {\FF_1} |\encoConV {\FF_2} \ArrN \tr  P_1 | P_2$.
\end{enumerate} 
\end{lemma}  

\begin{proof}
Proof of (1): 
the assertion follows by reasoning as in Lemma~\ref{l:dec_pi}.
Processes 
$\encoConV {\FF_1}$ and  $\encoConV {\FF_2}$ cannot interact and, by
Lemma~\ref{l:ci_trans}, the property is invariant under transitions. 
Thus each action (visible
or silent) that contributes to the given trace $\tr$ is derived using either  
rule {\trans{parL}} or  rule {\trans{parR}} on the outermost parallel composition. 
The former case represents  an action (in the ordinary LTS) from the left process, the
latter case an action from the right process. In either case, the action can be lifted to 
the  \opLTS using Lemma~\ref{l:dec_piI_aux}. 
By iterating the reasoning on each action of the trace  $\tr$
we obtain traces $\tr_i$ ($i=1,2$) with  
 $\encoConV {\FF_i} \ArrN { \tr_i} P_i$. 

For (2): From $\encoConV {\FF_i} \ArrN { \tr_i} P_i$, we deduce that
$\encoConV {\FF_i} \ArrPI{ \tr_i} P_i$.
So applying Lemma~\ref{l:dec_pi}, we get that 
$\encoConV {\FF_1} |\encoConV {\FF_2} \ArrPI \tr  P_1 | P_2$.
Moreover, from $\encoConV {\FF_i} \ArrN { \tr_i} P_i$ we also deduce
that each $\tr_i$ is alternating. Since at most one of the two
$\FF_1,\FF_2$ is active, we deduce that at most one of $s_1,s_2$ 
is $P$-starting. We have that $\tr$ is alternating since
$\tr \in \interlA{ \tr_1 }{ \tr_2 }$, and 
active iff one of the composing traces is $P$-starting. This means that we can assume 
that when an input transition of $\tr$ is performed the source process is 
input reactive.
\end{proof} 

\noindent
We now state the corresponding lemmas for complete traces,
whose proofs are similar to Lemma~\ref{l:COGS-decomp-ct} and Lemma~\ref{l:COGS-recomp-ct}.
\begin{lemma}
\label{l:AOGS-decomp-ct}
Suppose
\begin{itemize}
  \item $\FF  = \appendCon{\FF_1}{\FF_2}$;
  \item $\tr$ is a complete trace for $\FF$,
  such that $\FF \ArrA \tr \FF'$.
\end{itemize} 
Then for both $i \in \{1,2\}$,
there exists $\FF_i'$ and a complete trace $\tr_i$
for $\FF_i$ such that 
\begin{enumerate}
  \item ${\FF_i} \ArrA{\tr_i} \FF_i'$;
  \item $\tr \in \interlA{\tr_1}{\tr_2}$; 
  \item $\FF' = \appendCon{\FF'_1}{\FF'_2}$.
\end{enumerate}
\end{lemma}              

\begin{lemma}
  \label{l:AOGS-recomp-ct}
  Suppose 
  \begin{enumerate}
    \item $\FF = \FF_1 \otimes \FF_2$;   
    \item $\tr_i$ a complete trace for $\FF_i$,
    such that $\FF_i \ArrA{\tr_i} \FF_i'$,
    $i=1,2$. 
  \end{enumerate} 
  Then there is 
   $\tr \in \interlA{\tr_1}{\tr_2}$ complete for $\FF$ such that
  $\FF \ArrA{\tr}{\FF'_1} \otimes {\FF'_2}$. 
\end{lemma}
\noindent
Next, we show in Lemma~\ref{l:decACN} a decomposition property for 
complete trace equivalence $\TEc$ \wrt $\otimes$ on configurations.
%
To prove it, we rely on the following lemmas
that show that the action of a configuration is 
determined by a specific singleton component, 
and it is the action itself that allows us to determine such a component.

\begin{lemma}
\label{l:decNsig}
Suppose $\FF \arrA\mu \FF'$.
Then there is a singleton $\FF_1$, and some $\FF_1',\FF_2,\FF_2'$ 
such that:
\begin{enumerate}
\item $\FF =
 {\FF_1} \otimes {\FF_2}$;
\item $\FF_1 \arrA\mu \FF_1'$;
\item $\FF' =
 {\FF_1'} \otimes {\FF_2}$.
\end{enumerate}
\noindent
Moreover, the P-\Supports and the structure of continuations of $\FF_1$ and $\FF_1'$  
are determined by 
$\mu$; that is, 
if  $\okC \FF\GG$ and
 $\GG \arrA\mu \GG'$, then for  
 the singleton $\GG_1$, and the configurations $\GG_1',\GG_2,\GG_2'$ such that
 $\GG = \appendCon{\GG_1}{\GG_2}$,
 $\GG_1 \arrA\mu \GG_1'$, and 
 $\GG' = \appendCon{\GG_1'}{\GG_2}$ we have: 
 $\okC{\FF_1}{\GG_1}$, $\okC{\FF_2}{\GG_2}$, and 
$\okC{\FF_1'}{\GG_1'}$.  
 \end{lemma} 

\begin{proof}
By a simple case analysis on the form of $\mu$.
\end{proof} 
\noindent
The next lemma uses the previous result to show that the decomposition of a trace 
is uniform.

\begin{lemma}
\label{l:decNsig2}
Suppose 
\begin{enumerate}
\item
$\FF \ArrA \tr   \FF'$
and $\GG \ArrA \tr   \GG'$;
\item  
 $\FF = \FF_1 \otimes \FF_2$ and 
 $\GG = \GG_1 \otimes \GG_2$;  
\item 
  $\okC{\FF}{\GG}$ and  
  $\okC{\FF_i}{\GG_i}$, for $i=1,2$;   
\item 
$\FF_i \ArrA{\tr_i}  \FF_i'$, $i=1,2$,  with $\tr \in \interlA{\tr_1}{\tr_2}$ and 
 $\FF' = \FF'_1 \otimes \FF'_2$.
\end{enumerate} 
Then we also have 
$\GG_i \ArrA{\tr_i}   \GG_i'$, for some $\GG_i'$, 
$i=1,2$ with 
 $\GG' = \GG'_1 \otimes \GG'_2$.
\end{lemma}

\begin{proof}
  Using Lemma~\ref{l:decNsig}, and proceeding by induction on the length of the trace.  
\end{proof}

\begin{lemma}
\label{l:decACN}
 Suppose 
 \begin{enumerate}
   \item $\FF$ is terminating;
   \item 
   $\FF  \TEc \GG$; 
   \item 
   $\FF = \FF_1 \otimes \FF_2$,   
   and  
   $\GG = \GG_1 \otimes \GG_2$;
   \item $\okC{\FF_1}{\FF_2}$ and $\okC{\GG_1}{\GG_2}$.
 \end{enumerate} 
 Then
 $\FF_i \TEc \GG_i$, $i=1,2$. 
\end{lemma}  

\begin{proof}
  From the fact that $\FF$ is terminating and 
  $\FF \TEc \GG$, we deduce the existence of 
  a complete trace $t$ for both $\FF$ and $\GG$.
  Applying Lemma~\ref{l:AOGS-decomp-ct},
  we deduce for both $i=1,2$ the existence of a complete trace $t_i$
  for both $\FF_i$ and $\GG_i$. 
 
  Taking $j=1,2$ and $\tr_j$ a complete trace for $\FF_j$,
  from Lemma~\ref{l:AOGS-recomp-ct}
  we deduce the existence of a complete trace $\tr$ for $\FF$, with:
  \begin{itemize}
    \item either $\tr \in \interlA{\tr_1}{t_2}$ if $j=1$;
    \item or $\tr \in \interlA{t_1}{\tr_2}$ if $j=2$.
  \end{itemize}
  From $\FF  \TEc \GG$, we deduce that $\tr$ is a complete trace for $\GG$,
  so that applying Lemma~\ref{l:decNsig2},
  we get that $\tr_j$ is a complete trace for $\GG_j$.
\end{proof}
\noindent
As explained in Remark~\ref{r:CONFnoSINGL}, 
this previous lemma fails for ordinary (not necessarily complete) trace equivalence.
What makes it true for complete trace equivalence is the additional hypothesis
that $\FF$ is terminating.
In the other direction, a form of compositionality  
for $\TEc$ \wrt $\otimes$ also holds.

\begin{lemma}
\label{l:compCN}
Suppose 
\begin{enumerate}
\item 
$\FF = \appendCon{\FF_1}{\FF_2}$,   
 and  
$\GG = \appendCon{\GG_1}{\GG_2}$;
\item 
${\FF_i}  \TEc {\GG_i}$, $i=1,2$. 
\end{enumerate} 
Then 
${\FF} \TEc {\GG}$.
\end{lemma}

\begin{proof}
First note that $\okC{\FF}{\GG}$ holds, 
as ${\FF_i}  \TEc {\GG_i}$, $i=1,2$ and therefore 
$\okC{{\FF_i}}{{\GG_i}}$. 

Suppose $\FF \ArrA \tr \FF'$ where $\tr$ is a complete trace for $\FF$. 
By Lemma~\ref{l:AOGS-decomp-ct}, we have 
${\FF_i} \ArrA{\tr_i} \FF_i'$,
where $ \tr_i$ is a complete trace, $i=1,2$,
with $\tr \in \interlA{\tr_1}{\tr_2} $ and  
$\FF' = \appendCon{\FF'_1}{\FF'_2}$.
Since ${\FF_i}  \TEc {\GG_i}$, $i=1,2$, we also have  
${\GG_i} \ArrA{\tr_i}   \GG_i'$. 
By Lemma~\ref{c:dec_GS}(\ref{l:dec_GS_2}) we infer that 
$\GG \ArrA{\tr}$. Moreover, $\tr$ is complete for $\GG$,
since $\GG$ has the same continuation structure as $\FF$.
\end{proof}  

\subsection{\WBOGS}
\label{a:tensor-wbogs}

We now introduce the tensor product of stacked configurations.
Compared to the tensor product of \COGS or \AOGS configurations,
there are multiple ways in \WBOGS
to tensor two stacked configurations, so that the definition of the 
tensor product $\otimes_{\stt}$ has to be 
indexed by a stack of continuation names $\stt$.
Indeed, when tensoring two stacked configurations, we can
choose how to interleave their stack.

First, we adapt the notion of compatible configuration
to stacked ones.
\begin{definition}
 Two stacked configurations $\FFb_1,\FFb_2$ are compatible 
 when their stacks are disjoint and 
 their corresponding \AOGS
 configurations $\es{\FFb_1},\es{\FFb_2}$ are compatible.
\end{definition}

\begin{definition}
Let $\FFb_1$ and $\FFb_2$ be two compatible stacked configurations
with at most one active,
and $\stt_1,\stt_2$ their respective stacks.  
Taking $\stt$ an interleaving of $\stt_1,\stt_2$,
we write $\FFb_1 \otimes_{\stt} \FFb_2$ for the following 
stacked configuration:
\begin{itemize}
 \item $\conf{M,p,\stt,\gamma_1 \cdot \gamma_2,\phi_1 \cup \phi_2}$
 when one of the two configurations is active, with $M,p$
 its toplevel term and continuation name;
 \item $\conf{\stt,\gamma_1 \cdot \gamma_2,\phi_1 \cup \phi_2}$
 when the two configurations are passive.
\end{itemize}
with $\gamma_i,\phi_i$ the environment and the set of names in $\FFb_i$.
\end{definition}
\noindent
Taking $\tr_1,\tr_2$ two traces well-bracketed 
\wrt respectively $\stt_1$ and $\stt_2$,
and $\stt$ a valid interleaving of 
$\stt_1,\stt_2$;
we write $\interlWBs{\stt}{\tr_1}{\tr_2}$ for the subset of $\interlA{\tr_1}{\tr_2}$
formed by well-bracketed traces \wrt $\stt$.

\begin{lemma}
  \label{c:lift-atrace-wbtrace}
  Suppose $\FFb_1,\FFb_2$ are compatible stacked configurations
  and $\stt$ is an interleaving of their respective stack,
  such that $\FFb_1 \otimes_{\stt} \FFb_2 \ArrWB{\tr} \GGb$.
  Taking $\tr_1,\tr_2$ two traces, $\tr \in \interlWBs{\stt}{\tr_1}{\tr_2}$,
  and $\GG_1,\GG_2$ two compatible \AOGS configurations
  we suppose that $\es{\FFb_1} \ArrA{\tr_1} \GG_1$
  and $\es{\FFb_2} \ArrA{\tr_2} \GG_2$.
  Then there exists $\GGb_1,\GGb_2$ two compatible stacked configurations
  and $\stt'$ an interleaving of their respective stack such that
  $\GGb = \GGb_1 \otimes_{\stt'} \GGb_2$
  and for all $i=1,2$, 
  $\es{\GGb_i} = \GG_i$
  and 
  $\FFb_i \ArrWB{\tr_i} \GGb_i$.
\end{lemma}

\begin{proof}
We reason by induction on $\tr$.
If $\tr$ is empty, then $\GGb = \FFb_1 \otimes_{\stt} \FFb_2$.
Then both $\tr_1,\tr_2$ are empty as well,
so that $\GG_1 = \es{\FFb_1}$ and $\GG_2 = \es{\FFb_2}$.
So we define $\GGb_1$ as $\FFb_1$ and $\GGb_2$ as $\GGb_2$.

If $\tr$ is of the shape $\act \tracerho$,
then either $\act$ is the first action of $\tr_1$, or the first action of $\tr_2$.
Suppose it is the first action of $\tr_1$, then $\tr_1$
can be written as $\act \tracerho_1$.
Then there exists $\FF'_1$ such that 
$\es{\FFb_1} \arrA{\act} \FF'_1 \ArrA{\tracerho_1} \GG_1$.
Moreover, there exists $\FFb'$
such that $\FFb_1 \otimes_{\stt} \FFb_2 \arrWB{\act} \FFb' \ArrWB{\tr} \GGb$.

We prove by a case analysis over $\act$
that there exists $\FFb'_1$ and $\stt''$
such that $\FFb' = \FFb'_1 \otimes_{\stt''} \FFb_2$,
$\es{\FFb'_1} = \FF'_1$
and that $\tracerho \in \interlWBs{\stt''}{\tracerho_1}{\tr_2}$.
\begin{itemize}
  \item If $\act$ is an Opponent answer
  $\ansO{p}{x}$ then $\FFb_1$ is of the shape
  $\conf{\stt_1,\gamma\cdot [p \mapsto (E,q)],\phi}$,
  and $\stt$ is of the shape $p::\stt''$.
  So $\stt_1$ can be written as $p::\stt'_1$
  since $\stt$ is an interleaving of $\stt_1$ and $\stt_2$,
  and that $p$ must belong to $\stt_1$.
  Writing the stack of $\FFb_2$
  as $\stt_2$, we get that
  $\stt''$ is an interleaving of $\stt'_1$ and $\stt_2$.
  Defining $\FFb'_1$ as 
  $\conf{\stt'_1,\gamma,\phi \cup \{x\}}$,
  we then get the wanted properties.
  \item If $\act$ is a Player question
  $\questPV{x}{y}{q}$
  then $\FFb_1$ is of the shape
  $\conf{E[xV],p,\stt_1,\gamma,\phi}$,
  so that $\FF'_1$ is equal to
  $\conf{\gamma \cdot \pmap{y}{V}\cdot \pmap{q}{(E,p)},\phi \cup \{y,q\}}$.
  Moreover, the stack of $\FFb'$ is equal to
  $q::\stt$.
  Defining $\FFb'_1$ as 
  $\conf{q::\stt_1,\gamma \cdot \pmap{y}{V}\cdot \pmap{q}{(E,p)},\phi \cup \{x\}}$
  and $\stt''$ as $q::\stt$,
  we then get the wanted properties.
  \item If $\act$ is an Opponent question or a Player answer,
  then we can simply take $\stt'$ equal to $\stt$,
  and $\FFb'_1$ as $\FF'_1$ equipped with the same stack
  as $\FFb_1$.
\end{itemize}
\noindent
We then conclude using the induction hypothesis over $\tracerho$.
A similar reasoning applies when $\act$ is the first action of $\tr_2$.
\end{proof}

\begin{lemma}
\label{c:dec_GS_WB}
Suppose $\FFb_1,\FFb_2$ are compatible stacked configurations
and $\stt$ is an interleaving of their respective stack.
Then the set of traces generated by $\FFb_1 \otimes_{\stt} \FFb_2$
is equal to the set of interleavings 
$\interlWBs{\stt}{\tr_1}{\tr_2}$ well-bracketed \wrt $\stt$,
with $\FFb_1 \ArrWB{\tr_1}$ and $\FFb_2 \ArrWB{\tr_2}$.
Moreover, $\tr$ is complete if and only if $\tr_1,\tr_2$ are.
That is:
\begin{enumerate}
  \item If $\FFb_1 \otimes_{\stt} \FFb_2 \ArrWB \tr  \GGb$, then, 
  for $i=1,2$,  there are traces $\tr_i$
  such that ${\FFb_i}\ArrWB{\tr_i} \GGb_i$, 
  and $\tr \in \interlWBs{\stt}{\tr_1}{\tr_2}$, 
  with $\GGb = \GGb_1 \otimes_{\stt'} \GGb_2$;
  and if $\tr$ is complete for $\FFb$ then both $\tr_i$ are complete for $\FFb_i$;
  \item Conversely, if for  $i=1,2$, we have
  ${\FFb_i} \ArrWB {\tr_i} \GGb_i$,  
  and $\tr \in \interlWBs{\stt}{\tr_1}{\tr_2}$, 
  then $\FFb_1 \otimes_{\stt} \FFb_2 \ArrWB{\tr} 
  \GGb_1 \otimes_{\stt'} \GGb_2$, with $\stt'$
  an interleaving of the stack of $\GGb_1$ and $\GGb_2$;
  and if both $\tr_i$ are complete for $\FFb_i$ then $\tr$ is complete for $\FFb$.
\end{enumerate} 
\end{lemma}

\begin{proof}
  \begin{enumerate}
    \item Lifting the reduction to \AOGS we get that
    $\es{\FFb_1} \otimes_{\stt} \es{\FFb_2} \ArrA{\tr} \es{\GGb}$.
    Applying Lemma~\ref{c:dec_GS} we deduce that there exists
    $\GG_1,\GG_2$ two compatible \AOGS configurations
    such that $\es{\FFb_1} \ArrA{\tr_1} \GG_1$
    and $\es{\FFb_2} \ArrA{\tr_2} \GG_2$.
    We conclude using Lemma~\ref{c:lift-atrace-wbtrace}.
    \item Lifting the reductions to \AOGS we get that
    ${\es{\FFb_i}} \ArrA {\tr_i} \es{\GGb_i}$, $i=1,2$.
    Applying Lemma~\ref{c:dec_GS} we deduce that 
    $\es{\FFb_1} \otimes \es{\FFb_2} \ArrA{\tr} \es{\GGb_1} \otimes \es{\GGb_2}$.
    Since $\tr$ is well-bracketed,
    from Lemma~\ref{l:WBs_lts}
    we get that
    $\FFb_1 \otimes_{\stt} \FFb_2 \ArrWB{\tr} 
     \GGb_1 \otimes_{\stt'} \GGb_2$
    for $\stt'$
    an interleaving of the stack of $\GGb_1$ and $\GGb_2$.

\noindent
    If both $\tr_i$ are complete, then from Lemma~\ref{l:WBOGS_complete-strongly-passive},
    both $\GGb_i$ are strongly passive, so that $\GGb$ is also strongly passive.
    Thus applying Lemma~\ref{l:WBOGS_complete-strongly-passive} in the other direction, 
    we deduce that $\tr$ is complete for $\FFb$.
    \qedhere
  \end{enumerate}
\end{proof}
\noindent
Finally, we adapt Lemma~\ref{l:decACN} to \WBOGS.

\ifcutLMCS
Lemma~\ref{l:decNsig-WB} shows that the action of a configuration is 
determined by a specific singleton component, 
and it is the action itself that allows us to determine such a component.

\begin{lemma}
\label{l:decNsig-WB}
Suppose $\FFb \arrWB\mu \FFb'$.
Then there is a singleton $\FFb_1$, and some $\FFb_1',\FFb_2,\FFb_2'$ such that:

\begin{enumerate}
\item $\FFb =
 {\FFb_1} \otimes_{\stt} {\FFb_2}$,
\item $\FFb_1 \arrWB\mu \FFb_1'$,
\item $\FFb' =
 {\FFb_1'} \otimes_{\stt'} {\FFb_2}$.
\end{enumerate} 
 
Moreover, the P-\Supports and the structure of continuations of $\FFb_1$ and $\FFb_1'$  
are determined by 
$\mu$; that is, 
if  $\okC \FFb\GGb$ and
 $\GGb \arrWB\mu \GGb'$, then for  
 the singleton $\GGb_1$, and the configurations $\GGb_1',\GGb_2,\GGb_2'$ such that
 $\GGb = \appendCon
 {\GGb_1} {\GGb_2}$,
 $\GGb_1 \arrWB\mu \GGb_1'$, and 
 $\GGb' = \appendCon
 {\GGb_1'} {\GGb_2}$ we have: 
 $\okC{\FFb_1}{\GGb_1}$, $\okC{\FFb_2}{\GGb_2}$, and 
$\okC{\FFb_1'}{\GGb_1'}$.  
 \end{lemma}

\begin{proof}
The first part is proved using Lemma~\ref{c:dec_GS_WB}. 
The second part is a simple case analysis on the form of $\mu$.
\end{proof} 

The next lemma uses the previous result to show that the decomposition of a trace 
is uniform.

\begin{lemma}
\label{l:decNsig2-WB}
Suppose 
\begin{enumerate}
\item
$\FFb \ArrWB \tr   \FFb'$
and $\GGb \ArrWB \tr   \GGb'$, 
\item  
 $\FFb = \FFb_1 \otimes_\stt \FFb_2$ and 
 $\GGb = \GGb_1 \otimes_\stt \GGb_2$,  
\item 
  $\okC \FFb\GGb$ and  
  $\okC {\FFb_i}{\GGb_i}$, for $i=1,2$,   
\item 
$\FFb_i \ArrWB{\tr_i}  \FFb_i'$, $i=1,2$,  with $\tr \in \interlWBs{\stt}{\tr_1}{\tr_2}$ and 
 $\FFb' = \FFb'_1 \otimes_{\stt'} \FFb'_2$.
\end{enumerate} 
Then we also have 
$\GGb_i \ArrWB{\tr_i}   \GGb_i'$, for some $\GGb_i'$, 
$i=1,2$ with 
 $\GGb' = \GGb'_1 \otimes_\stt \GGb'_2$.
\end{lemma}

\begin{proof}
Using Lemma~\ref{l:decNsig-WB}, and proceeding by induction on the length of the trace.  
\end{proof} 
\fi

\begin{lemma}
\label{l:decWBCN}
  Suppose 
  \begin{enumerate}
    \item $\FFb$ is terminating;
    \item $\FFb  \TEwbc \GGb$; 
    \item $\FFb = \FFb_1 \otimes_{\stt} \FFb_2$,   
    and  
    $\GGb = \GGb_1 \otimes_{\stt} \GGb_2$;
    \item $\okC{\FFb_1}{\FFb_2}$ and $\okC{\GGb_1}{\GGb_2}$.
  \end{enumerate} 
  Then
  $\FFb_i  \TEwbc \GGb_i$, $i=1,2$. 
\end{lemma}   
\noindent
Its proof is a direct adaption of the proof of Lemma~\ref{l:decACN},
using Lemmas corresponding to Lemma~\ref{l:decNsig} and Lemma~\ref{l:decNsig2}
for \WBOGS.

\ifcutLMCS

\begin{proof}
  We write $\stt_1$ for the stack of both $\FFb_1, \GGb_1$
  and $\stt_2$ for the one of $\FFb_2, \GGb_2$.
  Suppose $\FFb_1  \TEwbc \GGb_1$ 
  does not hold (the  case of $\FFb_2$ and $\GGb_2$  is similar).
  Thus there is $\tracerho$ a well-bracketed complete trace \wrt $\stt_1$, for $\FFb_1$ and not for $\GGb_1$. 
  
  Since $\FFb$ is terminating
  and $\FFb  \TEwbc \GGb$, 
  there is for both $\FFb$ and $\GGb$ 
  a well-bracketed complete trace
  $\tr$ \wrt $\stt$
  the stack of $\FFb$ and $\GGb$.
  
  By Lemma~\ref{c:dec_GS_WB}, there are traces $\tr_i$
  well-bracketed complete \wrt $\stt_i$ such that 
  $\FFb_i \ArrWB {\tr_i}$ 
  and $\tr \in \interlWBs{\stt}{ \tr_1}{\tr_2}$.
  By Lemma~\ref{l:decNsig2-WB}, we also have that
  $\GGb_i \ArrWB {\tr_i}$.   
  
  By Lemma~\ref{c:dec_GS_WB}, there is 
   $\tracetheta \in \interlWBs{\stt}{\tracerho}{\tr_2}$ complete such that
  $\FFb \ArrWB {\tracetheta } $.
  Since 
  $\FFb  \TEwbc \GGb$, 
  we also have $\GGb \ArrWB {\tracetheta } $. 
  By Lemma~\ref{l:decNsig2-WB}, 
  we deduce that 
  $\GGb_1 \ArrWB {\tracerho }$ and 
  $\GGb_2 \ArrWB {\tr_2 }$.
  This contradicts the assumption that 
   ${\tracerho }$ is not a trace for $\GGb_1$.
\end{proof} 
\fi
\section{Complete-trace equivalences coincide }
\label{s:wbc_vs_c}

In this section, we show that complete-trace equivalences 
in \COGS, \AOGS, and \WBOGS  coincide.

\begin{remark}
Complete trace equivalence (in all the OGS forms discussed)  is strictly weaker than
ordinary trace equivalence, as shown by the following terms:
\[
 (\lambda z.\Omega)(x (\lambda y.\Omega)) \text{ and }
 (\lambda z.\Omega)(x (\lambda y.y))
\]
The terms  are equal  if only complete traces are observed, while  they are
not otherwise.
\end{remark}

\reviewComment{
 I am confused here. This example should be more
detailed, because to me these configurations do *not* seem to be trace
equivalent in WB-OGS. There seems to be a tension with Lemma 10.4.
}
\reviewCommentDS{I have rephrased the remark, see if ok; indeed i believe the example was about
the difference between complete and ordinary trace equivalence, something i guess well known}
\begin{theorem}
\label{t:main_2_WBC}
For any $\FFb,\GGb$, 
we have 
$\FFb  \TEwbc \GGb$
iff
$\es\FFb  \TEc \es\GGb$ iff
$\es\FFb  \TEpc \es\GGb$.
\end{theorem} 
\noindent
One direction of the theorem is easy, intuitively because a trace in 
 \WBOGS is 
also a trace in \AOGS, and these, in turn,  are traces in 
 \COGS.
The problem is the converse, namely showing 
that  
$\FFb  \TEwbc \GGb$
implies
$\es\FFb  \TEc \es\GGb$, and 
that 
$\es\FFb  \TEc \es\GGb$ 
implies $\es\FFb  \TEpc \es\GGb$.
They are proven in the next two sections.

\begin{remark}
\label{r:CT_COGS}
In earlier sections, (Section~\ref{s:cogs}),
we have  derived another equality result for
trace equivalence, namely Corollary~\ref{c:OGS_COGS}, 
relating  trace equivalences in the Alternating 
and Concurrent LTS. 
The structure of the proofs of that result and the one above 
(Theorem~\ref{t:main_2_WBC})
 are however quite different. 
In  Corollary~\ref{c:OGS_COGS} we derived the result via up-to techniques; 
we cannot apply such techniques here because we are reasoning on trace equivalence,
notably a restricted form of traces (the complete ones). 
On the other hand,  the proof strategy here 
rests on decomposition properties 
such as  Lemma \ref{l:decWBCN} that rely on a termination property for initial
configurations, which may not be assumed in 
Section~\ref{s:cogs}.
\end{remark} 

\subsection{Relating \WBOGS and \AOGS complete traces}
\label{ss:main_wbct}

\ifcutLMCS

\begin{lemma}
\label{l:FFbFF}
For  configurations $\stack \FFb \stt$ and $\stack \GGb \stt$, if 
$ \es\FFb  \TEcN  \stt  \es\GGb$,  then $
\FFb  \TEwbcN \stt \GGb$. 
\end{lemma} 

The problem is the converse. 
This is tackled in the remainder of the section 
(the main result is Lemma~\ref{l:mainWBCN}).
\fi
             
\begin{lemma}
\label{l:TCtoWB}
Taking $\FFb$ a \WBOGS configuration,
if $\es\FFb$ has a complete trace in \AOGS, 
then $\FFb$ has a well-bracketed complete trace. 
\end{lemma}

\begin{proof}
Assume $\trace$ is a complete trace for $\es\FFb$ in \AOGS.
We proceed by induction on the length of $\trace$. 
If $\trace$ is empty then the set of P- and O- continuation names of
$\es\FFb$ is empty.
So $\trace$ is also a complete trace for $\FFb$
since the stack of $\FFb$ must then be empty.

Suppose now $\trace = \mu \trace'$ is not empty,
so that there exists $\FF'$
such that $\es{\FFb} \ArrA \mu \FF' \ArrA{\trace'}$.
We reason by case analysis over $\mu$:
\begin{itemize}
\item If $\mu$ is a question,  then $\mu$ is either 
of the form $\questPV{x}{y}{p}$ 
or of the form $ \questOV{x}{y}{p}$.
We also have $\FFb  \ArrWB \mu \FFb'$
with 
$\es{\FFb'}= \FF'$.
Moreover $\trace $ is a complete trace 
for $\FF'$. 
By induction $\FFb'$ 
has a well-bracketed complete trace, 
say $\tracerho$.
Then $\mu \tracerho$ is a well-bracketed complete trace 
for $\FFb$. 
\item Suppose now $\mu$ is an answer. There are two cases:
\begin{itemize}
\item If $\FFb$ is active, then the answer $\mu$ must be of the form 
$\ansP{p}{x}$, where $p$ is the continuation name for the running term.
Thus $\mu$ is also the correct answer for \WBOGS, and we can easily conclude using
induction. 
\item If $\FFb$ is passive, then 
 $\mu$ must be of the form 
 $\ansO{p}{x}$. Writing $\stt$ for the stack of $\FFb$,
 if $\stt = p::\sttp$ then $\mu$ is also    
 the correct answer for a well-bracketed trace, and as before  we can conclude using
 induction.  Otherwise, suppose 
 $\stt = q::\sttp$ with $q\neq p$
 (i.e., $p$ is not the last  name in $\stt$).
 This means that $\FFb$ has at least two
 singleton components, located at $p$ and $q$. 
 Take a decomposition
 $\FFb = \FFb_1 \otimes_{\sigma} \FFb_2 $,    
 where $p$ and $q$ end up in different sub-configurations. 
 We get from Lemma~\ref{l:AOGS-decomp-ct}
 that there exists $\trace_i$ complete for $\es{\FFb_i}$ in \AOGS, 
 for $i\in \{1,2\}$, and 
 $\mu \trace \in \interlA{\trace_1}{\trace_2}$.
 By induction, there is $ \tracerho_i $ a complete well-bracketed trace 
 for $\FFb_i$ in \WBOGS.
%
 By Lemma~\ref{c:dec_GS_WB}, 
 there is $\tracerho \in \interlWBs{\stt}{\tracerho_1}{\tracerho_2}$, 
 a complete well-bracketed trace for $\FFb$.
\qedhere
\end{itemize}
\end{itemize}
\end{proof} 

\begin{lemma}
\label{l:mainWBCN}
Suppose 
$\FFb \TEwbc \GGb$,
then
$\es\FFb \TEc \es\GGb$
\end{lemma}              

\begin{proof}
First, if $\FFb,\GGb$ are not terminating, i.e. they have no complete trace in \WBOGS, 
then from Lemma~\ref{l:TCtoWB}, both 
$\es{\FFb},\es{\GGb}$ have no well-bracketed complete trace, thus
$\es\FFb \TEc \es\GGb$. 

Suppose now that $\FFb$ and $\GGb$ are both terminating.
By Lemmas~\ref{l:decWBCN} and \ref{l:compCN}, it is sufficient to consider 
the case when $\FFb$ and $\GGb$ are singleton. 
Suppose there exists a complete trace $\tr$ for $\es{\FFb}$, 
then we prove by induction over $\tr$
that $\tr$ is a complete trace for $\es{\GGb}$ as well.
If $\tr$ is empty, then from Lemma~\ref{l:OGS_complete-strongly-passive},
$\es{\FFb}$ must be strongly passive.
Recall that from $\FFb \TEwbc \GGb$ we enforce that
$\okC{\FFb}{\GGb}$.
So $\es{\GGb}$ must be strongly passive as well, so that 
the empty trace is complete for $\es{\GGb}$ as well.

Suppose now that 
$\tr = \act \tr'$, so that
$\es\FFb \Longrightarrow \arrA {\act} \FF' \ArrA{\tr'} \FF''$
with $\tr'$ complete for $\FF'$.
Then $\act$ is the only possible visible transition for $\es\FFb$ 
since $\es\FFb$ is a singleton.
So it is also the first
visible transition in any well-bracketed complete trace for $\FFb$ (i.e., we have 
$\FFb \ArrWB{} \arrWB{\act} \FFb'$ with $\es{\FFb'} = \FF'$). 
Indeed, if $\act$ is an Opponent answer $\ansO{p}{x}$, then the stack
of $\FFb$ must be equal to $p$ since $\FFb$ is a singleton.
Since $\FFb$ is terminating, it has thus a well-bracketed complete
trace starting with $\act$.
So does $\GGb$ since 
$\FFb  \TEwbc \GGb$, that is, we have  
$\GGb \ArrWB{} \arrWB{\act} \GGb'$, for some $\GGb'$, and therefore also 
$\es\GGb \ArrA{} \arrA{\act} \es{\GGb'}$. 
Moreover, $\FFb'  \TEwbc \GGb'$.
%

If $\FFb'$ and $\GGb'$ are singleton, we can now apply induction and 
derive that $\tr'$ is also complete for $\es{\GGb'}$, and then we are done. 
If they are not singleton, we can consider a decomposition for them 
(yielding $\FFb'_i$ and $\GGb'_i$), and associated decomposition
for $\tr'$ in complete traces $\tr'_i$ for $\es{\FFb_i}$ using Lemma~\ref{l:AOGS-decomp-ct}. 
We can then apply induction on the singleton components (using
Lemma~\ref{l:decWBCN}), then we infer that 
$\tr'_i$ is also a complete trace for $\es{\GGb_i'}$. 
We can now apply Lemma~\ref{c:dec_GS} and infer that $\tr'$ is a complete trace for $\es{\GGb'}$,
and then conclude that $\act\tr'$ is a complete trace for $\es\GGb$.
\end{proof} 

\subsection{Relating \AOGS and \COGS complete traces}
\label{a:aogs-cogs-complete}

\begin{lemma}
\label{l:CtoCA}
Taking $\FF$ an \AOGS configuration,
if $\FF$ has a complete trace in \COGS, then $\FF$ 
has a complete trace in \AOGS. 
\end{lemma}

\begin{proof}
Assume $\trace $ is a complete trace
for $\FF$ in \COGS.
We proceed by induction on the length of $\trace $.
If $\trace$ is empty then this is straightforward.
Suppose now there exists $\mu$ s.t. 
$\trace = \mu \tr'$. We reason by case analysis on $\mu$:
\begin{itemize}
\item If $\mu$ is a Player action,
then $\FF \ArrA \mu \FF'$, so we conclude using induction.
\item If $\mu$ is an Opponent action,
there are two cases:
\begin{itemize}
\item If $\FF$ is a passive configuration, 
then again $\FF \ArrA \mu \FF'$, and we conclude using induction.
\item If $\FF$ is an active configuration, then we can write
$\FF$ as $\FF_1 \otimes \FF_2$ since $\FF$ has both
an active term and at least one element in its environment,
that is triggered by $\mu$.
We get from Lemma~\ref{l:COGS-decomp-ct}
that there exist two traces $\trace_1,\trace_2$ complete respectively for
$\FF_1$ and $\FF_2$ in \COGS such that $\mu \trace \in \interl{\trace_1}{\trace_2}$.
By induction, there exists two traces $\tracerho_1,\tracerho_1$ 
complete respectively for $\FF_1$ and $\FF_2$ in \AOGS.
Finally,
by Lemma~\ref{l:AOGS-recomp-ct}, there exists a trace
$\tracerho \in \interlA{\tracerho_1}{\tracerho_2}$
that is complete for $\FF$ in \AOGS.
\qedhere
\end{itemize}
\end{itemize}
\end{proof} 

\begin{lemma}
\label{l:mainCC-CA}
Suppose 
$\FF  \TEc \GG$,
then
$\FF  \TEpc \GG$.
\end{lemma}   

\begin{proof}
First, if $\FF,\GG$ are not terminating, i.e. they have no complete trace in \AOGS, 
then from Lemma~\ref{l:CtoCA}, both 
$\FF,\GG$ have no complete trace \COGS, thus
$\FF \TEpc \GG$. 

Suppose now that $\FF$ and $\GG$ are both terminating in \AOGS.
By Lemmas~\ref{l:decACN} and the adaptation of Lemma~\ref{l:compCN} to \COGS, it is sufficient to consider 
the case where $\FF$ and $\GG$ are singleton. 
Suppose there exists a complete trace $\tr$ for $\FF$ in \COGS, 
then we prove by induction over $\tr$
that $\tr$ is a complete trace for $\GG$ in \COGS as well.
If $\tr$ is empty, then $\FF$ must be strongly passive.
Recall that from $\FF \TEc \GG$ we enforce that
$\okC{\FF}{\GG}$.
So $\GG$ must also be strongly passive, so that 
the empty trace is a complete trace for $\GG$ as well.

If $\tr$ is of the shape $\act\tr'$, then
$\FF \Longrightarrow \arrC {\mu} \FF' \ArrC{\tr'}$.
Since $\FF$ is singleton, $\FF$ cannot be active with $\mu$ an Opponent action.
So $\mu$ is a Player action, or an Opponent action with $\FF$ passive,
so that $\FF \ArrA{}  \arrA {\mu} \FF'$.
Since $\FF$ is terminating in \AOGS, we get that $\FF$ has a complete trace
starting with $\mu$ in \AOGS.
Since 
$\FF  \TEc \GG$, we also have
$\GG \ArrA{} \arrA \mu \GG'$, for some $\GG'$, and therefore also 
$\GG \ArrC{} \arrC \mu \GG'$. 
Moreover, $\FF' \TEc \GG'$.

If $\FF'$ and $\GG'$ are singleton, we can now apply the induction hypothesis 
and derive the same complete trace  $\tr'$ also for $\GG'$ in \COGS, and then we are done. 
If they are not singleton, we can consider a decomposition for them 
(yielding $\FF'_i$ and $\GG'_i$), and associated decomposition
for $\tr'$ in complete traces $\tr'_i$ for $\FF_i$ using Lemma~\ref{l:COGS-decomp-ct}. 
We can then apply induction on the singleton components (using
Lemma~\ref{l:decACN}), and getting that
$\tr'_i$ is also a trace for $\GG_i'$ in \COGS. 
We can now apply Lemma~\ref{l:COGS-recomp-ct} 
and infer that $\tr'$ is a complete trace for $\GG'$ in \COGS,
and then conclude that $\act\tr'$ is a complete trace for $\GG$ in \COGS.
\qedhere
\end{proof}

\section{The call-by-name $\lambda$-calculus}
\label{s:cbnNEW}

The relationship between OGS and $\pi$-calculus representations of the 
$\lambda$-terms examined in Sections~\ref{s:encoGames} and \ref{s:cogs} is not specific to
call-by-value.  A similar relationship can be established for call-by-name. 
We summarise here the main aspects and the main results. 
Further details may be found in Appendix~\ref{a:cbn}.

We maintain terminologies and notations in the call-by-value sections. 
The call-by-name language is given in Figure~\ref{f:cbnG}.
Its reduction $\redn$ is defined by the following two rules:
 
 \begin{mathpar}
   \inferrule*{ }{(\lambda x.M) N \redn M\subst{x}{N}} \qquad
   \inferrule*{M \redn N}{E[M] \redn E[N]}
 \end{mathpar}
 

The OGS for call-by-name uses
three kinds of names: 
\emph{continuation names} (ranged over by $p,q$), 
\emph{variable names} (ranged over by $x,y$),
and 
\emph{value names} (ranged over by $v,w$). 
Continuation names are used as in call-by-value; 
variable names act as pointers to 
$\lambda$-terms and correspond to the variables of the $\lambda$-calculus;
value names are pointers to values.

The OGS uses the syntactic categories in Figure~\ref{f:cbnG}, 
where $L$ range over extended $\lambda$-terms, 
which are needed to accommodate value names. 
We write $\LaoP$ for the set of extended $\lambda$-terms.
\begin{figure*}
\[ 
\begin{array}{rclr}
E & := &  EM \midd \contexthole  &  \mbox{evaluation contexts}\\
V & := &  \lambda x. M \midd v &  \mbox{values} \\
\gamma & := & 
\pmap{v}{V}\cdot\gamma' \midd 
\pmap{q}{(E,p)}\cdot\gamma' \midd 
\pmap{x}{M}\cdot\gamma' \midd 
\emptymap  &  \mbox{environments} \\
F & := & \conf{\LL,p,\gamma ,\phi } \midd 
 \conf{\gamma,\phi} \midd \initconf{M}{\phi} & \mbox{configurations} \\
\LL & := &  M \midd E [v M] \midd v   & \mbox{extended $\lambda$-terms}\\
\end{array} 
\] 
\caption{Grammars for the call-by-name language and \AOGS configurations}
\label{f:cbnG}
\end{figure*} 

\dsOLD{revise Figure~\ref{f:cbnG}}

\noindent
In a call-by-name setting, Player Questions are of two possible shapes:
 \begin{itemize}
  \item \emph{Player Value-Question} (PVQ) $\questPV{v}{x}{p}$,
  receiving a variable $x$ and a continuation name $p$ through a value name $v$;
  \item \emph{Player Term-Question} (PTQ) $\questPT{x}{p}$,
  receiving a continuation name $p$ through a variable $x$.
 \end{itemize}
\noindent
 Opponent Questions are also of two different shapes too:
  \begin{itemize}
   \item \emph{Opponent Value-Question} (OVQ) $\questOV{v}{x}{p}$,
  sending a variable $x$ and a continuation name $p$ through a value name $v$;
  \item \emph{Opponent Term-Question} (OTQ) $\questOT{x}{p}$,
  sending a continuation name $p$ through a variable $x$.
 \end{itemize}

 \begin{figure*}
 \setlength{\arraycolsep}{3pt}%
 \[
  \begin{array}{l|lll@{}l}
   (\tau) & \conf{M,p,\gamma,\phi} & \arrA{\ \tau \ } & 
     \conf{N,p,\gamma,\phi} & \text{ when } M \redn N\\
   (PA) & \conf{V,p,\gamma,\phi} & \arrA{\ansP{p}{v}} & 
     \conf{\gamma \cdot \pmap{v}{V},\phi \uplus \{v\}} \\
   (PVQ) & \conf{E[vM],p,\gamma,\phi} & \arrA{\questPV{v}{y}{q}} & 
     \conf{\gamma \cdot \pmap{y}{M}\cdot \pmap{q}{(E,p)},\phi \uplus \{y,q\}} \\
   (PTQ) & \conf{E[x],p,\gamma,\phi} & \arrA{\questPT{x}{q}} & 
     \conf{\gamma \cdot \pmap{q}{(E,p)},\phi \uplus \{q\}} \\
   (OA) & \conf{\gamma\cdot \pmap{p}{(E,q)},\phi} & \arrA{\ansO{p}{v}} & 
     \conf{E[v],q,\gamma,\phi \uplus \{v\}} \\ 
   (OVQ) & \conf{\gamma \cdot [v \mapsto V],\phi} & \arrA{\questOV{v}{y}{p}} & 
     \conf{V y,p,\gamma,\phi \uplus \{p,y\}} & \text{ when } \gamma(v) = V\\
   (OTQ) & \conf{\gamma,\phi} & \arrA{\questOT{x}{p}} & 
     \conf{M,p,\gamma,\phi \uplus \{p\}} & \text{ when } \gamma(x) = M\\
   (IOQ) & \initconf{\phi}{M} & \arrA{\questOinit{p}} &
     \conf{M,p,\emptymap,\phi \uplus \{p\}}
%
  \end{array}
 \]
 \caption{CBN OGS LTS}
\label{fig:scbn-lts}
\end{figure*}

%
%

\begin{figure}
\Mybar 
{\small
Encoding of $\Lao$ into  \piI processes: 
\[\begin{array}{rcl}
\encoN{ \lambda  x. M}   & \defi & 
\bind p
  {  \bout p v  .   \inp v {x,q}  .
\encoNa M q}
 \\[\mypt]
\encoN{  x}   & \defi & \bind p
 { \bout xr . \link r p   } \\[\mypt]
\encoN{M N } & \defi &  \bind p
 {\res q \Big(
\encoNa{M}{ q} |
 \inp q v . { \bout v { x, {p'}} . ( \link {p'} p |  ! \inp x r. {\encoNa{ N }{ r}})  }
\Big)
}
\end{array}
   \]
$ $ \dotfill $ $ 

Encoding of    $\LaoP$:
\[ 
\begin{array}{rcl}
\encoN {E [v M]} & \defi & \bind p \bout v {x,r}. ( \encoN { \pmap r{(E,p)}} | ! \inp x
                       {q}. \encoNa Mq   ) 
\\[\mypt]
\encoN {v } & \defi & \bind p \bout p w . \link wv
\end{array} 
\]
$ $ \dotfill $ $ 

Encoding of  OGS environments:
\[ 
\begin{array}{rcl}
\encoN { \pmap{v}{V} \cdot \gamma   } & \defi & 
\left\{ \begin{array}{ll} 
 \inp v {x,q}. \encoNa M q  | \encoN{ \gamma}
&      
 \mbox{if $V  = \lambda x. M$ } \\[\mypt]
\link v w  | \encoN{ \gamma}
   & 
\mbox{if $V  =  w$ } 
\end{array}  \right. 
\\[\myptt]
\encoN { \pmap{x}{M}  \cdot \gamma } & \defi &  ! \inp x q . \encoNa M q   | \encoN{ \gamma} \\[\mypt]
\encoN {\pmap{q}{(E,p)}  \cdot \gamma  } & \defi & 
\left\{ \begin{array}{@{}l@{\;}l@{}} 
 q(v).\bout v {x,r}.(\encoN{\pmap{r}{E',p}} | !\inp x q.\encoNa Mq)
  | \encoN{\gamma}
& \mbox{if $E=E'[\contexthole\,M]$} \\[\mypt]
\link qp   | \encoN{ \gamma} & \mbox{if $E=\contexthole$} 
\end{array}  \right. 
\\[\myptt]
\ifcutCBNLMCS
\multicolumn{3}{l}{ \mbox{(An alternative, but less confortable for
    proofs)} }
\\[\myptt]\multicolumn{3}{l}{ \mbox{(above: M is the first argument, below it is the last)} }
\\[\myptt]
\encoN {\pmap{q}{(E,p)}  \cdot \gamma  } & \defi & 
\left\{ \begin{array}{ll} 
\res r ( \pmap{q}{E',r }
|
 r(v). \bout v {x,p'}. 

(\link {p'}p |
 ! \inp x {q'} . \encoNa M{q'} )
  | \encoN{ \gamma}
&      
 \mbox{if $E   = 
E' M
$ } \\[\mypt]
\link qp   | \encoN{ \gamma} &      
 \mbox{if $E   = \contexthole$ } 
\end{array}  \right. 
\\[\myptt]
\fi
\encoN { \emptyset }  & \defi & \nil
\end{array} 
\]
$ $ \dotfill $ $ 

Encoding of  OGS configurations: 
\[ 
\begin{array}{rcl}
\encoN {\conf {M,p,\gamma , \phi}} & \defi & \pairP { \encoNa M p | \encoN \gamma }
\\[\mypt]
\encoN {\conf {v,p,\gamma , \phi}} & \defi & \pairP { \encoNa v p | \encoN \gamma }
\\[\mypt]
\encoN {\conf {E[v M],p,\gamma , \phi}} & \defi & \pairP { \encoNa {E[vM]} p | \encoN \gamma
                                             }
\\[\mypt]
\encoN {\conf {\gamma , \phi}} & \defi & \pairP { \encoN \gamma }
\end{array} 
\]
}

\caption{The encoding of
call-by-name $\lambda$-calculus and OGS  into \piI }
\label{f:namePI}
\Mybar 
\end{figure}

\noindent
The encodings of call-by-name 
$\lambda$-calculus,  OGS environments and configurations, and associated  syntactic categories,
 are reported in Figure~\ref{f:namePI}.
The \piI representation uses the same 3 types of  names as the OGS representation. 

As for call-by-value, so in call-by-name the encoding makes use of 
 \emph{link} processes, which are linear structures for continuation
 names and replicated structures for the other names.  We recall their
 definition. 
For tuples of names 
 $\til u = u_1, \ldots , u_n$
 and
$\til v =  v_1, \ldots , v_n$   
we write $\link{\til u}{\til v}$
to 
  abbreviate 
   $\link{u_1}{v_1} | \ldots |  \link{u_n}{v_n}$:
$$
\begin{array}{rcl}
\link{}{} &\defi &
\left\{ \begin{array}{l}
 \abs{x,y} ! \inp x {p}.\bout {y} {q}
           . {\fwd{ q}{ p}}\\[\myptSmall] 
 \multicolumn{1}{l}{  
        ~\quad  \mbox{if $x,y$ are variable names}
                   }              \\[\mypt]
\abs{a,b}   \inp a { \til c }.\bout b {\til d}
          .{\link{ \til d}{ \til c}})  \\
\multicolumn{1}{r}{  
          ~\quad \mbox{otherwise}
}\end{array} \right. 
\end{array}
 $$

\noindent
We have:
\[ 
\begin{array}{rcl}
\encoN {E [v ]} & = & 
\left\{ \begin{array}{ll}
\encoN {E' [vM ]}  & 
\mbox{if $E =
E'[\:  \contexthole M   \: ] 
$}
\\[\mypt]
\encoN {v}  & 
\mbox{if $E = \contexthole$}
\end{array} \right. 
\end{array} 
\]

\dsJ{I have rectified above: 
{if $E =
E'[\:  \contexthole M   \: ] 
$} rather than 
{if $E =
E' M
$}
}

\begin{lemma}
\label{l:Ex}
$ 
\begin{array}{rcl}
\encoN {E [x]} & = & \bind p 
\bout x {q}.  \encoN { \pmap q{(E,p)}} 
\end{array} 
$
\end{lemma} 
\noindent
The results of operational correspondence between OGS and \piI under the \opLTS are the
same as for call-by-value. 
We only report the operational correspondence for strong transitions and the 
analogue of Corollary~\ref{c:fa_traces_bis}.

\begin{theorem}
\label{t:opcorrCON_cbnshort}
$ $ \begin{enumerate}


\item  If $F \longrightarrowA F'$ then 
$ \encoN F \longrightarrowPI \;   \contr \encoN {F'}$; 

\item  If $F \arrA{\ell} F'$  
then 
$ \encoN F \arr{\mu}  \; \wbPI  \encoN  {F'}$.
\end{enumerate}
\end{theorem}

\begin{theorem} 
\label{t:opcorrCON3cbn}
$ $ \begin{enumerate}

\item If 
$ \encoN \FF \arrN \tau \PP$ then there is $\FF'$ such that 
$\FF \longrightarrowA \FF'$  and  
$\PP   \contr \encoN {\FF'}$;

\item If
$ \encoN \FF \arrN\ell \PP$
 then there is $\FF'$ such that
 $\FF \arrA\ell \FF'$
and $\PP \wbPI \encoN {\FF'}$.
\end{enumerate} 
\end{theorem}

\begin{corollary}
\label{c:CBNtraces}
For any configuration $\FF$ and trace $s$,
 we have 
  $\FF \ArrA s$   iff 
 $ \encoN \FF  \ArrN s$.
\end{corollary} 

\subsection{Concurrent Operational Game Semantics}
\label{ss:cogs_cbn}

The modification to obtain the concurrent version of OGS (\COGS) for call-by-name are
similar to those for call-by-value in Section~\ref{s:cogs}. 

Below are the rules for the call-by-name  OGS (\COGS).  
We use the same notations as the call-by-value version. 

 \setlength{\arraycolsep}{3pt}%
 \[
  \begin{array}{l|lll@{}l}
   (Int) & \conf{A \cdot \pmap{p}{M},\gamma,\phi} & \arrC{\qquad} & 
     \conf{A \cdot \pmap{p}{N},\gamma} & \text{ when } M \redn N\\
   (PA) & \conf{A \cdot \pmap{p}{V},\gamma,\phi} & \arrC{\ansP{p}{v}} & 
     \conf{A, \gamma \cdot \pmap{v}{V},\phi\uplus \{v\}} \\
   (PVQ) & \conf{A \cdot \pmap{p}{E[vM]},\gamma,\phi} & \arrC{\questPV{v}{y}{q}} & 
     \multicolumn{2}{l}{\conf{A, \gamma \cdot \pmap{y}{M}\cdot \pmap{q}{(E,p)},\phi\uplus \{y,q\}}} \\
   (PTQ) & \conf{A \cdot \pmap{p}{E[x]},\gamma,\phi} & \arrC{\questPT{x}{q}} & 
     \conf{A, \gamma \cdot \pmap{q}{(E,p)},\phi \uplus \{q\}} \\
   (OA) & \conf{A, \gamma,\phi} & \arrC{\ansO{p}{v}} & 
     \conf{A \cdot \pmap{q}{E[v]},\gamma,\phi \uplus \{v\}} & \text{ when } \gamma(p) = (E,q)  \\

   (OVQ) & \conf{A, \gamma \cdot [v \mapsto V],\phi} & \arrC{\questOV{v}{y}{p}} & 
     \conf{A \cdot \pmap{p}{V y},\gamma,\phi \uplus \{p,y\}} &\\
   (OTQ) & \conf{A,\gamma,\phi} & \arrC{\questOT{x}{p}} & 
     \conf{A \cdot \pmap{p}{M},\gamma,\phi \uplus \{p\}} & \text{ when } \gamma(x) = M\\
   (IOQ) &\initconf{\phi}{M}  & 
 \arrC{\questOinit{p}}
 &
     \conf{M,p,\emptymap,\phi \uplus \{p\}}
  \end{array}
 \]
\noindent
As usual, for encoding   \COGS,
it suffices  to consider
 running
terms and configurations,  as the encoding remains otherwise the same.
The encoding of the running term is: 
\[ 
\begin{array}{rcl}
\encoN{\pmap{p}{M}\cdot A  } & \defi & 
\encoNa  {M}{p} | \encoN{ A  } \\
\encoN{\emptymap} & \defi & \nil
\end{array} 
\]
The encoding of configurations is then:
\[
\encoN{\conf{ A , \gamma, \phi } }
\defi  \pairP{  \encoN A
 | \encoN \gamma }
\]

\noindent
We report the operational correspondence for strong transitions.

\begin{lemma}[from \COGS to \piI, strong transitions]
\label{l:CBNpogs_short}
$ $ \begin{enumerate}

\item  If $F \longrightarrowp F'$ then 
$ \encoSN F \longrightarrowPI \;  \contr \encoSN {F'}$; 

\item  If $F \arrC{\ell} F'$  then 
$ \encoSN F \arrPI{\ell}\;  \wbPI  \encoSN  {F'}$.
\end{enumerate}
\end{lemma}

\begin{lemma}[from  \piI to \COGS, strong transitions]
\label{l:CBNpogs_short_converse}
$ $ \begin{enumerate}

\item  If
$ \encoSN F \longrightarrowPI  P$, then there is $F'$ s.t.\  
 $F \longrightarrowp F'$ and 
$P 
\contr \encoSN {F'}$;

\item  If
$ \encoSN F \arrPI{\ell}  P$, then there is $F'$ s.t.\  
 $F \arrC{\ell} F'$ and 
$P 
\wbPI \encoSN {F'}$.
\end{enumerate}
\end{lemma}

\begin{lemma} 
\label{l:ogsp_pi_traces_bisimulation_cbn}
For any $\FF,\FF'$ we have:
\begin{enumerate}
\item
 $\FF_1 \TEp \FF_2$ iff 
 $\encoN {\FF_1} \TEPI \encoN {\FF_2}$;  
\item
 $\FF_1 \wbC \FF_2$ iff 
 $\encoN {\FF_1} \wbPI  \encoN {\FF_2}$.  
\end{enumerate} 
\end{lemma}

\begin{theorem}
\label{t:CBNfa_all}
For all $M,N$   with $\fv{M,N}\subseteq \phi$, we have: 
$\conf{ M,\phi} \TEp \conf{ N,\phi}$ iff 
$\conf{ M,\phi} \wbC \conf{ N,\phi}$ iff 
$\pairP{\encoSN M} \TEPI \pairP{\encoSN N}$
iff 
$\pairP{\encoSN M} \wb \pairP{\encoSN N}$. 
\end{theorem}

\subsection{Up-to techniques for games, and relationship between Concurrent and Alternating OGS}
\label{ss:upto_games_cbn}

As for call-by-value, so for call-by-name the 
'up-to
  composition' technique for bisimulation 
(Definition~\ref{d:uptoComp}) 
can be imported from the $\pi$-calculus, and then used to establish that the equivalence
induced on $\lambda$-terms by the sequential and concurrent versions of OGS  coincide. 

We use  the technique to prove the call-by-name analogue of 
Lemma~\ref{l:OGS_COGS}.

\begin{corollary}
\label{c:OGS_COGS_cbn}
For any   $\lambda$-terms $M,N$, the following statement are the same: 
$\conf{M} \TEA \conf {N}$;
$\conf{M} \wbA \conf {N}$;
$\conf{M} \TEp \conf {N}$;
$\conf{M} \wbC \conf {N}$;
$\encoN {{M} } \TEN\encoN {{N}}$; 
$\encoN {{M}} \wbOPI \encoN {{N}}$; 
$\encoN {{M} } \TEPI\encoN {{N}}$; 
$\encoN {{M}} \wbPI \encoN {{N}}$. 
\end{corollary} 
\noindent
The equivalence induced by the \piI (or  $\pi$-calculus) encoding of the call-by-name
$\lambda$-calculus is known to coincide with    that of the  
 L{\'e}vy-Longo Trees \cite{San95lazy}.  In the light of  
Corollary~\ref{c:OGS_COGS_cbn}, the result can thus be transported to OGS, both 
in its Alternating and its Concurrent variants, 
and both for traces and for bisimilarity.

Connections between game semantics and tree representation of $\lambda$-terms
have been previously explored, e.g.,  in~\cite{ong2004games}, 
where a game model of untyped call-by-name $\lambda$-calculus
is shown to correspond to Levy-Longo trees.

The connections that we have derived above, however, go beyond
L{\'e}vy-Longo Trees:
in particular,
the bisimilarities that have been established
between the observables (i.e., the dynamics) in the two models
set a tight relationship between them 
while allowing one to transfer results and techniques along the lines of what we have
shown for call-by-value.  

\section{Related and Future Work}
\label{s:cf}

\dsOLD{maybe this section could be merged with the conclusion ``Related  and future work'';
  also we could try to have a stronger structure (maybe subsections, not sure)}

Analogies between game semantics and $\pi$-calculus, as semantic frameworks in which
\emph{names} are central, have been pointed out from the very beginning of game
semantics. 
In the pioneering work~\cite{10.1145/224164.224189}, the authors 
obtain 
 a translation of PCF terms into the $\pi$-calculus from 
a game model of PCF by representing 
 \emph{strategies} 
(the denotation of PCF terms in the game model) as processes of the $\pi$-calculus.
The encoding bears similarities with Milner's, though they are syntactically rather
different (``it is clear that the two are conceptually quite unrelated'',
\cite{10.1145/224164.224189}).   The connection has been  developed in various papers,
e.g.,
\cite{DBLP:journals/entcs/Honda02,10.1016/S0304-3975-99-00039-0,10.5555/788020.788890,yoshida2020game,10.1145/3290340}.
Milner's encodings into the $\pi$-calculus have sometimes  been a source of inspiration  in
the definition of the game semantics models (e.g.,  transporting 
 the work 
\cite{10.1145/224164.224189}, in call-by-name, onto call-by-value
 \cite{10.1016/S0304-3975-99-00039-0}).  
In~\cite{yoshida2020game}, a  typed variant of  $\pi$-calculus, influenced by 
differential linear logic~\cite{ehrhard2018introduction},
is introduced as a metalanguage to represent game models.
\iflong 
 The approach is illustrated on a
a higher-order call-by-value language with shared memory concurrency, whose game model is build by first translated it
to this metalanguage, then using the game model of this metalanguage.
\fi
In~\cite{10.1109/LICS.2015.20}, games are defined using algebraic operations on sets of
traces, and used to  prove type soundness of a simply-typed call-by-value
$\lambda$-calculus with effects.
 Although the calculus of traces employed is not a $\pi$-calculus 
(e.g., being defined from operators and relations over trace sets rather than from
syntactic process constructs), there are similarities, which would be interesting to
investigate.

Usually in the above  papers
the source language is a form of $\lambda$-calculus, 
that is interpreted into game semantics,  
and the 
 $\pi$-calculus  (or dialects of it)  is used to
represent the resulting strategies and games. 
Another goal has been to shed light on typing
disciplines  for 
$\pi$-calculus processes, by transplanting conditions on strategies such as well-bracketing
and innocence into appropriate typings for the $\pi$-calculus (see, e.g., 
\cite{berger2001sequentiality,YBH11,HirschkoffPS21}).

In contrast with the above works,  where analogies between game semantics and
$\pi$-calculus are used to better understand one of the two models 
(i.e., explaining  game semantics in terms process interactions,  or enhancing type systems for
processes following structures in game semantics), 
in the  present paper we have carried out a direct comparison between 
the two models (precisely OGS and \piI). 
For this we have started from
 the (arguably natural)  representations of the $\lambda$-calculus 
into OGS and \piI  (the latter being Milner's encodings).
Our goal was
understanding the  relation between  the behaviours of the terms in the
two models, and   transferring techniques and results between them.

Technically, our work  builds on
\cite{DurierHS18,Durier20}, where a  detailed analysis of the behaviour of Milner's
call-by-value encoding is carried out using proof techniques for
$\pi$-calculus based on unique-solution of equations.
Various results in \cite{DurierHS18,Durier20} are essential to our own
 (the observation that Milner's
encodings should be interpreted in \piI rather than the full $\pi$-calculus is also from
\cite{DurierHS18,Durier20}).

An OGS 
for the call-by-name $\lambda\mu$-calculus is introduced
in~\cite{DBLP:conf/fossacs/Laird20}. It 
 is different, once restricted to the fragment without the $\mu$-binder, to the one we
 consider in Section~\ref{s:cbnNEW}.  It might correspond  
 to a different way of encoding  call-by-name into the $\pi$-calculus
where $\lambda$-abstractions are interpreted as input-prefixed processes, as e.g., in 
 Milner's original encoding~\cite{Mil92s}.


Bisimulations over  OGS terms, and tensor products of configurations,
were introduced in~\cite{10.1145/2603088.2603150},
in order to provide a framework to study compositionality properties of OGS.
In our case the compositionality result of OGS is derived from the correspondence with the
$\pi$-calculus.
In~\cite{sakayori2019categorical}, a correspondence between an i/o typed asynchronous $\pi$-calculus 
and a computational $\lambda$-calculus with channel communication is established, 
using a common categorical model (a compact closed Freyd category).
It would be interesting to see if our concurrent operational game model could be equipped with this categorical semantics.

\emph{Normal form} (or \emph{open}) bisimulations~\cite{lassentrees,Stovring-Lassen},
as game
semantics, 
manipulate open terms, and sometimes make use of
environments or stacks of evaluation contexts (see e.g., the recent
work~\cite{biernacki2019complete}, where  a fully abstract  normal-form 
bisimulation for a $\lambda$-calculus with store is obtained).  

Notice that contextual equivalence for such a language corresponds to complete well-bracketed 
trace equivalence.
It would be interesting to see if, once the reasoning on diverging configurations removed,
the enf-bisimulations obtained there could be seen as up-to composition bisimulations
over the OGS LTS \cutLMCSsecondRound{for the $\lambda\rho$-calculus presented in Appendix~\ref{a:LambdaRho},
that is fully abstract for this language}.

There are also works that build
game models directly for the $\pi$-calculus.
In~\cite{laird2005game}, a game models for the simply-typed asynchronous $\pi$-calculus 
is introduced, and proved fully abstract for may-equivalence.
In~\cite{10.1007/11944836_38}, this construction was extended to a model of Asynchronous Concurrent ML.
These models rely on a sequential representation of interactions as traces.
More recently, causal representation of game models of the $\pi$-calculus have been proposed.
In~\cite{sakayori2017truly}, a truly concurrent game model for the asynchronous $\pi$-calculus is introduced,
using a directed acyclic graph structure to represent interactions.
A game model for a synchronous $\pi$-calculus is defined in~\cite{eberhart2015intensionally}, and shown to
be fully abstract for fair testing equivalence. 
A correspondence between a synchronous $\pi$-calculus with session types and 
concurrent game semantics~\cite{winskel2017games}
is given in~\cite{10.1145/3290340}, relating 
games (represented as arenas) to  session types, and 
strategies (defined as coincident event structures) to processes.

We have exploited the full abstraction results between OGS and \piI
 to transport a few  up-to techniques for
  bisimulation  from
  \piI onto OGS.  However, in \piI there are various other such techniques, even a theory of
 bisimulation enhancements.  We would like to see which other techniques could be useful
 in OGS,  possibly transporting  the theory of enhancements itself. 
We would like also to study the possible game-semantics counterpart of recent work on a
 representation of functions as processes \cite{SakayoriS23} in which different tree structures of the
$\lambda$-calculus (Lévy-Longo  trees,  Böhm trees,  Böhm trees with infinite $\eta$) are obtained
by modifying the semantics of the certain special processes used as a parameter of the
representation
 (and related to the \emph{link} processes of
Section~\ref{s:encoLpi} and \ref{s:cbnNEW}, called \emph{wires} in~\cite{SakayoriS23}).

\bibliographystyle{alphaurl}
\bibliography{biblio}

\appendix
\newpage





\section{Auxiliary material for Section~\ref{subsec:seq-ts-pi}}
\label{a:behav}

\DSOCT{ moved the first half of this appendix onto the main text, and consequently changed
title of the appendix}

On
 the encoding of $\lambda$-terms, 
the ordinary
bisimilarity $\wbPI$ implies the `bisimilarity respecting divergence' $\wbPIDIV$. 
%
%
This
because, intuitively, a term  $\encoIa {M} p$ is divergent iff $M$ is so in the call-by-value
$\lambda$-calculus.    Formally, the result is proved by rephrasing the result about the 
correspondence between $\wbPI$ and Lassen's trees (Theorem~\ref{t:adrien}) in terms of 
$\wbPIDIV$ in place of $\wbPI$. Again, this property boils down to the fact that the uses of
$\expa$ in Lemmas~\ref{l:opt_sound}-\ref{l:opt_con} (on which also the characterisation in
terms of Lassen's trees is built) respect divergence.  
Thus, using also Lemma~\ref{l:bisDIV_op}, we derive the following result. 

\begin{theorem}
\label{t:wbPI-wbOPI}
For $M,N \in \Lambda $, we have:
$\encoIa {M} p \wbPI \encoIa {N} p$ iff 
$\encoIa {M} p \wbOPI \encoIa {N} p$. 
\end{theorem} 


\ds{The 7.1 should be rewritten for \opLTS -- at least put in in the appendix}

\ds{below i talk about ``singleton configurations'', not specifying Alternating or COGS,
  for instance, because the grammar of singleton conf is the same and also then the
  encoding. Not sure if we want to say something}

\noindent
Below we report an alternative proof of the theorem above, however
always relying  on the variant of lemmas, such as  
 Lemma~\ref{l:opt_sound} and Lemmas~\ref{l:opt_stuck}-\ref{l:opt_con}, with 
$\expaOP $ in place of $\expa$. 

\begin{lemma}
\label{l:shapesT}
If $\FF$ is a singleton configuration, then  $\encoConI{\FF}$ is of three possible forms: 
\begin{enumerate}
\item 
$\encoIa {M} p$, for some $M,p$;
\item
$\inp qx .\encoIa {E[x]} p $, for some $E,x,q,p$;
\item 
$\encoIVa {V} y$, for some $V,y$ (that is, either  $!\inp y {x,q}.\encoIa M q$ if $V  =
\lambda x. M$, or $ \fwd y x$ if $V = x$).
\end{enumerate}  
\end{lemma}

\begin{lemma}
\label{l:deco}
Suppose  $a$ not free in $P_1 $ and $Q_1$. 
\begin{itemize}
\item
Let $P = P_1 | ! \inp a {\tilb}.P_2$, 
and $Q = Q_1 | ! \inp a {\tilb}.Q_2$, 
and  $ P \wbPI  Q$.
Then also 
  $ {P_1}  \wbPI  {Q_1}$.

\item 
Similarly, let $P = P_1 |  \inp a {\tilb}.P_2$, 
and $Q = Q_1 |  \inp a {\tilb}.Q_2$, 
and  $ {P} \wbPI  {Q}$.
Then also 
  $ {P_1}  \wbPI  {Q_1}$.
\end{itemize}
\end{lemma} 

\begin{theorem}
\label{t:opBIS_BIS}
If $\FF$ and $\GG$ are singleton configurations, then
$\encoConI  \FF \wbPI \encoConI  \GG $  implies 
 $\encoConI  \FF \wbOPI \encoConI  \GG $.
\end{theorem} 

\begin{proof}
Let $\R$ be the relations with all pairs $(P,Q)$ where $P$ and $Q$ are of the
form 
\[
\begin{array}{rcl}
 P &=&   
P_1  | \ldots | P_n
\\
 Q &=& 
Q_1  | \ldots | Q_n
\end{array}
 \] 
for some $n$, where all $P_i,Q_i$ are 
encodings of singleton configurations and  for all $i, $ we have $P_i \wbPI Q_i$.
Moreover:
\begin{itemize}
\item
 there is at most one $i$ for which  $P_i,Q_i$ are encodings of  active
configurations;
\item  in any $P_j$ (resp.\ $Q_j$) that is the encoding of a passive
configuration, the name of the initial input does not appear free in any other component
$P_{j'}$ (resp.\ $Q_{j'}$) for $j'\neq j$. 
\end{itemize} 
\noindent
A consequence of the two conditions above is that two distinct components $P_i$ and $P_j$
cannot interact. That is, if $P \arr \mu P'$,  then there is $i$ such
that $P_i \arr \mu P'_i$ and $P' = P''_1| \ldots | P''_n$ where $P''_j = P_j $ if $j\neq i$
and $P''_i = P'_i $. And similarly if   $P \arrN \mu P'$.

We  show that  $\R$ is a  $\wbOPI$-bisimulation  up-to $\expaOP$ (the bisimulation
up-to expansion is sound for bisimilarity in any LTS). 
In the proof below,  Lemmas~\ref{l:opt_op}  and 
\ref{l:opt_con}
actually refer to the versions of the lemmas mentioned above, with $\expaOP$
in place of $\expa$.

Suppose $P \arrN \mu P'$.  The action orginates from some $P_i$ alone, say $P_i \arrN\mu P_i'$.
If $\mu$ is a $\tau $ action, then also $P_i \arr\mu P_i'$; 
and since $P_i \wbPI Q_i$, we have $Q_i \Longrightarrow  Q_i'$ with $P_i' \wbPI Q_i$.  
Moreover, also $Q_i \RaN  Q_i'$ --- and the corresponding transition from $Q$.

Using Lemma~\ref{l:opt_op} we infer that, up-to expansion, the
derivative processes $P_i'$ and $Q'_i$ can be rewritten into processes
that fit the definition of $\R$. 

For the case when $\mu$ is an output, one reasons similarly, possibly
also using Lemma~\ref{l:deco} when $\mu $ is of the form $\bout x {z,q}$.  

Finally, if $P \arrN \mu P'$ and $\mu$ is an input, then $P $ is input
reactive. From the conditions in the definition of $\R$ we infer that
$Q$ is input reactive too.  Then we can conclude, reasoning as in the
previous cases, but this time  using Lemma~\ref{l:opt_con} with
$\gamma   $ a singleton.  
\end{proof}  
\noindent
Theorem~\ref{t:wbPI-wbOPI}  is then a corollary of
Theorem~\ref{t:opBIS_BIS}.



\section{Auxiliary results for Section~\ref{s:encoGames}}
\label{a:oc}

We present results that are needed to establish the operational correspondence between 
OGS and \piI, studied in Section~\ref{s:encoGames}.

We begin discussing an optimisation
$\qencV$ of
the  encoding of 
 call-by-value $\lambda$-calculus.
Following 
 Durier et al.'s 
\cite{DurierHS18}, 
we sometimes use $\qencV$  to simplify proofs. 
We report the 
full definition of   $\qencV$;
for  convenience, we also list the 
 definition of $\qencV$  on 
 OGS environments and configurations, thought  is the same as 
that for the initial encoding $\cbvSymb$~--- just replacing $\cbvSymb$ with 
 $\qencV$ in Figure~\ref{f:OGS_pi_cbv}).

The encoding 
 $\qencV$
 is 
obtained  from the  initial one 
$\cbvSymb$ 
by inlining the encoding and performing a few (deterministic) reductions,
at the price of 
 a more complex definition.  Precisely, in the encoding of application 
 some of the initial
communications are removed, including those with which a term signals to have become a value. 
To achieve this,  the encoding of an 
application  goes by a case analysis
on the occurrences of values in the subterms.


\begin{figure*}
\Mybar 
{\small
Encoding of  $\lambda$-terms:

\[
\setlength{\arraycolsep}{3pt}%
\begin{array}{rcl@{\;}l}
\equaDS{}
{  \encoV {x\val}} \defi{ \bind p \outb x {z,q}.(\encoVVa \val z|\fwd q p)}
\\
\equaDS{}
{ \encoV {(\lambda  x.  M)\val}} \defi{ \bind p \new {y,w} (\encoVVa {\lambda  x. M} y | \encoVVa {\val} w
  | \outb y {w',r'}.(\fwd {w'} w|\fwd {r'} p))} 
\\
\equaDS{{(*)}}
{ \encoV {\val M}} \defi{ \bind p \new y (\encoVVa \val y | \new r (\encoVa M r | \inp r w
  .\outb y {w',r'}.(\fwd {w'} w|\fwd {r'} p)))}
\\
\equaDS{{(*)}}
{ \encoV {M\val} } \defi {\bind p \new q (\encoVa M q | \inp q y. \new w (\encoVVa \val w |
  \outb y {w',r'}.(\fwd {w'} w|\fwd {r'} p)))}
\\
\equaDS{{(**)}}
{ \encoV {MN}} \defi {\bind p \new q (\encoVa M q | \inp q y. \new r (\encoVa N r | \inp r w
  .\outb y {w',r'}.(\fwd {w'} w|\fwd {r'} p)))}
\\
\equaDS{}
{ \encoV {\val}}\defi{ \bind p \outb p y . \encoVVa\val y}
\end{array} 
  \]
where
in the  rules marked $(*)$, $M$ is  not   a value,  and  
in the rule marked $(**)$   $M$ and $N$  are  not values;
and where 
 $ \encoVV\val$ is thus defined:
\[
\begin{array}{rcl}
\encoVV {\lambda  x. M}\defi{ \bind y  !\inp y {x,q}.\encoVa M q}\\
\encoVV x \defi {\bind y \fwd y x}
\end{array}
\]
$ $ \dotfill $ $

Encoding of environments: 
\[ 
\begin{array}{rcl}
\encoEnvV{\pmap{y}{V}\cdot\gamma'  } & \defi & \encoVVa {V} y | \encoEnvV{\gamma'  }
\\[\mypt] 
\encoEnvV{\pmap{q}{(E,p)}\cdot\gamma'  } & \defi & 
\inp qx .\encoVa {E[x]} p  | \encoEnvV{\gamma'  } \\[\mypt]
\encoEnvV{\emptymap  } & \defi & \nil
\end{array} 
\]
$ $ \dotfill $ $ 

Encoding of configurations:

\[ 
\begin{array}{rcl}
\encoConV{ \conf{M,p,\gamma}}
& \defi& 
 \pairP {
 \encoVa {M} p | \encoEnvV \gamma }
\\[\mypt]
 \encoConV{ \conf{\gamma}} & \defi&  
 \pairP {
\encoEnvV \gamma}
\\[\mypt]
 \encoConV{ \conf{M }}
& \defi&  \pairP { \encoV   M }
\end{array}
 \] 
}
\caption{The optimised  encoding}
\label{f:opt_encod}
\Mybar 
\end{figure*}

\DSOCTb

We begin with  a few  results from \cite{DurierHS18,Durier20}
that are needed to reason about the  
optimised encoding $\qencV$, in addition to Lemma~\ref{l:opt_sound}
reported in the main text. 

\DSOCTe

\begin{lemma} 
\label{l:opt_stuck}
We have: 
$$
\begin{array}{rcl}
\encoVa {\evctxt[x\val]} p\exn \outb x {z,q}.(\encoVVa \val z|\inp q
y.\encoVa {\evctxt[y]}p).
\end{array}
 $$ 
\end{lemma}
\noindent
The proof \cite{DurierHS18,Durier20}
goes by induction on the evaluation context \evctxt.

The following lemma \cite{DurierHS18,Durier20}
uses Lemma~\ref{l:opt_stuck}  
to establish the shape of the possible transitions that a term 
$\encoVa M p $ can perform.

\begin{lemma}
\label{l:opt_op}
For any
  $M\in\Lao$ and  $p$, process  $\encoVa M p$ has exactly
  one immediate transition and exactly one of the following clauses holds:
\begin{enumerate}
\item $\encoVa M p\arr{\outb py}P$  
and $M$ is a value, with $P=\encoVVa M y$; 
\item $\encoVa M p\arr{\outb x {z,q}} P$  and 
 $M $ is of the form $\evctxt[x\val]$, for some $E,x,\val$,  with 
$$P\exn \encoVVa \val z|\inp q y.\encoVa {\evctxt[y]} p ,$$
and moreover $z$ is not free in $\inp q y.\encoVa {\evctxt[y]} p$
whereas $q$  is not free in $\encoVVa \val z$;
\item $\encoVa M p\arr\tau P$ and there is $N$  with $M\red N$ and
  $P\exn \encoVa N p$. 
\end{enumerate}
\end{lemma}
 
\noindent
We now move to reasoning about the behaviour of the encoding of configurations. 
Two key lemmas are the following ones; they are derived from Lemmas~\ref{l:opt_stuck} and
\ref{l:opt_op}.

\begin{lemma} 
\label{l:opt_con}
If  $\encoConV{\gamma } \arr \mu P$ and $\mu$ is an input action, then we have three possibilities:
\begin{enumerate}

\item $\mu $ is an input $y(x,q)$ and $\gamma 
= \pmap{y}{\lambda x . M}\cdot\gamma' $, and $P \contr 
\encoVa {M} q  | \encoEnvV{\gamma  } $; 

\item $\mu $ is an input $y(x,q)$ and $\gamma 
= \pmap{y}{z}\cdot\gamma' $, and $P \contr 
\encoVa {zx} q  | \encoEnvV{\gamma  } $;

\item $\mu $ is an input $q(x)$ and $\gamma 
= \pmap{q}{(E,p)}\cdot\gamma' $, and $P\contr 
\encoVa {E[x]} p  | \encoEnvV{\gamma'  } $. 
\end{enumerate} 
\end{lemma}
 

\begin{lemma}
\label{l:VAL}
  $\encoVVa {\val} y  \arr {\inp y {x,q}} \wbPI 
\encoVa {\val x } q | \encoVVa {\val} y$
\end{lemma}              
\noindent
In Lemma~\ref{l:VAL}, the occurrence of $ \wbPI$ cannot be replaced by $\contr$; 
for instance, if $V = \lambda x. M$ then, using Lemma~\ref{l:opt_op}, we can infer that 
the derivative process is in the relation $\contr$ with 
$\encoVa {M } q$; however it is not in the same relation with 
$\encoVa {\val x } q$ (which has extra $\tau$ transitions with respect to $\encoVa {M } q$).

\iflong
\finish{below i think it is actually not needed now} 
Here is another important result that we borrow from \cite{DurierHS18}.
It shows  that the only $\lambda$-terms
 whose encoding is bisimilar to 
$\enca x$ reduce either to $x$, or to a (possibly infinite) 
$\eta$-expansion of $x$.

\begin{lemma}\label{l:opt_eta}
  If $\val$ is a value and $x$ a variable,
  $\encoVV \val\wbPI \encoVV x$ 
  implies that either $\val=x$ or $\val = \lambda   z. {M}$, where the eager
  normal form of $M$ is of the form $\evctxt[x\valp]$, with
  $\encoVV {\valp}\wbPI \encoVV z$ and $\encoV {\evctxt[y]}\wbPI\encoV y$
  for any $y$ fresh.
\end{lemma}

\fi


\noindent
We can now establish the operational correspondence between OGS and \piI. 
(In all the results about operational correspondence, we exploit the convention on
freshness of bound names of actions and traces produced by OGS configurations and \piI
processes, as by Remark~\ref{r:bn}.)
\begin{lemma}[strong transitions, from \AOGS to \piI]
\label{l:opcorrCON_str}
$ $ \begin{enumerate}

\item If $\FF \longrightarrowA \FF'$, then 
$ \encoConV \FF \longrightarrowPI  {\contr} \: \encoConV {\FF'}$;

\item If $\FF \arrA\ell \FF'$, then
$ \encoConV \FF \arrPI\ell  
 {\wbPI} \:
\encoConV {\FF'}$.
\end{enumerate} 
\end{lemma} 

\begin{proof}
The proof goes by a case analysis on the rule used to derive the transition from 
$\FF $, using the definition of the encoding, using Lemmas~\ref{l:opt_op}-\ref{l:VAL}.
\end{proof}

\begin{lemma}[strong transitions,  from  \piI to \AOGS]
\label{l:opcorrCON2}
$ $
\begin{enumerate}
\item If 
$ \encoConV \FF \longrightarrowPI \PP$ then there is $\FF'$ such that 
$\FF \longrightarrowA \FF'$  and  
$\PP   \contr \encoConV {\FF'}$; 

\item If
$ \encoConV \FF \arrPI\mu \PP$
and $\mu$ is an output, 
then there is $\FF'$ such that
$\FF \arrA\mu \FF'$
and $\PP \contr \encoConV {\FF'}$;

\item If  $\FF$ is passive and 
$ \encoConV \FF \arrPI{\mu} \PP$, 
then 
 there is $\FF'$ such that 
$\FF \arrA{\mu} \FF'$
and $\PP  \wbPI \encoConV {\FF'}$.
\end{enumerate} 
\end{lemma} 

\begin{proof}
The proof goes by a case analysis on the possible transition from 
$ \encoConV \FF$, again 
exploiting Lemmas~\ref{l:opt_op}-\ref{l:VAL}.
\end{proof} 
\noindent
We can strengthen both  
Lemma~\ref{l:opcorrCON_str} and 
Theorem~\ref{t:opcorrCON_w}
 to a correspondence between transitions from
\AOGS configurations and from their \piI translation under the (more
restrictive) \opLTS.  
Further, as discussed in Remark~\ref{r:expa_op_ord} and
Appendix~\ref{a:behav}, we can use the expansion relation and
bisimilarity on the
\opLTS, written $\expaOP$ and $\wbOPI$. 
We write \piIop to refer to \piI under the \opLTS.

\begin{theorem}[%
strong transitions, from \AOGS to \piIop]
\label{t:opcorrCON_str_OP}
$ $ \begin{enumerate}

\item If $\FF \longrightarrowA \FF'$, then 
$ \encoConV \FF \raN  \; \: \contrOP  \encoConV {\FF'}$;

\item If $\FF \arrA\ell \FF'$, then
$ \encoConV \FF \arrN\ell  
\; \: \wbOPI 
\encoConV {\FF'}$.
\end{enumerate} 
\end{theorem}

\begin{theorem} 
[weak transitions, from \AOGS to \piIop]
\label{t:opcorrCON_w_OP}
$ $ \begin{enumerate}

\item If $\FF \LongrightarrowA \FF'$, then 
$\encoConI \FF \RaN  \; \contrOP \encoConI {\FF'}$;

\item If $\FF \ArrA\ell \FF'$, then
$ \encoConI \FF \ArrN\ell  
\; \wbOPI
\encoConI {\FF'}$.
\end{enumerate} 
\end{theorem} 
\noindent
The following Theorem~\ref{t:opcorrCON3}  is needed in the proof of Theorem~\ref{t:opcorrCON3_w}:
it relates strong transitions from  \piIop processes to
transitions of \AOGS configurations. 
Again, in both theorems, the occurrences of $\expa$ and $\wbPI$ can be
replaced by $\expaOP$ and $\wbOPI$. 

\begin{theorem} 
\label{t:opcorrCON3}
$ $ \begin{enumerate}

\item If 
$ \encoConV \FF \arrN \tau \PP$ then there is $\FF'$ such that 
$\FF \longrightarrowA \FF'$  and  
$\PP   \contr \encoConV {\FF'}$;

\item If
$ \encoConV \FF \arrN\ell \PP$
 then there is $\FF'$ such that
 $\FF \arrA\ell \FF'$
and $\PP \wbPI \encoConV {\FF'}$.
\end{enumerate} 
\end{theorem} 
\noindent
We report now more details on the proof of Corollary~\ref{c:traces}.
First, we recall the assertion:

\begin{center}
\begin{tabular}{l}
For any trace $s$, we have 
  $\FF \ArrA s$   iff 
 $ \encoEnvI \FF  \ArrN s$.
\end{tabular} 
\end{center}

\begin{proof}
The result is first established for 
 $\qencV$, and then extended to 
$\cbvSymb$,  exploiting the correctness of the optimisations 
in  $\qencV$ (Lemma~\ref{l:opt_sound}; again, in this lemma, the occurrence of $\expa$
can be lifted to $\expaOP$).

For  $\qencV$, 
both directions are proved proceeding by induction on the length of a trace 
 $\ell_1, \ldots,   \ell_n $.  
In the direction from left to right, we use 
 Theorem~\ref{t:opcorrCON_w_OP}(2). 
For the converse direction, we 
proceed similarly, this time relying on 
Theorem~\ref{t:opcorrCON3_w}(2) 
 (again, the version with $\wbOPI$ in place of 
 $\wbPI$).
 \end{proof}

\section{Auxiliary results for Section~\ref{ss:compa_cogs_piI}}
\label{a:aux_tr_bisi_sing}

We report here the auxiliary results needed to establish the 
relationship between Concurrent OGS (\COGS) and \piI. 
%
%
%
%
%
%
%
%
%
%
%
%
%
 %
%
%
 %
The results below are used in the proof of Lemma~\ref{l:tr_bisi_sing}.

A process $\PP$  \emph{has deterministic immediate transitions} if, for any $\act$, 
whenever $\PP \arrPI\act \PP_1 $ and   $\PP \arrPI\act \PP_2$, then $\PP_1 = \PP_2$. 

\begin{lemma}
\label{l:immediate_tr}
Suppose $\PP,\QQ$ have  deterministic immediate transitions, $\PP \TEPI \QQ$, 
 and $\PP \arrPI\act \PP_1$. Then 
$\QQ \ArrPI \act \QQ_1$ implies $\PP _1 \TEPI \QQ_1$. 
\end{lemma} 
\noindent
Lemma below is the analogue to  Lemma~\ref{l:deco} for traces.

\begin{lemma}
\label{l:par_tr}
$ $ 
\begin{itemize}
\item
Suppose $P = P_1 | ! \inp a {\tilb}.P_2$, 
and $Q = Q_1 | ! \inp a {\tilb}.Q_2$, 
and  $\pairP P \TEPI \pairP Q$.
Then also 
  $\pairP {P_1}  \TEPI \pairP {Q_1}$.

\item 
Similarly, suppose $P = P_1 |  \inp a {\tilb}.P_2$, 
and $Q = Q_1 |  \inp a {\tilb}.Q_2$, 
and  $\pairP {P} \TEPI \pairP {Q}$.
Then also 
  $\pairP {P_1}  \TEPI \pairP {Q_1}$.
\end{itemize}
 \end{lemma} 
\noindent
We can now conclude the proof of 
Lemma~\ref{l:tr_bisi_sing}.
We recall the assertion:
\begin{center}
\begin{tabular}{l}
For any $M,N$ we have:  \\
$
{\encoEnvI M} \TEPI {\encoEnvI N}
$ iff 
${\encoEnvI M} \wbPI {\encoEnvI N}$.
\end{tabular} 
\end{center} 

\begin{proof}
We have to prove that trace equivalence
implies bisimilarity. 
We work with the optimised encoding  $\qencV$.
The relation 
\[ 
(\encoConV \FF,  \encoConV \GG) \st \mbox{ $\FF,\GG$ are singleton
  with $\FF \TEp 
\GG$}
\]
is a bisimulation up-to context and up-to $(\contr,
 \wbPI)$.

 $\FF$ and $\GG$, as singleton, are of the  forms 
described in Lemma~\ref{l:shapesT}.
Moreover, to be in the relation $ \TEp$ they must be of the same
form, and, using also 
  Lemmas~\ref{l:opt_op}  and 
\ref{l:opt_con} we deduce that they have deterministic immediate
transitions. 

Suppose  $\FF$ and $\GG$ are of the form 
$\encoIa {M} p$ and 
$\encoIa {N} q$, respectively.

Consider the case 
$\encoIa {M} p \arr\tau P$; then by Lemma~\ref{l:opt_op}, 
$P \contr \encoIa {M'} p$, for some $M'$. 
Moreover, we have $ \encoIa {M'} p \TEp \encoIa {M} p \TEp  \encoIa
{N} q$, and we are done.

The case when $\encoIa {M} p$ performs an output action $\bout px$ is
simple, as usual using  Lemma~\ref{l:opt_op}; we also deduce that
$p=q$. 

Suppose now $\encoIa {M} p$ performs an output action $\bout x{z,p'}$.
Using  Lemma~\ref{l:opt_op}, we have 
\[\encoIa {M} p\arr{\bout x{z,p'}}  
\exn \encoVVa \val z|\inp {p'} y.\encoVa {\evctxt[y]} p
\]
for some $\val,y,E$.
Since $\encoIa {M}p \TEp \encoIa {N}q$, and again using 
 Lemma~\ref{l:opt_op},
we deduce 
\[\encoIa {N} q\Longrightarrow \arr{\bout x{z,p'}}  
\exn \encoVVa W z|\inp {q'} y.\encoVa {E'[y]} q
\]
By Lemma~\ref{l:immediate_tr}, 
\[
\begin{array}{rc}
\encoVVa \val z|\inp {p'} y.\encoVa {\evctxt[y]} p
& \TEp \\[3pt]
\encoVVa W z|\inp {q'} y.\encoVa {E'[y]} q
\end{array}
 \] 
Applying Lemma~\ref{l:par_tr} twice, we deduce 
\[
\encoVVa \val z 
\TEp 
\encoVVa W z
 \]
 and
\[
\inp {q'} y.\encoVa {E'[y]} q \TEp
\inp {q'} y.\encoVa {E'[y]} q
\]
and we are done, up to expansion and context.

\noindent The cases when
$\FF$ and $\GG$ are of a different form are similar, this time using 
Lemma~\ref{l:opt_con}.
\end{proof}

\section{Proofs about Section~\ref{s:cbnNEW}}
\label{a:cbn}
             
\subsection{Some auxilary results}


For $n>0$ we  define, for \piI\ abstractions  $A_1 \ldots A_n$:
\[\begin{array}{l}
\OUTU n {r_n}{r_0} {A_1 \ldots A_n}{} \defi 
\begin{array}[t]{l}
\res{r_1, \ldots r_{n-1}} \\
\! \! (\inpi r 0 {v_1} . \bopi v 1{x_{1},r_{1}}.    
 | \\
\inpi r 1 {v_2} . \bopi v 2{x_{2},r_{2}}    .  
 | \\
\ldots  \\
\inpi r {n-1} {v_n} . \bopi v n{x_{n},r_{n}}    .  
 \\
  !  \inp {x_1} {q} . {\app{A_1}{q} } | \\
  !  \inp {x_2} {q} . {\app{A_2}{q} } | \\
\ldots  \\
  !  \inp {x_n} {q} . {\app{A_n}{q} } )
  \end{array} 
  \end{array} 
\]

\begin{lemma} 
We have: 
\[
\begin{array}[t]{rcl}
\encoNa{ x M_1 \ldots M_n} {r_n}
\arrPI{\bout x{r_0}}
&\equiv  &
\begin{array}[t]{l}
\res{r_1, \ldots r_{n-1}} \\
\! \! (\inpi r 0 {v_1} . \bopi v 1{x_{1},r'_{1}}    .  \link {r'_{1}}{r_{1}}
 | \\
\inpi r 1 {v_2} . \bopi v 2{x_{2},r'_{2}}    .  \link {r'_{2}}{r_{2}}
 | \\
\ldots  \\
\inpi r {n-1} {v_n} . \bopi v n{x_{n},r'_{n}}    .  \link
{r'_{n}}{r_{n}} \\
  !  \inp {x_1} {q} . {\encoNa{M_1} {q} } | \\
  !  \inp {x_2} {q} . {\encoNa{M_2}{q} } | \\
\ldots  \\
  !  \inp {x_n} {q} . {\encoNa{M_n}{q} } )
\end{array}
\\ \ \\
&\wbPI &
\OUTU n {r_n}{r_0} {\encoN{M_1} \ldots \encoN{M_n}}{} 
\end{array}
\]
 \end{lemma}
\noindent
The two lemmas below are known results about encoding of call-by-name $\lambda$-calculus
into \piI\  \cite{San95i}.

\begin{lemma} 
\label{l:sk}
 Let $M $ be a  $\lambda$-term.
$ $ 
\begin{enumerate}
\item  If  $a,b$ and $c$  are distinct names of the same type,  then 
  $\res b ( \link ab | \link bc ) \wb \link ac$;

\item  If $x$ and $y$ are distinct  trigger  names and $y$ is not free in
$M$, then
  $\res{ x} (\link  xy  | \encoNa M r ) \wb \encoNa{ M
\sub y x} p$;

\item  If $p$ and $r$ are distinct  continuation names, then 
 $\res{ r} (\link  r p | \encoNa M r ) \wb \encoNa  M p$.

\end{enumerate} 
\end{lemma}

\begin{lemma} 
\label{l:beta_cbn}
If $M \longrightarrow M'$ then $\pairP {\encoNa M p } \longrightarrow \: {\contr}\; \pairP
{\encoNa        {M'} p }$
\end{lemma}

\subsection{Operational correspondence}
\label{ss:oc_cbn}

We report here some results concerning the operational correspondence between 
OGS  and \piI terms.

\begin{lemma}[Actions from encodings of 
$\LaoP$ terms]
\label{l:opt_op_cbn}
For any
  $\LL\in\LaoP$  and $\phi $ with 
 $p,\fv \LL \subseteq \phi$,  term  $\encoSPN L p$ has exactly
  one immediate transition, and exactly one of the following clauses holds:
\begin{enumerate}
\item $\encoSPN \LL p\arr{\bout pv} P $  
and $\LL$ is a value (hence of the form either  $ \lambda x. M $ or
$v$ ),  
 with $P =\encPPN { \pmap{v}{L}}{\phi \uplus \{v\}}$;

\item  
 $\encoSPN \LL p\arr{\bout x {q}}P$  
and $L =  
{E [x]}$, for some $E$, and 
$$P \contr  
  \encPPN { \pmap q{(E,p)}} {\phi \uplus \{q\}}
$$;

\item 
 $\encoSPN \LL p\arr{\bout v {x,r}}P$  
and $L =  
{E [v M]}$, for some $E,M$ 
and 
$$P \equiv 
  \encPPN { \pmap r{(E,p)}\cdot  \pmap x M } {\phi \uplus \{x,r\}}
$$;

\item 
$\encoSPN M p\arr\tau P$ and there is $N$  with $M\red N$ and
  $$P\exn \encoSPN  N p.$$ 
\end{enumerate}
\end{lemma}

\ifcutCBNLMCS 
\begin{proof}
\dsJ{I have checked (and thus rectified the def of $ \pmap q{(E,p)}$ in the encoding )} 
\end{proof}
\fi

\begin{lemma}[Actions from \piI encodings of  configurations]
\label{l:opt_con_cbn}
If  $\encoN{\conf {\gamma, \phi }} \arrPI \mu P$,
 then we have three possibilities:
\begin{enumerate}

\item $\mu $ is an input $v(x,q)$ and $\gamma 
= \pmap{v}{V}\cdot\gamma' $, and 
$$P \wb 
\encoN {
\conf{ V x,  q, \gamma', \phi \uplus \{ q,x\}}
} ;$$

\item $\mu $ is an input $x(q)$ and $\gamma 
= \pmap{x}{M}\cdot\gamma' $, and 
$$P  =
\encoN {
\conf{ M,  q, \gamma, \phi \uplus \{ q\}}
} ;$$

\item $\mu $ is an input $p(v)$ and $\gamma 
= \pmap{p}{E,q}\cdot\gamma' $, and 
$$P  = 
\encoN {
\conf{ E[v],  q, \gamma, \phi \uplus \{ v\}}
} .$$ 
\end{enumerate} 
\end{lemma}

\begin{lemma}[from OGS to \piI, strong transitions]
\label{l:opcorrCON_cbn}
$ $ \begin{enumerate}

\item \caserule{IOQ}  If $F \arr{\initName (p)} F'$ then 
$\encoN{F } \arr  {\initName (p)}  
 \encoN{F'}$ ;

\item  \caserule{Int} If $F \longrightarrow F'$ then 
$ \encoN F \longrightarrow  \contr \encoN {F'}$; 

\item  \caserule{PA} If $F \arr{\bout pv} F'$ then 
$ \encoN F \arr{\bout pv}  \equiv  \encoN {F'}$; 
\item  \caserule{PVQ} If $F \arr{\bout v {y,q}} F'$ then 
$ \encoN F \arr{\bout v {y,q}}   \equiv  \encoN {F'}$; 
\item  \caserule{PTQ} If $F \arr{\bout x {q}} F'$ then 
$ \encoN F \arr{\bout x {q}}  \contr  \encoN {F'}$; 
\item  \caserule{OA} If $F \arr{\inp  {q} v} F'$ then 
$ \encoN F\arr{\inp  {q} v}    \encoN {F'}$; 
\item  \caserule{OVQ} If $F \arr{\inp   v{y,q}} F'$ then 
$ \encoN F\arr{\inp   v{y,q}}  \wb  \encoN {F'}$; 
\item  \caserule{OTQ} If $F \arr{\inp   x{q}} F'$ then 
$ \encoN F\arr{\inp   x{q}}    \encoN {F'}$.
\end{enumerate}
\end{lemma} 







\begin{theorem}[from OGS to \piI, traces]
\label{t:tracesGSpi_cbn}
If $F$  has the trace $s$ then also $ \encoN F$ has the trace $s$.
\end{theorem} 

\begin{proof}
Straightforward induction on the length of the trace. 
\end{proof} 
\noindent
Now the converse direction.

\begin{lemma}[from \piI to OGS, strong transitions]
\label{l:opcorr_cbnshort_piI}
$ $ \begin{enumerate}

\item   
If $\encoN{F } \arr  {\initName (p)}  
P$, 
then there is $F'$ such that 
 $F \arr{\initName (p)} F'$
 and
 $P = 
 \encoN{F'}$ ; 

\item  If 
$ \encoN F \longrightarrow P$, then 
then there is $F'$ such that 
$F \longrightarrow F'$ and  
$ P  \contr \encoN {F'}$; 

\item  If
$ \encoN F \arr{\mu} P$, 
where $\mu$ is an output,
 then 
then there is $F'$ such that 
$  F \arr{\mu} F' $ and 
$P  \contr  \encoN  {F'}$; 

\item  If
$ \encoN F \arr{\mu} P$, 
where $\mu$ is an input,  and 
$ \encoN F$ cannot  perform an output or a $\tau$,  
then there is $F'$ such that 
$  F \arr{\mu} F' $ and 
$P  \wb  \encoN  {F'}$. 
\end{enumerate}
\end{lemma} 
\noindent
In the modified transition system $\arrN {\mu'}$  for \piI, where output and silent transitions have
priority over input transitions, we can then derive from the above results the  full
abtractions on traces, reported in Corollary~\ref{c:CBNtraces}. 

\subsection{Concurrent Operational Game Semantics}

We have reported in the main text the operational correspondence result for strong
transitions. We can thus derive the 
operational correspondence result for weak
transitions.
 For silent transition, we proceed by
induction on the number of strong transitions involved. 

\begin{theorem}[from \COGS to \piI, weak transitions]
\label{t:CBNpogs_short_weak}
$ $ \begin{enumerate}

\item  If $F \Longrightarrowp F'$ then 
$ \encoSN F \LongrightarrowPI \; \contr \encoSN {F'}$; 

\item  If $F \ArrC{\ell} F'$  then 
$ \encoSN F \ArrPI{\ell} \;  \wbPI  \encoSN  {F'}$.
\end{enumerate}
\end{theorem}

\begin{theorem}[from  \piI to \COGS, weak transitions]
\label{t:CBNpogs_short_converse_weak}
$ $ \begin{enumerate}

\item  If
$ \encoSN F \LongrightarrowPI  P$, then there is $F'$ s.t.\  
 $F \Longrightarrowp F'$ and 
$P 
\contr \encoSN {F'}$;

\item  If
$ \encoSN F \ArrPI{\ell}  P$, then there is $F'$ s.t.\  
 $F \ArrC{\ell} F'$ and 
$P 
\wbPI \encoSN {F'}$.
\end{enumerate}
\end{theorem}

\begin{corollary}
\label{c:CBNpogs_pi}
 $F$ and $ \encoSN F$ have the same set of traces. 
\end{corollary}  

\begin{proof}
By induction on the length of a trace performed by 
 $F$  or $ \encoSN F$: thus one shows that, for a trace $t$,
\begin{itemize}
\item  If $F \ArrC{t} F'$  then 
$ \encoSN F \ArrPI{t}\;  \wbPI  \encoSN  {F'}$;
\item  If
$ \encoSN F \ArrPI{t}  P$, then there is $F'$ s.t.\  
 $F \ArrC{t} F'$ and 
$P 
\wbPI \encoSN {F'}$.\qedhere
\end{itemize} 
\end{proof}

\begin{corollary}
\label{c:CBNpogs_pi_traces}
We have: 
 $\FF_1 \TEp \FF_2$ iff 
 $ \encoSN {\FF_1} \TEPI  \encoSN {\FF_2}$.  
\end{corollary}

\begin{corollary}
\label{c:CBNpogs_pi_bis}
we have: 
 $\FF_1 \wbC \FF_2$ iff 
 $ \encoSN {\FF_1} \wbPI  \encoSN {\FF_2}$.  
\end{corollary}  

\ifcutCBNLMCS
\begin{proof}
In one direction, one shows that 
\[  \{ (    \encoSN {\FF_1},  \encoSN {\FF_2} ) \st \FF_1 \wbC \FF_2 \}\]
is a weak bisimulation up-to $\wb$. 
For the other direction, one shows that 
\[  \{ (    {\FF_1},  {\FF_2} ) \st \encoSN {\FF_1} \wbPI   \encoSN {\FF_2}\}\]
is a bisimulation.
\end{proof} 
\fi

\noindent
A configuration $F$ \emph{is a singleton} if its domain has only one
element.

\begin{lemma}
\label{l:CBNtr_bisi_sing}
For any singleton $\FF,\GG$, we have: $
{\encoSN \FF} \TEPI {\encoSN \GG}
$ iff 
${\encoSN \FF} \wb {\encoSN \GG}$.
\end{lemma} 

\begin{proof}
Weak bisimilarity implies trace equivalence, hence the direction from right to left. 

For the converse, we show
 that 
the relation 
\[ 
(\encoSN F,  \encoSN G) \st \mbox{ $F,H$ are singleton with $F \TEp
G$}
\]
is a bisimulation up-to contexts and expansion. 
The proof uses Lemmas~\ref{l:immediate_tr} and \ref{l:deco}. 
\finish{why this holds: the operational correspondence resuts; the fact that the names in
  the domain of a configuration do not appear anywhere else in the configuration} 
\end{proof} 
\noindent
We think that the above lemma can be generalised to aribitrary configurations; the current
statement is however sufficient for our purposes.
By combining Lemma~\ref{l:CBNtr_bisi_sing} and 
Corollaries~\ref{c:CBNpogs_pi_traces} and 
\ref{c:CBNpogs_pi_bis} we then  derive  Theorem~\ref{t:CBNfa_all}. 

\noindent
The following result holds for call-by-name but not for call-by-value.
\finish{above: check again...  is it so, really? } 
\finish{the result can be generalised to arbitrary configurations. It needs a bit more of
  work. Do we want to do it? } 
\begin{lemma}
\label{l:CBNonly}
Suppose $\gamma$ and $\delta$ (resp.\   
$\gamma'$ and $\delta'$) have the same domains.
Then 
 $\conf{\gamma \cdot \gamma',\phi } \TEp 
\conf{ \delta  \cdot \delta',\phi } $
implies
 $\conf{\gamma,\phi } \TEp 
\conf{ \delta  ,\phi } $
 \end{lemma}

\begin{proof}
This is a consequence of the fact that any action from the components $\gamma', \delta'$ 
has, as subject, a name in the domain of these components. Such a name does not appear in $
\gamma , \delta $, and thus it will not appear in any action of their traces.  
\end{proof} 

\ifcutCBNLMCS
\subsection{Up-to techniques for games}
\label{ss:CBNupto_games}
\finish{again, this section is essentially the same as fo call-by-value, with some
  notations revisited} 

The goal of this section is to trasnsfer up-to techniques from 
\piI  onto GS. We will then use such techniques to prove that
\COGS and ordinary OGS yield the same semantics on
$\lambda$-terms.

A relation $\R$ on configuration is \emph{well-formed} 
if  the domains of any two configurations 
$\FF_1$ and $ \FF_2$
with  $\FF_1 \RR \FF_2$  are the same.

Two configurations are \emph{disjoint} if the set of names that appear
in their flat  domain are disjoint. 

If 
$ \FF = \conf{ A , \gamma,\phi }$
and 
$ \GG = \conf{ B , \delta ,\phi' }$
are disjoint, then 
\[ \appendCon \FF \GG  \defi
 \conf{ A \cdot B , \gamma \cdot \delta,\phi\uplus \phi' }
 \] 
\finish{above: should $\phi'$ be $\phi$?} 

Below, all relations on configurations are meant to be well formed. 

Given a well-formed relation $\R$ we write
$\appendRel \R $ for the relation
\[ 
\{ (\FF_1,\FF_2) \st
\mbox{there is $\GG$ s.t.\ }
 \FF_i = \appendCon{\FF'_i}\GG   \mbox{ $i=1,2$}  \mbox{ and }  \FF'_1\RR \FF'_2\}
\]

Finally, $\uptoComp \R $ is the reflexive and transitive closure of 
$\appendRel \R $

Note: equality on environments ignores the order of their components
(ie, it is  seen as a set).

\begin{definition}
\label{d:uptoComp-cbn}
A relation on configurations is a \emph{bisimulation up-to
  composition} if whenever $\FF_1 \RR \FF_2$: 
\begin{itemize}
\item if $\FF_1 \arrC\mu \FF_1'  $ then there is $\FF_2' $ such that 
$\FF_2 \arrC\mu \FF_2'  $ and 
$\FF'_1 \uptoComp \R \FF'_2$

\item  the converse
\end{itemize}  
\end{definition} 

If $R$ is a relation on configuration, then 
$\encoSN \R$ is the relation on \piI terms obtained by mapping each pair
in $\R$ in the expected manner: 
\[ \encoSN \R \defi \{ (\encoSN {\FF_1}, \encoSN {\FF_2}) \st \FF_1 \RR \FF_2 
 \} \]

\begin{theorem}
\label{t:uptoComp-cbn}
If 
$\R$ is bisimulation up-to
  composition then $\R \subseteq {\wbC}$.
\end{theorem} 

\begin{proof}
By showing that 
$\encoSN \R$ is a bisimulation up-to expansion and contexts in 
\piI.
\finish{all is needed here is a result of operational correspondence} 
\finish{upto transitive closure is captured by having ``upto polyadic contexts''} 
\end{proof}

\subsection{Relationship between \COGS and OGS}

\finish{actually here $M$ ranges over  extended $\lambda$-terms} 

For any term $M$, its traces in ordinary OGS are a subset of the traces in 
Concurrent OGS (\COGS), precisely, the traces with an alternation between player and
opponent moves. 
\finish{correct?}   
In other words, the  syntax of a trace of $M$   reveals whether that trace is also a trace
for $M$ in ordinary OGS.  
As a consequence, trace equivalence in \COGS implies trace equivalence in OGS. 

\begin{lemma}
\label{l:CBNPOGS_OGS}
If $\FF_1 \TEp \FF_2  $ then also $\FF_1 \TEPI \FF_2$. 
\end{lemma}

For the opposite direction, we exploit the 
`bisimulation up-to
  composition' technique.

A configuration $\FF$ is a \emph{singleton} if the domain of $\FF$ has only 1 element. 
That is, $\FF$ is either of the form 
$\conf{\pmap{p}{M}, \emptyset, \phi }$, or 
$\conf{\pmap{v}{V}, \emptyset, \phi }$, or 
$\conf{\pmap{p}{(E,q)}, \emptyset, \phi}$, or
$\conf{\pmap{x}{M}, \emptyset, \phi}$.

\finish{whatch out for the initial term....}

In the proof, 
we write $\FF \subseteq \GG$ when, as partial maps, $\GG$ is an extension of $\FF$. 
\finish{and moreover there is inclusion on the $\phi$ components} 
As for call-by-value, we distinguish between transitions that a configuration can make in \COGS and \AOGS, 
writing $\arrC \mu  $ for the former, and $\arrA\mu$ for the latter. 
Similarly for weak transitions $\ArrC \mu  $ and $\ArrA \mu  $.

\begin{lemma}
\label{l:CBNOGS_COGS}
If $\FF_1, \FF_2  $ are singleton and 
$\FF_1 \wbA \FF_2  $, then also $\FF_1 \wbC \FF_2$. 
\end{lemma}  

\begin{proof}
Let $\R $ be the relation on configurations with $\FF \RR \GG$ if:
\begin{enumerate}
\item $\FF,\GG$ are singleton, and
\item  either 
\begin{enumerate}
\item
there are passive  $\FF', \GG'$ with $\FF \subseteq \FF'$ and $\GG \subseteq \GG'$ and $\FF' \wbA
  \GG'$,
\item   or 
there are $p,M,N,  \gamma , \delta $ with
$\FF = \conf{\pmap{p}{M},\phi}$,  
$\GG = \conf{\pmap{p}{N},\phi}$, and 
$\conf{\pmap{p}{M}, \gamma,\phi } \wbA 
\conf{\pmap{p}{N}, \delta,\phi  } $
\end{enumerate} 
\end{enumerate}
\finish{initial terms are hopefully special cases of 2b...} 
\finish{i think one has also to impose that related singleton have the same flat domains}

We show that $\R$ is  a bisimulation up-to
  composition.
We distinguish the cases in which (2.a) or (2.b) holds.

First suppose $\FF \RR \GG$ because (2.a) holds, and there are $\FF',\GG'$ as in (2.a).
Suppose 
\begin{equation}
\label{e:CBNKoaH}
 \FF \arrC  \mu \FF_1 
\end{equation} 
in \COGS.  The rule can be OA or OVQ, or OTQ.  We assume it is OA; the cases  for  OVQ and
OTQ are similar.  
We thus have, for some $p,q,E,v$ 
$$
\begin{array}{rcl}
 \FF &=&  \conf{ \gamma_1 ,\phi} \mbox{ with }  \gamma_1 =    \pmap{p}{(E,q)} \\
\mu & = & p(v) \\
\FF_1 & = & \conf{A, \gamma_1, \phi \uplus \{v\}} \mbox{ with }  A = \pmap{q}{E[v]}
\end{array} 
$$
Moreover,  we have $\GG = \conf{ \delta_1,\phi }$  with $  \delta_1 =    \pmap{p}{(\DD,q)}$, for
some $\DD$ (as related singleton configurations have the same flat domain), and, for
$\FF',\GG'$ as in (2.a), for some $\gamma, \delta $:  
$$
\begin{array}{rcl}
\FF' & = & \conf{ \gamma_1 \cdot \gamma  ,\phi} \\ 
\GG' & = & \conf{ \delta_1 \cdot \delta  ,\phi} 
\end{array} $$

We also have 
$ \FF' \arrA  \mu \conf{A,  \gamma_1 \cdot \gamma  ,\phi\uplus \{v\}} \defi \FF''$. 
Since $\FF' \wbA \GG'$, 
\[ \GG'  \arrA  \mu \conf{B,  \delta_1 \cdot \delta  ,\phi \uplus \{v\}}  \defi \GG''\] 
for 
$B \defi \pmap{q}{\DD[v]}$, 
with $\FF'' \wbA \GG''$.

Let $\GG_1 =   \conf{B, \delta_1}$.
We also have 
\[ \GG \arrC  \mu \GG_1 \] 
This is sufficient to match \reff{e:CBNKoaH}, as 
$\FF_1 \uptoComp \R  \GG_1$: indeed we have  
both $\conf{ \gamma_1,\phi\uplus \{v\}} \R  \conf{ \delta_1,\phi\uplus \{v\}}$ and  
(using clause (2.b))
$\conf{ A, \emptyset , \phi\uplus \{v\}} \R  \conf{ B, \emptyset, \phi\uplus \{v\}}$.

\finish{above : i am implicitly using several times a result of the kind 
$\conf{..., \phi} \wbA \conf{..., \phi}$ implies 
$\conf{..., \phi'} \wbA  \conf{..., \phi'}$ where $\phi'$ is larger than $\phi$ 
}

Now the case (2.b);  let $p,  \gamma , \delta $ as in (2.b).
 There are 4 possibilities of transitions, corresponding to rules 
Int, PA, PVQ, and PTQ. The case of 
Int  is straightforward (also note that bisimilarity is preserved by internal
moves).  We only look at PVQ, as the others are simpler. 
Thus we have 
$\FF = \conf{\pmap{p}{E[v M_1]},\emptyset ,\phi }$, for some $E,v,M_1$,
and
$\GG = \conf{\pmap{p}{N]},\emptyset ,\phi }$
 (here again we use the property that  the flat domains of $\FF$ and
$\GG$ are the same) , 
 and then 
\begin{equation}
\label{e:CBNKpqH}
\FF  \arrC{\questPV{v}{y}{q}}
     \conf{ \pmap{y}{M_1}\cdot \pmap{q}{(E,p)},\emptyset , \phi\uplus\{y,q\}}  \defi \FF_1
 \end{equation} 
 We also have  
\[
\FF' \defi \conf{\pmap{p}{E[v M_1]}, \gamma,\phi }  
 \arrA{\questPV{v}{y}{q}}
     \conf{ \pmap{y}{M_1}\cdot \pmap{q}{(E,p)},\gamma, \phi\uplus\{y,q\}} 
\defi \FF''\]
Since 
$\FF' 
\wbA 
\conf{\pmap{p}{N}, \delta,\phi  } \defi \GG' $ 
we must have $N \LongrightarrowA  \DD[v N_1]$, for some $\DD,N_1$ so that 
we  can derive
\[
\GG' \defi \conf{\pmap{p}{\DD[v N_1]}, \delta,\phi }  
 \ArrA{\questPV{v}{y}{q}}
     \conf{ \pmap{y}{N_1}\cdot \pmap{q}{(\DD,p)},\delta, \phi\uplus\{y,q\}} 
\defi \GG''\]
with 
$\FF'' \wbA \GG''$. 

Let $$
\GG_1 \defi      \conf{ \pmap{y}{N_1}\cdot \pmap{q}{(\DD,p)},\emptyset , \phi\uplus\{y,q\}} 
.$$
We also have $\GG \ArrC
{\questPV{v}{y}{q}}
 \GG_1$. 

From $\FF'' \wbA \GG''$, appealing to (2.a)  we deduce that both 
$     \conf{ \pmap{y}{M_1},\emptyset , \phi\uplus\{y,q\}} \RR
     \conf{ \pmap{y}{N_1},\emptyset , \phi\uplus\{y,q\}}$ 
and 
$
     \conf{ \pmap{q}{(\DD,p),\emptyset , \phi\uplus\{y,q\}} 
 \RR
\pmap{q}{(E,p)},\emptyset , \phi\uplus\{y,q\}}$ hold.
Hence $\FF_1 \uptoComp \R  \GG_1$.

In summary, as an answer to the challenge 
\reff{e:CBNKpqH}, we have found  $\GG_1$ such that 
 $\GG \ArrC
{\questPV{x}{y}{q}}
 \GG_1$ and  
$\FF  _1 \uptoComp \R  \GG_1$; this closes the proof.
  \end{proof} 

\dsJ{note proof : done  assuming permanent environments, check again with ephemeral ones}

We can thus derive 
Corollary~\ref{c:CBNOGS_COGS}, reported in the main text, and  stating that all forms of
trace equivalence and bisimilarity, on GS or \piI\ semantics, coincide.   
%
We remark that such a  result  cannot be extended to arbitrary configurations.
For instance, 
we have 
\[ \conf {\pmap  p \Omega ,  \pmap  x M , \phi} 
\TEA
\conf {\pmap  p \Omega ,  \pmap  x N , \phi} 
\]
for any $M,N$, which need not be the case for 
$\TEp$, as the visible behaviours of $M$ and $N$ may differ.

However the result does  hold for passive configurations, or for configurations reachable from
$\lambda$-terms.

\finish{intuitively, why it works: the operational correpondence with pi; the fact that all new entries in the environment are on new names and the first action will be on such names. These properties hold both for CBN and for CBV} 

\fi

\end{document}